\documentclass[format=acmsmall, review=false, nonacm]{acmart}
\usepackage{acm-ec-25-proc}
\usepackage{booktabs} 
\usepackage[ruled]{algorithm2e} 
\usepackage{natbib} 

\SetAlFnt{\small}
\SetAlCapFnt{\small}
\SetAlCapNameFnt{\small}
\SetAlCapHSkip{0pt}
\IncMargin{-\parindent}

\setcitestyle{authoryear}

\usepackage{hyperref}
\usepackage{url}
\usepackage{subcaption}

\usepackage{outlines}
\usepackage{enumitem}

\usepackage{tikz}
\usepackage{appendix}

\usepackage{xcolor}
\usepackage{amsmath}
\usepackage{amsfonts}
\usepackage{amsthm}
\usepackage{thmtools, thm-restate}

\usepackage{color-edits}

\usepackage{multirow}

\usepackage{cleveref}

\addauthor[Hoda]{hh}{red}
\addauthor[Solon]{sib}{green}
\addauthor[Karen]{kl}{purple}
\addauthor[Jon]{jk}{blue}

\usepackage{thmtools, thm-restate}

\declaretheorem[name=Lemma]{lem}
\declaretheorem[name=Result]{result}

\definecolor{orange}{rgb}{0.8, 0.3,0}
\definecolor{red}{rgb}{0.7, 0,0}
\newcount\Comments
\newcommand{\comm}[2]{\ifnum\Comments=1\textcolor{#1}{#2}\fi}

\newcount\arXiv
\newcommand{\varXiv}[2]{\ifnum\arXiv=1#1\else#2\fi}

\newcommand{\Var}{\sigma^2}

\newcommand{\idxset}[1]{I_#1}

\newcommand{\loss}{\ell}

\newcommand{\E}{\mathbb{E}}

\newcommand{\hcats}{\mathcal{H}}
\newcommand{\mcats}{\mathcal{M}}
\newcommand{\dcats}{\mathcal{D}}

\newcommand{\kept}{\mathcal{R}}

\newcommand{\x}{\mathbf{x}}
\newcommand{\y}{\mathbf{y}}
\newcommand{\z}{\mathbf{z}}

\newcommand{\hcat}{C}
\newcommand{\mcat}{K}

\newcommand{\hcatvec}[1]{\mathbf{#1}_H}
\newcommand{\mcatvec}[1]{\mathbf{#1}_M}
\newcommand{\hcati}[1]{\hcat_{#1}}
\newcommand{\mcati}[1]{\mcat_{#1}}

\newcommand{\hloss}{\loss_H}
\newcommand{\mloss}{\loss_M}

\newcommand{\modvar}{\Delta^2}
\newcommand{\Prob}{P}

\newcommand{\fn}{f}
\newcommand{\fnopt}{\fn^*}
\newcommand{\mfn}{f_{M}}
\newcommand{\mfnobliv}{\mfn^\text{obliv}}
\newcommand{\mfniter}{\mfn^{\textnormal{iter}}}
\newcommand{\mfnopt}{\fnopt_{M}}
\newcommand{\hfn}{f_{H}}
\newcommand{\hfnopt}{\fnopt_{H}}

\newcommand{\fx}{a}
\newcommand{\fy}{b}

\DeclareMathOperator*{\argmin}{arg\,min}

\setcitestyle{authoryear}

\title[Designing Algorithmic Delegates]{Designing Algorithmic Delegates:\\The Role of Indistinguishability in Human-AI Handoff}

\acmYear{2025}\copyrightyear{2025}
\acmConference[EC '25]{The 25th ACM Conference on Economics and Computation}{July 7--10, 2025}{Stanford, CA, USA}
\acmBooktitle{The 25th ACM Conference on Economics and Computation (EC '25), July 7--10, 2025, Stanford, CA, USA}
\acmDOI{*******10.1145/3736252.3742520****}
\acmISBN{979-8-4007-1943-1/25/07}
\acmYear{2025}\copyrightyear{2025}
\setcopyright{acmlicensed}

\author{Sophie Greenwood}
\email{sjgreenwood@cs.cornell.edu}
\affiliation{
  \institution{Cornell University}
  \city{Ithaca}
  \state{NY}
  \country{USA}
}
\author{Karen Levy}
\email{karen.levy@cornell.edu}
\affiliation{
  \institution{Cornell University}
  \city{Ithaca}
  \state{NY}
  \country{USA}
}
\author{Solon Barocas}
\email{solon@microsoft.com}
\affiliation{
  \institution{Microsoft Research}
  \city{New York}
  \state{NY}
  \country{USA}
}
\author{Hoda Heidari}
\email{hheidari@cmu.edu}
\affiliation{
  \institution{Carnegie Mellon University}
  \city{Pittsburgh}
  \state{PA}
  \country{USA}
}
\author{Jon Kleinberg}
\email{kleinberg@cornell.edu}
\affiliation{
  \institution{Cornell University}
  \city{Ithaca}
  \state{NY}
  \country{USA}
}

\begin{abstract}
As AI technologies improve, people are increasingly willing to \textit{delegate} tasks to AI agents. In many cases, the human decision-maker chooses whether to delegate to an AI agent based on properties of the specific instance of the decision-making problem they are facing. Since humans typically lack full awareness of all the factors relevant to this choice for a given decision-making instance, they perform a kind of \textit{categorization} by treating indistinguishable instances -- those that have the same observable features -- as the same. In this paper, we define the problem of designing the optimal algorithmic delegate in the presence of categories. This is an important dimension in the design of algorithms to work with humans, since we show that the optimal delegate can be an arbitrarily better teammate than the optimal standalone algorithmic agent. The solution to this optimal delegation problem is not obvious: we discover that this problem is fundamentally combinatorial, and illustrate the complex relationship between the optimal design and the properties of the decision-making task even in simple settings. Indeed, we show that finding the optimal delegate is computationally hard in general. However, we are able to find efficient algorithms for producing the optimal delegate in several broad cases of the problem, including when the optimal action may be decomposed into functions of features observed by the human and the algorithm. Finally, we run computational experiments to simulate a designer updating an algorithmic delegate over time to be optimized for when it is actually adopted by users, and show that while this process does not recover the optimal delegate in general, the resulting delegate often performs quite well.
\end{abstract}

\begin{document}

\maketitle
\renewcommand{\shortauthors}{Greenwood et al.}

\section{Introduction}

Algorithmic agents -- such as robots, AI assistants, and chatbots -- are increasingly effective tools. These agents play an important role in recent optimism about the potential for autonomous entities on the Internet who can perform tasks on behalf of a user \citep{deng2023,yao2022webshop,drouin2024workarena,zhou2023webarena}. In many cases, a human user works with the algorithmic agent to complete a task, and the quality of the algorithmic agent must be measured by the performance of the combined human-algorithm system, or the \textit{team performance} \citep{bansal2019a, bansal2021}. Recent work has demonstrated that the algorithmic agent with the highest standalone accuracy does not necessarily achieve the optimal team performance \citep{bansal2021, hamade2024chesshandoff}.
Thus, it is important and natural to ask: \textit{How can one design the optimal algorithmic teammate?}

The answer to this depends on the structure of the team. There are many ways a human may partner with an algorithmic agent; we focus on teams where the machine
is a \textit{delegate} to which the human may hand off the choice of action entirely \citep{stout2014, lai2022, milewskilewis1997delegatingtosa}. That is, a delegate is an AI agent the human employs to take an action \textit{without reviewing its choice}. In cases where it is impossible or prohibitively time-consuming for the human and algorithm to communicate before selecting each action, the algorithm must be a delegate. Moreover, as AI technology improves, users may prefer to delegate decisions to increase efficiency. 
Examples of algorithmic delegates include AI assistants that may be dispatched to complete tasks on the Internet, such as shopping, answering emails, scheduling meetings, and booking tickets \citep{bookedaitool, wingaitool}. Another set of examples are semi-autonomous vehicles, where a driver must decide whether to drive herself or to delegate to an AI system \citep{gm2024,tesla2024}. In medicine, supervised autonomous surgical robots have been prototyped \citep{shademan2016surgery}.

A key factor in evaluating a delegate's team performance is determining in a given scenario whether a human user will hand off the choice of action to the delegate, and if not, which action she will take. Both of these decisions are subject to a fundamental limitation: the human will not have complete information about the specific instance of the decision-making problem in front of her, and as a result cannot make either decision perfectly \citep{bansal2019a}. For example, 
the user of an AI shopping agent on the Internet may not have details about past trends in the online market, 
and a driver may not know about rough road conditions ahead or a nearby hidden driveway that tends to cause accidents. The human must then make the best decision possible given what she can observe about the decision-making instance. We refer to these sets of human-indistinguishable instances as  \textit{categories}, evoking the body of work in cognitive science that studies humans treating distinct objects or circumstances in the same manner \cite{mervis1981, ashby2005,smith2014}.
The algorithmic agent -- hereafter, the \textit{machine} -- similarly has incomplete information \citep{alur2024}.
The AI shopping agent may not have full knowledge of the user's preferences, 
and the autonomous vehicle may not know about a large sports event in the neighborhood and other social context that might affect the vehicle's capacity to drive safely. Thus the machine also only observes its own \textit{categories} of states. We illustrate human and machine categories in Figure \ref{fig:grid:cats:nn}. 

\begin{figure}[ht]
    \centering
    \includegraphics[width=0.8\linewidth]{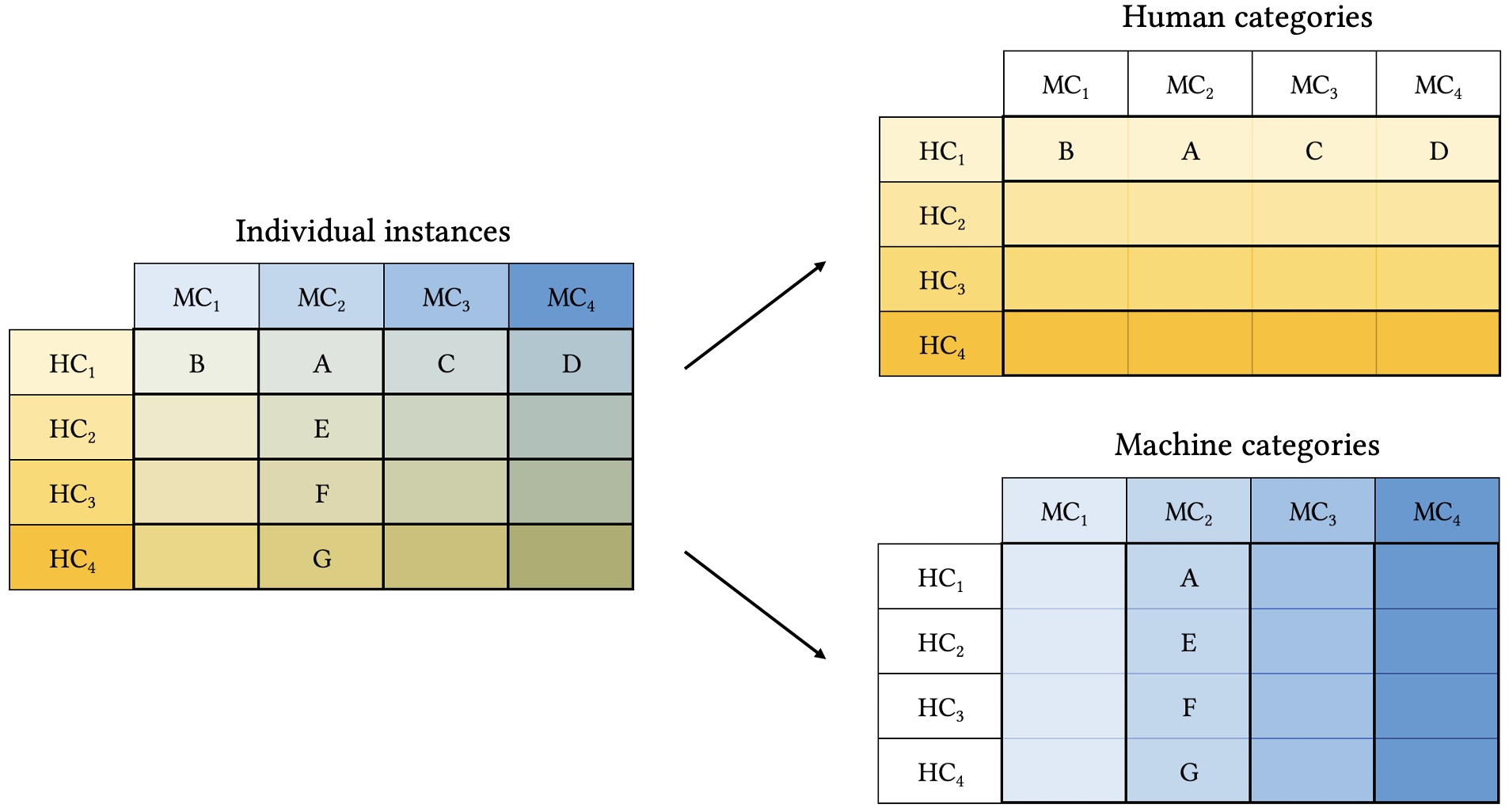}
    \caption{Visualization of human and machine categories. A decision-making instance corresponds to a cell in the grid. The human observes some information about the instance, and based on that information can identify a set $\textnormal{HC}_i$ ($i \in \{1,2,3,4\}$ in this example) of possible instances, or \textit{human category}, corresponding to a row in the grid. When encountering a scenario in $\textnormal{HC}_i$, the human must make decisions -- whether to delegate, and if not, which action to take -- to achieve the best \textit{expected} performance across the possible instances in $\textnormal{HC}_i$. The machine likewise observes some other information about the decision-making instance, which enables it to identify a set $\textnormal{MC}_j$ ($j \in \{1,2,3,4\}$ in this example) of possible instances it might be in, or \textit{machine category}, corresponding to a column in the grid. The machine selects actions that perform best across this set. For example, consider instance A laying at the intersection of $\textnormal{HC}_1$ and $\textnormal{MC}_2$. The human cannot distinguish this instance from B, C, and D, and the machine cannot distinguish it from E, F, and G. Note that this figure is a simplification: in general, there might be zero or many instances in the intersection of $\textnormal{HC}_i$ and $\textnormal{MC}_j$ for each $i, j$.}
    \Description{On the left, two four by four grids with rows $C_1$ to $C_4$ and columns $K_1$ to $K_4$. The top left grid has rows colored in different shades of yellow; the bottom left grid has columns colored in different shades of blue. On the right, the same four by four grid but with each entry having a unique color.}
    \label{fig:grid:cats:nn}
\end{figure}

We are interested in the optimal design of a machine in human-AI collaboration settings of the following nature, which we formalize in Section \ref{sec:model}. Given a decision-making instance, a human observes the human category the decision-making instance falls in, and compares the machine's expected performance across all instances falling in that category to her own performance. If the human has better expected performance, she will act herself. Otherwise, she will hand off\footnote{We use the verbs ``delegate'', ``adopt'', ``use'', and ``hand off'' interchangeably.} the decision to the machine. The machine will then observe the machine category the instance falls in, and will take some action. The machine's designer -- a third party -- must specify a choice of action for the machine to take in each machine category. We take the perspective of this designer, and ask what this choice of action should be: how should we design the machine so that the team performance is as good as possible?

To establish intuition about this optimal design problem, we consider the perspective of the designer of an AI agent for booking flights. An initial design choice would be an agent which maximizes team performance across all settings. However, the designer might discover that users never use this agent when booking business travel (for example, because the optimal decision depends heavily on the user's business travel guidelines and itinerary and is relatively independent of market data). It is then best for users if the designer removes business travel use cases from the training data, and re-optimizes the machine to perform well in human categories where the human actually adopts the AI agent (here, the category of ``personal travel''), potentially sacrificing the performance of the agent in booking business travel. This process need not stop here: this new machine may be used in a different set of categories, so that further re-optimizing the machine leads to further improvements to the team's performance. 
This suggests an interplay between the {\em delegation choices} made by the human user and the {\em design choices} made by the creator of the machine; we show in Section \ref{sec:simulations} that these dynamics converge to a machine with {\em locally optimal} team performance. 
Our interest is in the optimization problem where the designer of the machine anticipates this interplay between delegation choices and design choices without actually running the dynamics, and instead tries to directly design a machine whose use in a delegation scenario results in {\em globally optimal} team performance. 

\subsection{Overview of results} In this work, we investigate how to design the optimal algorithmic delegate in the presence of non-identical human vs. machine categories. 
We have four main contributions.
\begin{itemize}
    \item We develop a  model of human-AI delegation which captures categories, and find a surprisingly clean \textit{characterization} of the optimal design problem as a search problem for the optimal subset of categories to perform well on. Thich reveals a fundamental combinatorial structure of this problem, and leverages the principle of designing AI tools to perform well in categories where they are actually used \citep{bansal2021, vafa2024}.
    \item 
    To build intuition, we consider special cases where the optimal design problem can be fully and tractably characterized; we do this for the case where the human and machine each observe a \textit{single feature}. We show that the optimization landscape is surprisingly rich even in this basic case.
    \item We use our characterization result to examine the \textit{computational tractability} of designing the optimal delegate. We discover an efficient algorithm for finding the optimal delegate when the optimal actions and category distributions are each decomposable in a sense we make precise. This includes settings where the human and machine features of the decision-making instance are realized independently and the optimal action is a linear function of these features. We also find an efficient algorithm for delegation settings in which the human or machine only observes a small number of features. However, we show that it is NP-hard to find an optimal delegate in general.
    \item Finally, we run \textit{computational experiments} to simulate the process of iteratively developing a machine based on observing where it is used, and show that this does not find an optimal delegate in general, but may perform quite well.
\end{itemize}

Our model of human delegation to an algorithmic agent is intentionally parsimonious: we only need to specify which features of the decision-making instance the human and machine observe, and for each instance, the probability of observing that instance, and the optimal action to take. We show that this simple model captures a range of insights about optimal delegation and the role of categories. Moreover, this model is easy to extend in several directions: in Section \ref{sec:extensions}, we show how to 
adapt our model
to incorporate other human behavioral biases that affect delegation \citep{bobadillasuarez2017,candrianscherer2022}, such as algorithm aversion \citep{dietvorst2014} and control premiums \citep{owens2014}.

The remainder of the paper follows the organization above: in Section \ref{sec:model} we introduce our model, set up the problem of finding the optimal algorithmic delegate, and show a characterization theorem for the optimal delegate. We then examine optimal delegation in the simple case in which the human and machine each access a single binary feature in Section \ref{sec:twofeature}.
We consider the tractability of finding an optimal delegate in Section \ref{sec:results}, and compare the optimal delegate to delegates learned over time 
in Section \ref{sec:simulations}. Finally, in Section \ref{sec:discussion} we discuss extensions of our model, and general relationships in performance between different teams.

\subsection{Related work} 
This paper extends a long line of work addressing how to design algorithmic agents to improve performance in human-AI  teams \citep{bansal2019b,bansal2021,milewskilewis1997delegatingtosa, detoni2024, grgichlaca2022}. In particular, we focus on cases in which the human must choose to either delegate to the machine \citep{lai2022,lubarstan2019} or take action by herself; to our knowledge, we are the first to study the design of optimal delegates that account for categories. \citet{bansal2021} propose that a machine's performance should be optimized for the categories where it is used, but do not study delegation. Other work studies machine design in other team structures, such as settings where the machine delegates to the human \citep{raghu2018, lykouris2024}, settings where the human may observe the machine's output before deciding \citep{bansal2021, bansal2019b}, and settings where a human chooses between accepting a machine's prediction and delegating to another human \citep{xu2024}.

The existence of categories is a known aspect of human decision-making \citep{anderson1991,mullainathan2002, townsend2000, zheng2023}. This act of agents grouping decision-making scenarios has been modeled both as a behavioral phenomenon \citep{mullainathan2002,townsend2000} and as a consequence of limited information \citep{bansal2019a, alur2024, vafa2024}.
We focus on categories arising from limited information, but in Appendix \ref{app:arbitrarycategories} we show that our results apply to any partitions of the set of decision-making instances -- including those arising as a behavioral convenience. Our focus on information-driven categories is consistent with contemporary human-algorithm collaboration work: ``mental models'' \citep{bansal2019a} and ``indistinguishable'' inputs \citep{alur2024} are analogous to our human and machine categories respectively. Moreover, our human categories are captured by ``generalization functions'' in \citet{vafa2024} when the human maintains equivalent beliefs about all states with the same human-accessible features. \citet{iakovlev2024} also study a model where a human and machine have access to different information in the form of binary features, but focus on the problem of a third party selecting an evaluator. Some of these prior works study delegation in the presence of categories \citep{lai2022,vafa2024}, but do not investigate the question of optimal delegate design. 

This work also interfaces with a variety of other disciplines. Our model can represent a type of \textit{interpretability} \citep{carvalho2019} by the number of shared features between the human and machine: in our model, more shared features corresponds to more similar categories and actions taken. Categories can cause \textit{over-reliance} \citep{passi2022}: the need to make the same delegation decision across a category can result in over- or under-delegation within categories. In human factors analysis, function allocation prescribes a qualitative process by which a designer determines which tasks in a system should be automated \citep{marsden2005,price1985}; we take a quantitative approach which is more suited to modern AI design. There is also an extensive literature in management on human to human delegation \citep{leana1986, akinola2017}; while human delegates can be hired to perform specialized tasks, or trained to perform these tasks, the natural models for such a human-to-human interaction would not lead to the type of optimization design problem we study here because human decision-makers cannot be `designed' to perform well as teammates in the same fine-grained way as algorithmic agents can in our setup.

Finally, \textit{agentic AI} systems, which complete complex, long-term, or underspecified goals with a high degree of agency, are an increasingly popular topic of study in research and industry \citep{hbr2024, oai2023, chan2023}. Our work studies AI systems in arguably the \textit{simplest} agentic setting, where the agent must  act autonomously in clearly defined scenarios involving a single decision. We  find that even in this simple setting, designing an optimal agent is a fundamentally hard task. This suggests that in more general settings it may also be challenging to design agentic systems which will be optimally relied on by humans.
    
\section{Model}\label{sec:model}

Let the world of all decision-making instances be described by $d$ binary features $x_1, ..., x_d \in \{0, 1\}$. There are then $n = 2^d$ possible \emph{instances} or \textit{states} of the world $\x = (x_1, ..., x_d) \in \{0, 1\}^d$. Let $\Prob(\x)$ denote the probability of observing state $\x$.
In state $\x$, there is some ground truth correct action $\fnopt(\x) \in \mathbb{R}$ that should be taken. If an agent takes action $a \in \mathbb{R}$ in state $\x$, the team will receive a loss of $(a-\fnopt(\x))^2$. There are two agents, a human and a machine. The human can observe features 
$\idxset{H} \subseteq \{1, ..., d\} =: [d]$; the machine can observe features 
$\idxset{M} \subseteq [d]$. 
A \textit{delegation setting} is determined by $\idxset{H}, \idxset{M}, \Prob,$ and $\fnopt$.

To situate how this formalism works through a brief example,
suppose that the machine is an AI shopping agent as in the introduction, which can traverse the Internet to purchase items on a user's behalf. In this case, $\x$ represents the features of an item on a given day, and the ground truth optimal action $\fnopt(\x) \in \mathbb{R}$ is the highest price a user should pay for that item on that day. The human-observable features $\idxset{H}$ could be the users' preferences and level of urgency for purchasing the item, and $\idxset{M}$ could be information on how the item's current price compares to market trends for that item on the Internet. The human and machine might share some features -- the human could communicate some preferences to the machine -- but the human cannot fully articulate all the subtleties of her preferences and the machine cannot fully summarize all the complex market trends involved in a way that's legible to the human. 

A human category $C$ is a set of states that are indistinguishable to a human because they all share the same human-observable feature values.  Let $\hcatvec{x}$ and $\mcatvec{x}$ denote the restrictions of $\x$ to human and machine-observable features. Formally, states $\x$ and $\z$ are in the same human category if and only if $\hcatvec{x} = \hcatvec{z}$. Similarly, a machine category $K$ is a set of states indistinguishable to the machine, so that states $\x$ and $\z$ are in the same machine category if and only if $\mcatvec{x} = \mcatvec{z}$. Let $\hcats$ and $\mcats$ denote the set of all human and machine categories respectively; $\hcats$ and $\mcats$ are each a partition of the states. There are $h = 2^{|\idxset{H}|}$ and $m = 2^{|\idxset{M}|}$ human and machine categories respectively; we enumerate the human categories $\hcati{1}, \hcati{2}, ..., \hcati{h} \in \hcats$ and the machine categories $\mcati{1}, \mcati{2}, ..., \mcati{m} \in \mcats$. We illustrate the categories induced by human- and machine-accessible features through an example in Figure \ref{fig:categoryillustration}.
For a state $\x$, let $\hcat(\x)$ and $\mcat(\x)$ be the human and machine categories containing $\x$. Moreover, for a 
set $S$ let $\Prob(S) := \sum_{\x \in S} \Prob(\x)$. 

\begin{figure}[ht]
    \centering
    \includegraphics[width=0.5\linewidth]{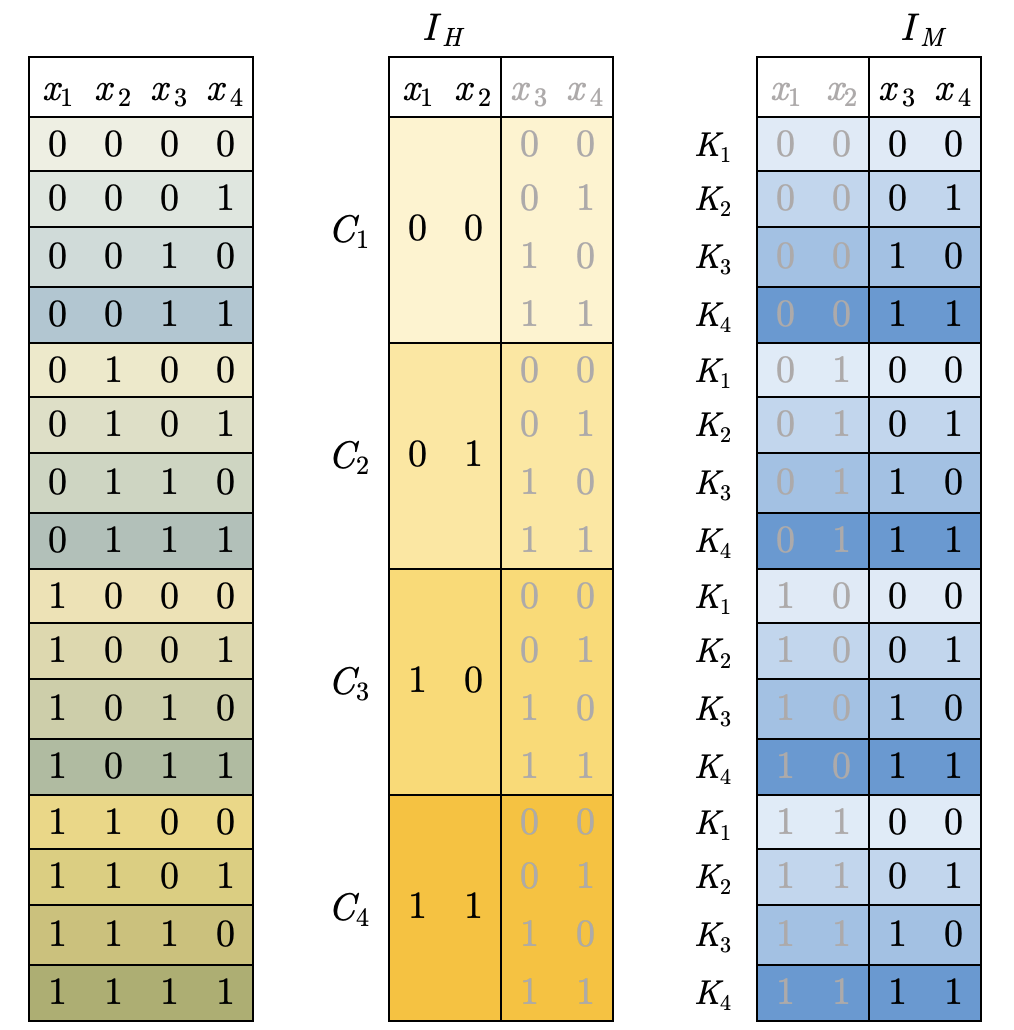}
    \caption{Categories emerge from limited access to information. Here there are $d = 4$ features, resulting in $n = 2^d = 16$ states shown in the first panel. The human has access to features $\idxset{H} = \{1, 2\}$. The middle panel visualizes the human categories in different shades of yellow: each collection of states with the same values of $x_1$ and $x_2$ corresponds to a category. The machine has access to features $\idxset{M} = \{3, 4\}$. The right panel visualizes the machine categories in shades of blue: each collection of states  with the same values of $x_3$ and $x_4$ corresponds to a category.}
    \Description{On the left, there is a table with four binary features $x_1$ to $x_4$, and all possible settings of these features. Each setting of features has a different color. In the middle, the rows are colored different shades of yellow based on the values of $x_1, x_2$. On the right, the rows are colored different shades of blue based on the values of $x_2, x_3$.}
    \label{fig:categoryillustration}
\end{figure}

Since each agent can't distinguish between states within a category, the human and machine choose actions as a function of the category they observe: let $\hfn: \hcats \to \mathbb{R}$ and $\mfn: \mcats \to \mathbb{R}$ denote the human and machine's choices, or \textit{action functions}. The human's action function $\hfn$ specifies the action she will take if she does not delegate to the machine; as we will see, this is simply the action she takes in the machine's absence. The machine's action function is a policy specified by the machine's designer before the machine's deployment.

The delegation process works as follows.
\begin{enumerate}
    \item Given state $\x$, the human observes the category $\hcat(\x)$.
    \item The human decides whether or not to delegate based on which agent has better expected performance in $\hcat(\x)$.
    \item If the human does not delegate, she takes action $\hfn(\hcat(\x))$. 
    \item If the human delegates, the machine observes the category $\mcat(\x)$, and takes action $\mfn(\mcat(\x))$. 
\end{enumerate}
We illustrate this in Figure \ref{fig:decisionflow}. Returning to the shopping agent example, the human can observe her preferences but not the market trends, resulting in a set of indistinguishable states that corresponds to a human category $C$. Based on $C$, she decides whether to delegate to the machine, or take action $\hfn(C)$. If she delegates to the machine, the machine observes the market trends but not all aspects of the user's preferences, resulting in a set of indistinguishable states that corresponds to a machine category $K$, and takes action $\mfn(K)$.
\begin{figure}[ht]
    \centering
    \includegraphics[width=\linewidth]
    {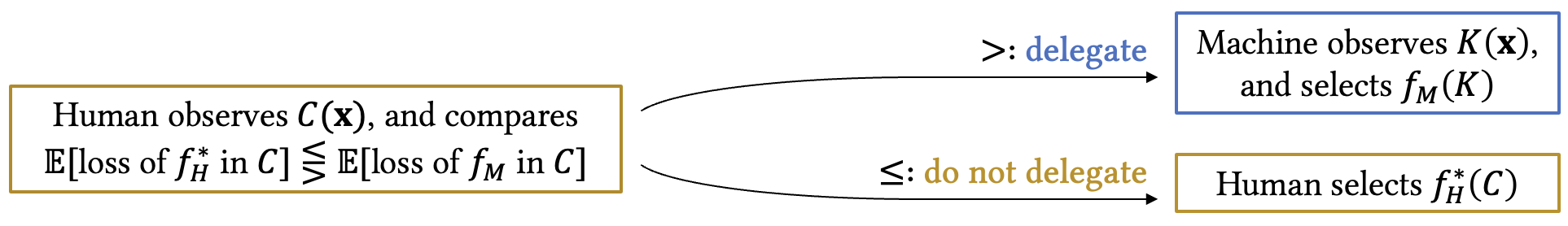}
    \caption{Flow of decision-making control in delegation.}
    \Description{...}
    \label{fig:decisionflow}
\end{figure}

Note that we assume the human knows whether her expected loss is better or worse than the machine's loss in each category; for example she may learn this over time from past experiences delegating to the machine, or from other users who have observed the machine's performance.

Besides the human and machine, there is a third actor: the machine's designer. We aim to characterize the optimal machine action function $\mfn$, and determine whether the machine's designer is able to tractably discover this optimal policy. We will consider two regimes: for the main portion of the paper, we will assume the machine's designer has \textit{full knowledge} of the delegation setting $\idxset{H}, \idxset{M}, \Prob$ and $\fnopt$. In Section \ref{sec:simulations}, we relax this assumption. The assumption of full knowledge provides a lower bound on the difficulty of the problem: if the designer cannot efficiently construct the optimal machine with full information, the designer has no hope of achieving optimality with less information. Moreover, while stylized, the assumption reflects a genuine property of the problem domain, which is that the designer of the machine has more information than the human user or the machine itself have individually: the designer has some understanding of the categories a human user will have access to (but of course they don't not know which category any instance belongs to at the future time of usage), and the designer also has insight into the features that the machine will use. For example, a company developing a shopping agent might have a user experience team which aims to understand customer use cases, and an engineering team which assembles the machine's inputs and sensors.

\subsection{An illustrative example} Before we finish formalizing the optimal design problem, we will work through an example, shown in Figure \ref{fig:grid:fulldemo}. Here, the human and machine each have four categories, as would arise when each has access to two underlying Boolean features. If the human and machine features partition a set of four total features, there will be one state $\x_{ij}$ in the intersection of $C_i \cap K_j$ for each $i,j$. Figure \ref{fig:grid:fulldemo:f} visualizes these categories in a grid, where the rows are human categories and the columns are machine categories. The value of entry $(i,j)$ in the grid is the ground truth correct action $\fnopt(\x_{ij})$, and the entry is shaded by the magnitude of $\fnopt$. We let each state have equal probability $\Prob(\x_{ij}) = 1/n$.

\begin{figure*}[!ht]
    \centering
    \begin{subfigure}[b]{\linewidth} 
        \centering
        \includegraphics[width=0.34\linewidth]{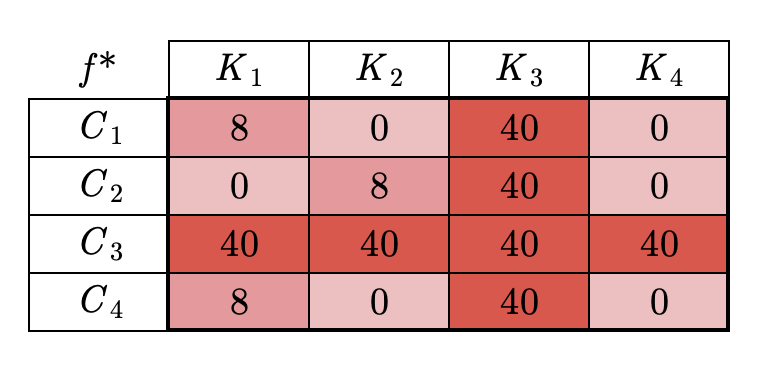}
        \caption{We arrange decision-making instances in a grid as in Figure \ref{fig:grid:cats:nn}, where rows are human categories and columns are machine categories. The value in entry $(i, j)$ is the optimal action $\fnopt(\x_{ij})$ in the state $\x_{ij} \in C_i \cap K_j$.}
        \Description{A grid labeled $\fnopt$ with rows $C_1$ to $C_4$ and columns $K_1$ to $K_4$. The grid entries are colored according to the magnitude of the value in that entry: both $K_3$ and $C_3$ have very large outlying values; the remainder of the entries are eiteher 0 or 8.}
        \label{fig:grid:fulldemo:f} 
    \end{subfigure}

\begin{subfigure}[b]{0.49\linewidth} 
        \centering
        \includegraphics[width=\linewidth]{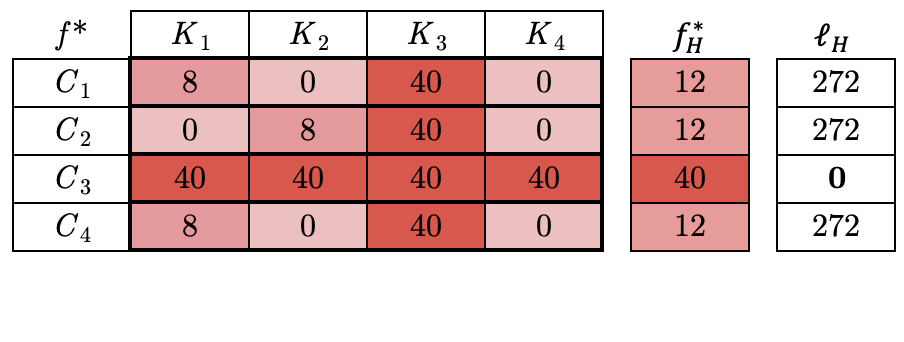}
        \caption{In category $C_i$, the human selects the loss-minimizing action $\hfnopt$, the average action in row $i$. The expected loss of $\hfnopt$ in human category $C_i$ is $\ell_H(C_i)$.\\}
        \Description{Worked example showing the loss of a human acting alone.}
        \label{fig:grid:fulldemo:humanonly} 
    \end{subfigure}
    \hfill
\begin{subfigure}[b]{0.49\linewidth} 
        \centering
        \includegraphics[width=0.8\linewidth]{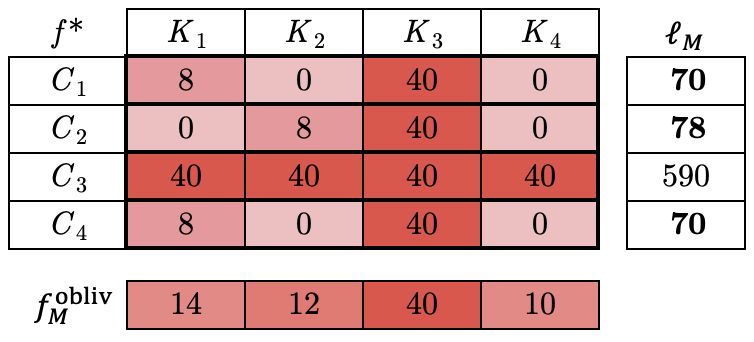}
        \caption{In category $K_j$, the machine selects some action $\mfn$. The expected loss of $\mfn$ in human category $C_i$ is $\ell_M(C_i)$. Here we display $\mfn = \mfnobliv$, corresponding to selecting the average action in column $j$.}
        \Description{Worked example showing the loss of a machine designed to perform well across all human categories.}
        \label{fig:grid:fulldemo:machineonly} 
    \end{subfigure}
    
    \begin{subfigure}[b]{\linewidth} 
        \centering
        \includegraphics[width=0.6\linewidth]{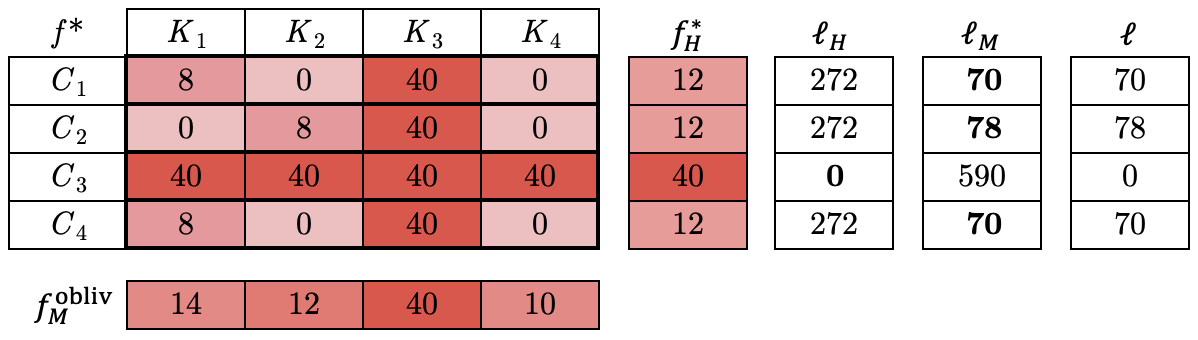}
        \caption{Team with $\hfnopt$ and $\mfnobliv$; team loss $\ell(C)$ is the minimum of $\ell_H(C)$ and $\ell_M(C)$.}
        \Description{A machine designed to operate in all human categories has bad team performance.}
        \label{fig:grid:fulldemo:bad} 
    \end{subfigure}
    
\begin{subfigure}[b]{\linewidth} 
        \centering
        \includegraphics[width=0.6\linewidth]{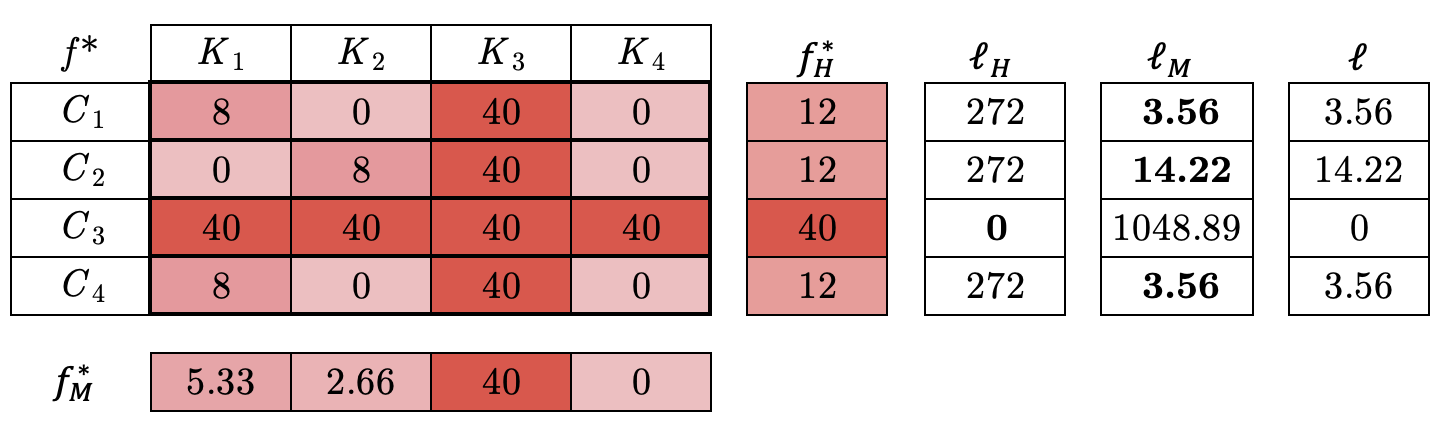}
        \caption{Team with $\hfnopt$ and $\mfnopt$; in category $K_j$, $\mfnopt$ averages over $\{C_1, C_2, C_4\}$ and achieves a much lower team loss.}
        \Description{A machine designed to operate only in $\{C_1, C_2, C_4\}$ performs well.}
        \label{fig:grid:fulldemo:good} 
    \end{subfigure}
    \caption{Example delegation problem.}
    \label{fig:grid:fulldemo}
\end{figure*}

In a given state $\x_{ij}$, the human can observe which row $i$ the state is in, and will either delegate to the machine, or take action $\hfn(\hcat_i)$. When the human does not delegate, her choice of $\hfn(\hcat_i)$ is simple: the only information the human has is the row $i$, so to minimize the squared loss of her decision she can do no better than to choose the average action in row $i$, which we show in Figure \ref{fig:grid:fulldemo:humanonly} as $\hfnopt$. This yields an average loss of $\ell_H(C_i)$ in each human category $C_i$, which corresponds to the variance of the optimal action across $C_i$. When the human does delegate to the machine in state $\x_{ij}$, the machine can observe the column $j$ and take some action $\mfn(\mcat_j)$. In Figure \ref{fig:grid:fulldemo:machineonly} we show a machine $\mfnobliv$, which upon observing category $K_j$ takes the average action in column $j$. This yields an average loss of $\ell_M(C_i)$ in each human category $C_i$.

In Figures \ref{fig:grid:fulldemo:bad} and \ref{fig:grid:fulldemo:good} 
 we contrast the team performance of different delegates. In state $\x_{ij}$, the human makes her delegation decision by comparing the losses $\ell_H$ and $\ell_M$ in row $i$, and choosing the agent with the lower loss $\ell$. The overall team loss is the average of the column $\ell$. 
 
 In Figure \ref{fig:grid:fulldemo:bad} we show the team consisting of $\hfnopt$ from Figure \ref{fig:grid:fulldemo:humanonly} and $\mfnobliv$ from Figure \ref{fig:grid:fulldemo:machineonly}. In this case, the human will delegate in categories $C_1, C_2,$ and $C_4$. If the machine can only observe the machine category $K_j$, this delegate $\mfnobliv$ that takes the average action in column $j$ is the best choice to minimize the squared loss. 

However, the machine will only be used in decision-making instances where the human chose to delegate to it. Can we use this information to improve the team loss? In this example, states in the human category $C_3$ are difficult for the machine to handle: the optimal action $40$ is often very different than in other states in the machine's categories. However, $C_3$ is very easy for the human to take action in: the human can take action $40$ with no loss. Without $C_3$, the remaining values in each machine category have very low variance, and the machine achieves very good performance even if the human adopts the machine in all other human categories. Figure \ref{fig:grid:fulldemo:good} 
shows a machine design that averages only over the human categories excluding $C_3$; this machine achieves a significantly lower loss in categories $C_1, C_2$ and $C_4$, and is again adopted in these categories. This is in fact 
the optimal delegate.

The example above suggests two high-level insights. First, the machine's designer should partition the human categories into ``retained'' categories where the machine will be responsible ($C_1$, $C_2$, and $C_4$) and categories that are ``yielded'' to the human's responsibility  $(C_3)$, and optimize the machine to be perform well in the retained categories. We generalize this in \Cref{prop:dropisoptimal}, which shows that finding an optimal machine amounts to finding the best set of retained categories, and that the human will adopt the optimal machine in exactly these categories.
Second, to find this partition, 
the machine's designer should take the yielded human categories to be those with low variance in the optimal action among states in the category, and the retained human categories to be those that make the machine categories have low variance. The optimal delegate is then designed to only be used in the latter set of categories. In \Cref{prop:variancegame:full} we show this is exactly the problem of designing the optimal delegate.

One natural approach to finding this partition of human and machine categories is to start with a delegate that is designed to perform well in all states, and then to observe which categories human adopts the machine in. The designer could then re-optimize the machine to perform well in only those categories and observe which categories this new machine is adopted in. The designer then reoptimize the machine to perform well in these categories, and thus iteratively improve the machine's design. Unfortunately, we show in Section \ref{sec:simulations} that this is not guaranteed to converge to an \textit{optimal} delegate, but reaches a local optimum. As a result, in the next few sections we focus on finding a global optimum.

\subsection{Team loss minimization} We now finish formalizing the problem of finding an optimal delegate. Let \[\hloss(\hfn, \hcat) := \E[(\hfn(C) - \fnopt(\x))^2 | \x\in \hcat],\]
\[\mloss(\mfn, \hcat) := \E[(\mfn(K(\x)) - \fnopt(\x))^2 | \x \in \hcat]\] denote the expected loss of the human function $\hfn$ and machine function $\mfn$ respectively in category $C$. When the human is delegating optimally, she will delegate to the machine in category $\hcat$ if and only if the machine's loss is lower than the human's in that category.\footnote{We assume, without loss of generality, that the human will not delegate when the two agents have equal expected losses.}
        The loss associated with the team in which the human with function $\hfn$ optimally delegates tasks to the machine with function $\mfn$ is denoted by $\loss(\hfn, \mfn)$. For simplicity, we refer to this as the \textit{team loss}. Formally, the expected team loss is 
        \[
            \loss(\hfn, \mfn) :=\sum_{\hcat \in \hcats} \Prob(\hcat)\cdot \min\left\{\hloss(\hfn,\hcat) ,\mloss(\mfn, C)\right\}.\]
            
In a given category $\hcat$, the human will observe $C$, and either delegate to the machine, or will take an action. If the human chooses to take an action, given that the human can only observe $\hcat$, the loss-minimizing action in category $\hcat$ is the expected value of $\fnopt$ in $\hcat$; we denote the human function that takes this optimal action in each $\hcat$ by $\hfnopt$, that is,
\[\hfnopt(\hcat) := \E[\fnopt | \hcat].\]
This will be the optimal human function regardless of the machine's function $\mfn$, so we will restrict our consideration to $\hfn = \hfnopt$.

The optimal machine action is less clear: in machine category $\mcat$, the machine's designer can make this choice depend not only on the features corresponding to $\mcat$, but also on the information that the human chose to delegate to it (by nature of having to choose an action at all). The machine's design could be oblivious, relying only on the machine-interpretable features and ignoring the fact that it has been delegated to, as in a traditional machine learning approach. The loss-minimizing action $\mfnobliv$ in machine category $K$ that ignores delegation information would be the expected value of $\fnopt$ across states in $K$,
\[\mfnobliv(\mcat) := \E [\fnopt | \mcat].\] 
Alternatively, an \textit{optimal delegate} $\mfnopt$ makes use of all the information available and attains the optimal team loss, 
\[\mfnopt \in \argmin_{\mfn} \loss(\hfnopt, \mfn).\]

\varXiv{
\begin{restatable}{fact}{factarbworse}\label{fact:arbworse}
    There exist families of delegation settings such that the oblivious machine $\mfnobliv$ has unbounded team loss, but the team loss of the optimal machine $\mfnopt$ is constant.
\end{restatable}
We show this in Section \ref{app:arbitrarilyworse}; this implies that the oblivious delegate can have arbitrarily worse performance than the optimal delegate. This observation motivates the need to design optimal delegates that account for human adoption.
}
{
We show in Section \ref{app:arbitrarilyworse} that there exist families of delegation settings such that the oblivious machine $\mfnobliv$ has unbounded team loss, but the team loss of the optimal machine $\mfnopt$ is constant. This implies that the oblivious delegate can have arbitrarily worse performance than the optimal delegate. This observation motivates the need to design optimal delegates that account for human adoption.
}

\subsection{Reformulating the problem} 
Our goal is now to find optimal delegates $\mfnopt$. We now state two results 
(Propositions \ref{prop:dropisoptimal} and \ref{prop:variancegame:full})
that give general versions of  
the principles from our example above; in Section \ref{sec:results} we will use these results to examine the tractability of finding optimal delegates. 

First, we transform our problem into a discrete optimization problem.
Let $\mfn^\kept$ be the machine that in machine category $\mcat$ takes the average action only over states also in ``retained'' human categories $C \in \kept$:
\[\mfn^\kept(\mcat) := \E[\fnopt | X(\kept) \cap K],\] 
where $X(\kept) := \bigcup_{C \in \kept} C$. In contrast, the oblivious machine takes the expected optimal action over \textit{all} states in the machine category. Note that the functions $\mfn$ we computed in the example above were examples of $\mfn^\kept$ for different sets of retained categories $\kept$.

The categories in which a human with function $\hfn$ will delegate to a machine with function $\mfn$ are those where the machine's expected loss is lower than the human's; we denote this by
\[\dcats(\hfn, \mfn) := \{\hcat \subseteq \hcats: \mloss(\mfn, C) < \hloss(\hfn, C)\}.\]

We may rewrite the problem of finding an optimal delegate $\mfnopt$ as \begin{align}\label{problem:continuous}
        \min_{\mfn} \loss(\hfnopt, \mfn) \equiv \min_{\mfn}\sum_{\hcat \in \hcats \setminus \dcats(\hfnopt, \mfn)}\Prob(\hcat) \hloss(\hfnopt, \hcat) + \sum_{\hcat \in \dcats(\hfnopt, \mfn)} \Prob(\hcat) \mloss(\mfn, \hcat).
    \end{align}
The following result shows that to find an optimal delegate $\mfnopt$, it is sufficient and necessary to find a set of human categories $\kept^*$ that attains the minimum team loss when the human delegates to $\mfn^{\kept^*}$ in precisely the categories in $\kept^*$. We defer the proof of this result to Appendix \ref{app:charproofs}.
\begin{restatable}{proposition}{dropisoptimalinformal}\label{prop:dropisoptimal}
    Consider the combinatorial optimization problem \begin{equation}\label{problem:discrete}
        \min_{\kept \subseteq \hcats} \sum_{C \in \hcats \setminus \kept} \Prob(C) \hloss(\hfnopt, C) + \sum_{C \in \kept} \Prob(C) \mloss(\mfn^\kept, C)
    \end{equation}

    In any delegation setting, if $\kept^*$ is a solution to Problem \ref{problem:discrete}, $\mfn^{\kept^*}$ is a solution to Problem \ref{problem:continuous}. 
    
    Likewise, if $\mfnopt$ is a solution to Problem \ref{problem:continuous}, then $\dcats(\hfnopt,\mfnopt)$ is a solution to Problem \ref{problem:discrete}.
\end{restatable}

Since we measure team performance using the squared loss, the objective function in Problem \ref{problem:discrete} depends fundamentally on the variance -- the variance across each non-retained human category, and the variance across all retained states in each machine category. To formalize this, we use \textit{conditional variance} \citep{spanos2019}, \[\Var(\fnopt | S) := \E[(\fnopt - \E[\fnopt | S])^2 | S].\] Define $\modvar(\cdot | S) := \Prob(S)\Var(\cdot | S).$ 
\begin{restatable}{proposition}{propvariancegameinformal}\label{prop:variancegame:full} To find an optimal delegate $\mfnopt$, it is sufficient to find a set $\kept^*$ that solves \[\min_\kept \sum_{\hcat \in \hcats \setminus \kept} \modvar(\fnopt| \hcat) + \sum_K \modvar\left(\fnopt| K \cap X(\kept)\right)\] and take $\mfnopt = \mfn^{\kept^*}$. 
\end{restatable}

This follows from properties of the squared loss; we provide a formal proof in Appendix \ref{app:charproofs}.

We first use this reformulation to characterize optimal delegates in the simple setting where the human and machine each access a single feature in Section \ref{sec:twofeature}, and again in Section \ref{sec:results} to investigate the tractability of finding an optimal delegate. 

While our results in Section \ref{sec:results} will apply for arbitrary human- and machine-observable features, we note that this problem has an especially clean combinatorial formulation when the human and machine partition the set of features $[d]$.

\begin{restatable}{corollary}{propvariancegamesimple}\label{cor:variancegame:simple}
    In the case that $\idxset{H} \cap \idxset{M} = \emptyset$ and $\idxset{H} \cup \idxset{M} = [d]$, the problem of finding an optimal delegate is as follows. There is a single $\x_{ij} \in C_i \cap K_j$ for each $i,j$. Define a matrix $V \in \mathbb{R}^{h} \times \mathbb{R}^{m}$ with entries $v_{ij} = \fnopt(\x_{ij})$.  
    In this notation, we use $\modvar(v_{ij} |i)$ to denote $ \modvar(\fnopt | C_i)$, and likewise $\modvar(v_{ij} | j)$ to denote $\modvar(\fnopt | K_j)$.
\begin{quote}\textnormal{\textsc{VarianceAssignment.} Fix a set of rows $R \subseteq [h]$. For each row $i \notin R$, pay a cost of $\modvar(v_{ij} |i)$. For each column $j$, pay a cost of $\modvar(v_{ij}| i \in R, j)$ (indicating that we take the conditional variance for the given $j$ only over the rows in $R$). Find the set $R^*$ that minimizes the total cost.
    }
\end{quote}
    Then for $\kept^* = \{C_i : i \in R^*\}$, $\mfn^{\kept^*}$ will be an optimal delegate.
\end{restatable}

\section{Delegation in two-feature settings}\label{sec:twofeature}

Before introducing our general results, we will consider this model in the special case where the world only has two relevant features -- the first feature being accessible only to the human and the second accessible only to the machine -- and where each state occurs with equal probability. We will first show two extreme examples of delegation settings in which the optimal machine design is intuitive. We will then fully characterize the relationship between the ground truth function $\fnopt$ and the optimal machine design. This exercise will demonstrate that this relationship follows simple principles but may itself be complex. 

Formally, let there be $d =2$ features $x_1$ and $x_2$. Let the human access $x_1$ only, and the machine access $x_2$ only ($\idxset{H} = \{1\}, \idxset{M} = \{2\}$). In this case, each agent has two categories, and there are four states $\x$ -- one state $\x_{ij} \in C_i \cap K_j$ for each pair of human and machine categories $C_i, K_j$. Thus, any $\fnopt$ can be specified by four values: the optimal action in each state. It is easy to show that the optimal set of retained categories $\kept^*$ is invariant to scaling and shifting of the ground truth function $\fnopt$, so we may without loss of generality fix two of the values of $\fnopt$: let $\fnopt(\x_{11}) = 0$ and $\fnopt(\x_{12}) = 1$. We denote the remaining two values $a := \fnopt(\x_{21})$, $b := \fnopt(\x_{22}) $, and write $\fnopt = \fnopt_{a,b}$.\footnote{All functions $\fnopt$ can be scaled and shifted to equal $\fnopt_{a,b}$ for some $a,b$ as long as $\fnopt(\x_{11}) \neq \fnopt(\x_{12})$. If $\fnopt(\x_{11}) = \fnopt(\x_{12})$, the human achieves zero loss in $C_1$, the optimal delegate takes $\kept = \{C_2\}$ and achieves zero loss there, so that the team achieves zero loss.} We visualize $\fnopt_{a,b}$ in Figure \ref{fig:twofeature:general}. 
\begin{figure}[ht]
    \centering
    \includegraphics[width=0.25\linewidth]{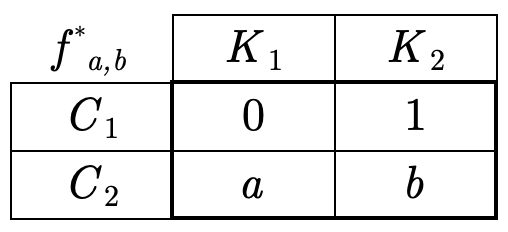}
    \caption{General two-feature setting}
    \Description{A grid labeled $\fnopt_{\fx,\fy}$ with two rows and two columns. The rows are labeled $C_1, C_2$; the columns are labeled $K_1, K_2$. The value in row $C_1$ and column $K_1$ is 0; the value in row $C_1$ and column $K_2$ is 1; the value in row $C_2$ and column $K_1$ is $\fx$; the value in row $C_2$ and column $K_2$ is $\fy$.}
    \label{fig:twofeature:general}
\end{figure}

First, consider the case where $a = 0$, $b = 1$, as shown in Figure \ref{fig:twofeature:center}. In this case, the optimal action is entirely determined by the machine's feature $x_2$: the machine has perfect information. The optimal machine therefore retains all human categories, and attains perfect team performance. In general the expected team loss from $\kept = \{C_1, C_2\}$ increases radially as away from $(0, 1)$: we show in Appendix \ref{app:twofeatjustification} that it is proportional to $a^2 + (b-1)^2$. 

\begin{figure}[ht]
    \centering
    \begin{subfigure}[t]{0.49 
    \linewidth} 
        \centering
        \includegraphics[width=0.95\linewidth]{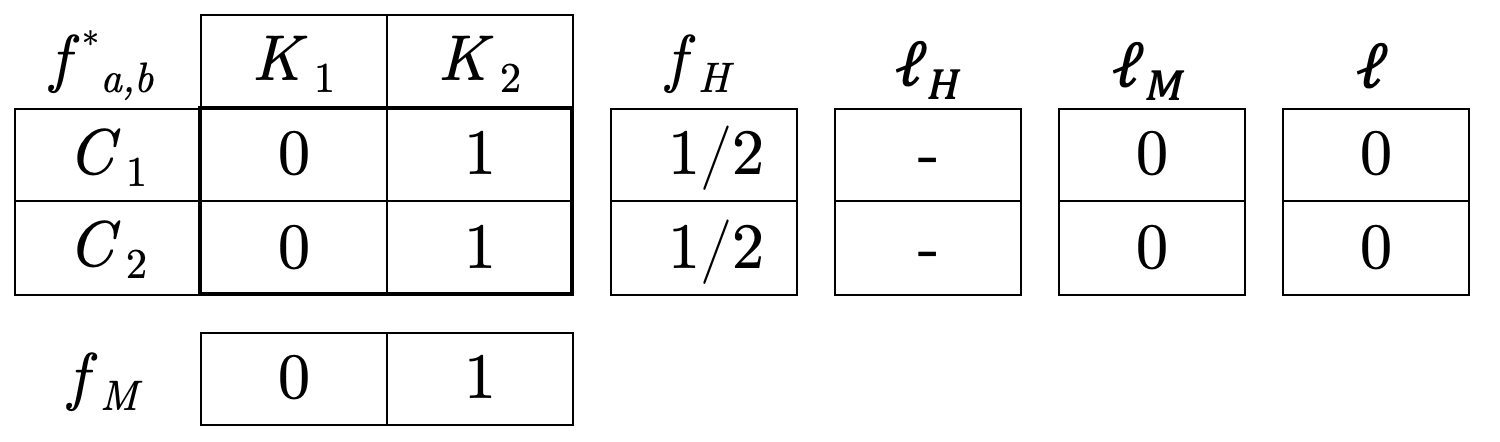}
    \caption{$a = 0$, $b = 1$.}
    \Description{A worked example showing the loss of the human and the optimal delegate for $\fn_{a,b}$ for $a = 0, b = 1$}
    \label{fig:twofeature:center}
    \end{subfigure}
    \begin{subfigure}[t]{0.49
    \linewidth} 
        \centering        \includegraphics[width=0.95\linewidth]{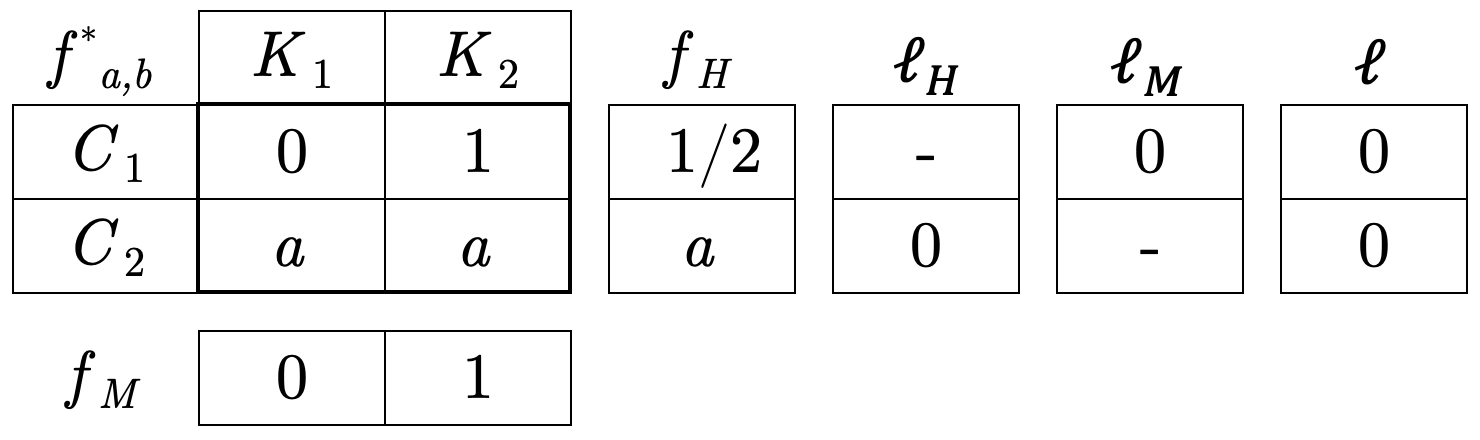}
        \caption{$a = b$.}
        \label{fig:twofeature:line}
        \Description{A worked example showing the loss of the human and the optimal delegate for $\fn_{a,b}$ for $a = b$}
    \end{subfigure}
    \caption{Optimal delegation for different settings of $a,b$. A dash (-) denotes an irrelevant value.}
\end{figure}

If $a = b$, the human-machine team operates with zero loss when the machine retains only $\kept = \{C_1\}$. We illustrate this in Figure \ref{fig:twofeature:line}: the human has perfect performance in $C_2$, and delegates to the machine in $C_1$ where the machine has perfect performance. In general, the loss from $\kept = \{C_1\}$ increases quadratically away from the line $a = b$: it is proportional to $(a-b)^2$.

Note that similarly, if the machine retains only the second human category, $\kept = \{C_2\}$, the machine will predict $\mfn(K_1) = a, \mfn(K_2) = b$ and attain perfect performance in $C_2$. Thus, when $\kept = \{C_2\}$, the expected team loss is determined by the human's loss in $C_1$, which is constant. Moreover, for all $(a,b)$, this is strictly better than the human acting alone, so $\kept = \emptyset$ is never optimal.\footnote{In general, one can always design a delegate that the human should use in at least one category, as long as the machine has at least one feature. This is because anytime the machine only retains a single human category $C$, the corresponding machine function $\mfn^{\{C\}}(K)$ can effectively observe use both $C$ and $K$ to select an action, which is strictly more information than the human.} 

By \Cref{prop:dropisoptimal}, the expected team loss of an optimal delegate $\mfnopt$ is the minimum attainable loss amongst different choices of $\kept$, which we show in Appendix \ref{app:twofeatjustification} is \[\loss(\hfnopt, \mfnopt) = \frac{1}{8}\cdot \min\left\{1, (a-b)^2, a^2 + (b-1)^2\right\}.\]
The optimal set of retained categories $\kept^*$ is the one that produces the minimal value above. In Figure \ref{fig:twofeaturecategories} we plot $\loss(\hfnopt, \mfnopt)$ and $\kept^*$ 
of $\fnopt_{\fx, \fy}$ for each $\fx, \fy$. 

\begin{figure}[ht!]
    \centering
    \begin{subfigure}[t]{0.49 
    \linewidth} 
        \centering
        \includegraphics[width=\linewidth]{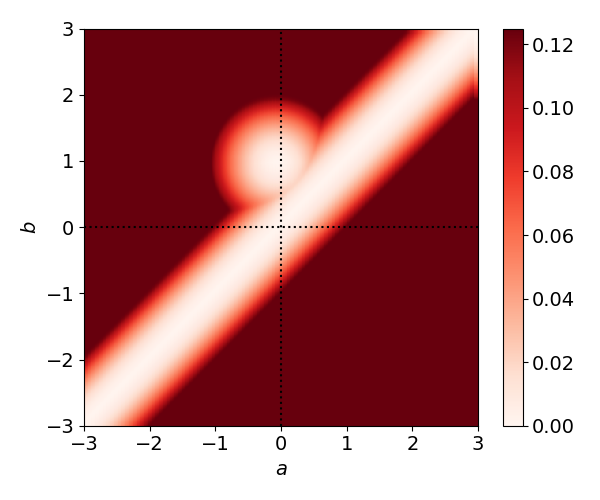}
        \caption{$\loss(\hfnopt, \mfnopt)$}
        \label{fig:twofeature:heatmap} 
        \Description{Heatmap showing the team loss of the optimal delegate for different values of $a, b$. There is a strip of low loss around the line $a = b$, and a sink of low loss around the point $(a,b) = (0,1)$.}
    \end{subfigure}
    \begin{subfigure}[t]{0.49\linewidth} 
        \centering
        \includegraphics[width=0.83\linewidth]{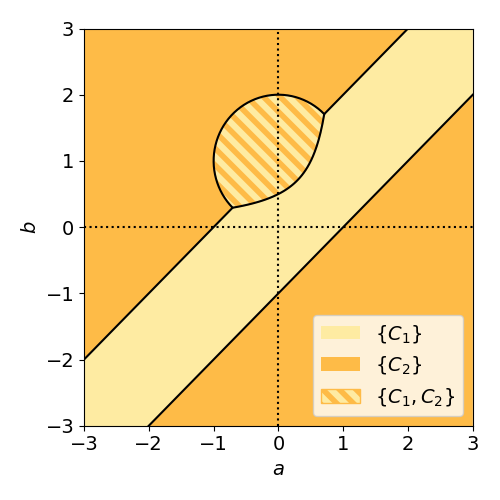}
        \caption{$\kept^*$}
        \Description{A set of three regions corresponding to different retained categories $\kept$. This plot visually resembles \ref{fig:twofeature:heatmap}.}
        \label{fig:twofeature:regions}
    \end{subfigure}
    \caption{For each pair $(a,b)$, we compute the optimal delegate $\mfnopt$ for ground truth function  $\fnopt_{\fx,\fy}$. In (a) we plot the expected loss of the human-machine team where the human delegates to $\mfnopt$, $\loss(\hfnopt, \mfnopt)$. In (b) we plot the categories $\kept^*$ that are retained by $\mfnopt$. By \Cref{prop:dropisoptimal} this corresponds exactly to the categories where $\mfnopt$ will be adopted; (b) can also be considered a plot of $\dcats(\hfnopt, \mfnopt)$.}
    \label{fig:twofeaturecategories} 
\end{figure}
Each region of Figure \ref{fig:twofeature:regions} corresponds to a qualitatively different delegate. In the two solid regions, the optimal delegate specializes to one category and lets the human handle the other one (in the lighter, inner strip the delegate specializes to $C_1$; in the darker, outer region the delegate specializes to $C_2$). In the striped ``teardrop'', the optimal delegate is a general-purpose machine designed to handle all categories. Despite the extremely simple setting in which each agent observes a single binary feature, the three regions have very unusual structures. The region corresponding to $\kept^* = \{C_1\}$ is non-convex, and the region corresponding to $\kept^* = \{C_2\}$ is not even connected. \Cref{prop:dropisoptimal} showed that to find the optimal delegate for ground truth function $\fnopt$ it is necessary to find the optimal set of retained categories: Figure \ref{fig:twofeature:regions} gives intuition that the relationship between the optimal set of retained categories and the ground truth function is complex even in the simplest of settings, so that finding the optimal delegate is a difficult endeavor. In the following section we formalize this intuition, showing that in general finding the optimal delegate is NP-hard.

We pause to compare the optimal machine to other potential machine designs in the context of the two-feature setting. 

\subsection{Oblivious and optimal delegation}\label{app:arbitrarilyworse}

\begin{figure}[h]
    \centering
    \begin{subfigure}[b]{0.49\linewidth}
        \centering
        \includegraphics[width=0.95\linewidth]{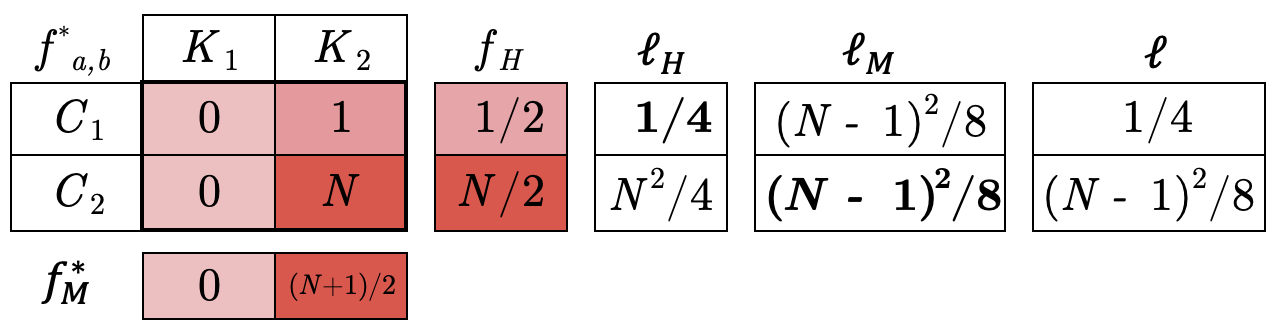}
        \caption{Oblivious delegate}
    \end{subfigure}
    \begin{subfigure}[b]{0.49\linewidth}
        \centering
        \includegraphics[width=0.95\linewidth]{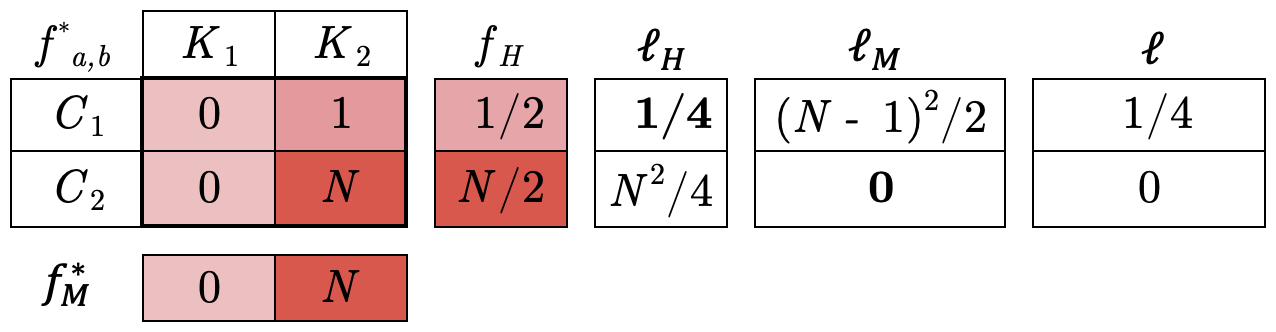}
        \caption{Optimal delegate}
    \end{subfigure}
    \caption{A family of delegation settings when the human and machine each observe a single distinct feature, in which the oblivious machine's team performance becomes arbitrarily worse than the optimal machine's team performance as $N$ grows large.}
    \Description{Grid resembling \ref{fig:twofeature:general} with $a = 0$, $b = N$.}
    \label{fig:twofeat:N}
\end{figure}

We first give an example of a family of delegation settings where the team performance of the oblivious machine $\mfnobliv$ can be \textit{arbitrarily} worse than the performance of the optimal machine $\mfnopt$. Consider the family of delegation settings shown in Figure \ref{fig:twofeat:N}. This corresponds to $(a,b) = (0, N)$; inspecting Figure \ref{fig:twofeaturecategories}, $(0, N)$ is in the region where $\kept^* = \{C_2\}$, and $\mfnopt$ achieves an expected team loss of $1/8$. Both the human and the oblivious machine struggle to select an action in state $\x_{22}$, since the optimal action $\fnopt(\x_{22})$ is very different from the other values in the same row and column. The human observes category $C_2$, and averages between $\fnopt(\x_{21}) = 0$ and $\fnopt(\x_{22}) = N$, which has a loss on the order of $N^2$ in state $\x_{22}$. The oblivious machine observes category $K_2$ and averages between $\fnopt(\x_{12}) = 1$ and $\fnopt(\x_{22}) = N$, which also has a loss on the order of $N^2$ in $\x_{22}$. Thus regardless of whether human delegates to the oblivious machine or not in category $C_2$, the expected team loss of the human and the oblivious machine will be on the order of $N^2$. As $N$ grows large, this is arbitrarily worse than the expected team loss of the human and the optimal machine, which is constant in $N$.

\subsection{Maximally adopted delegates}\label{app:maxadopt}

We now extend our analysis of the two-feature setting to compare the optimal delegate to the delegate that maximizes the likelihood of delegation, as might be implemented by a self-interested firm. In the two-feature setting, the designer will always be able to design a machine that is adopted in at least \textit{one} human category: he could simply design the optimal delegate. For which ground truth functions $\fnopt_{\fx, \fy}$ is it possible to design a delegate $\mfn$ that is adopted in \textit{both} categories, regardless of whether this has optimal team performance? 

\begin{figure}[ht]
    \centering
    \begin{subfigure}{0.49\linewidth} 
        \centering
        \includegraphics[width=0.83\linewidth]{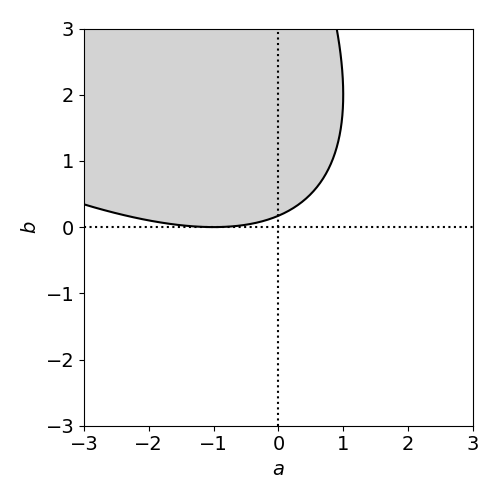}
        \caption{}
        \label{fig:twofeat:maxadopt:grey}
    \end{subfigure}
    \begin{subfigure}{0.49\linewidth} 
        \centering
        \includegraphics[width=0.83\linewidth]{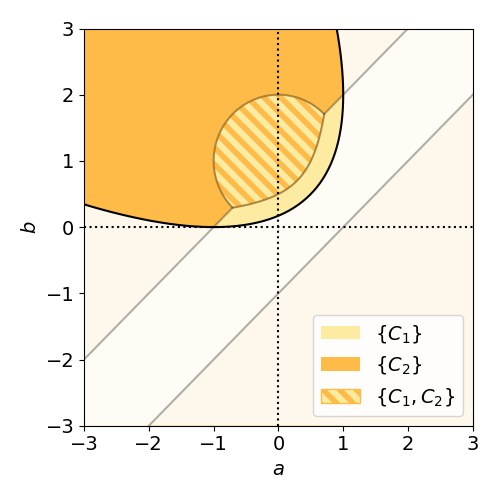}
        \caption{}
        \label{fig:twofeat:maxadopt:overlay}
    \end{subfigure}
    \caption{The grey shaded region of Figure \ref{fig:twofeat:maxadopt:grey} is the set of delegation settings $\fnopt_{\fx, \fy}$ where a firm can design a machine which is adopted in all categories. In Figure \ref{fig:twofeat:maxadopt:overlay}, we highlight this region in Figure \ref{fig:twofeature:regions}, which showed the regions where the optimal delegate will be adopted by the human user. The solid yellow regions within the highlighted set correspond to delegation settings where a firm seeking to maximize adoption will design a sub-optimal machine.}
    \label{fig:twofeat:maxadopt}
\end{figure}

In Appendix \ref{app:twofeatjustification:maxadopt}, we derive the set of $\fnopt_{a,b}$ such that it is \textit{possible} to design a machine that will be adopted in both $C_1$ and $C_2$; we plot the corresponding $(a,b)$ in Figure \ref{fig:twofeat:maxadopt}. Figure \ref{fig:twofeat:maxadopt:overlay} compares this region to the regions where the optimal delegate will be adopted from Figure \ref{fig:twofeature:regions}; we see that -- necessarily -- it contains the striped region where it is optimal to design a machine that will be adopted in both categories. However, it also slices through both regions where only the optimal delegate is only adopted in a single category: here, maximizing adoption comes at a cost to team performance. This indicates that delegates which maximize adoption are suboptimal for an \textit{unbounded} set of delegation settings. In Appendix \ref{app:twofeatjustification:maxadopt} we give intuition for the shape of this region. Moreover, we show that the delegate with maximal adoption can be \textit{arbitrarily} worse than the optimal delegate, and that for most settings the maximally adopted machine has worse team performance than even the oblivious machine acting alone. 

\section{Tractability}\label{sec:results}

We have established that it is desirable to design optimal delegates. In this section, we investigate whether it is tractable. We discover two large families of delegation settings in which it is possible to efficiently find an optimal delegate. However, we show that this is intractable in general. We provide proofs in Appendix \ref{app:tractproofs}.

We first pause to define efficiency. In general, in order to describe a delegation setting $\idxset{H}, \idxset{M}, \Prob, \fnopt$, one must specify $\Prob(\x)$ and $\fnopt(\x)$ for each state $\x$, so that the size of a delegation setting is the number of states $n$. An efficient algorithm is then one in which an optimal delegate may be found in time polynomial in $n$.

\subsection{Efficient algorithms}

We first show that if the delegation setting is separable, that is, $\fnopt$ can be decomposed additively into functions of the human and machine features respectively, and the probabilities $\Prob$ satisfy  independence and regularity conditions, we can efficiently find an optimal delegate. 

To formalize separability, we first need to define the consistency of a human category $C$ and a machine category $K$: suppose there are some shared features observed by both the human and machine, that is, $\idxset{S} = \idxset{H} \cap \idxset{M} \neq \emptyset$. Since a human can't distinguish between states in $C$, each state $\x\in C$ will have the same shared features $\x_S$. Likewise, all states in $\z \in K$ will have the same $\z_S$. Then $C$ and $K$ are inconsistent if the shared features are different the two categories -- in this case, it would be impossible for the human to be in category $C$ and the machine to be in category $K$ at the same time. Formally, $C$ and $K$ are inconsistent if for $\x \in C$, $\z \in K$, we have $\x_S \neq \z_S$.

Now, a \textit{function} $\fnopt$ is separable with respect to observable feature sets $\idxset{H}$ and $\idxset{M}$ if there are some functions $u, w$ such that $\fnopt(\x) = u(C(\x)) + w(K(\x))$. A \textit{distribution} $\Prob$ is separable if all consistent human and machine categories are independent, that is, there exist rationals $p_i, q_j$ such that $\Prob(C_i \cap K_j) = p_i \cdot q_j$ for all consistent $C_i, K_j$. A \textit{delegation setting} $\idxset{H}, \idxset{M}, \Prob, \fnopt$ is separable if $\fnopt$ is separable with respect to $\idxset{H}, \idxset{M}$ and $\Prob$ is separable.  Notably, if $\fnopt$ is linear in $\x$, it is separable for all $\idxset{H}, \idxset{M}$. Moreover, if the features take values independently, $\Prob$ will be separable.

We will also assume that the probabilities $p_i$ have \textit{polynomial precision}, that is, $p_i = t_i / T$ for $t_i, T \in \mathbb{N}$ and $T \in O(\text{poly}(n))$.

\begin{restatable}{theorem}{thmseparable}\label{thm:separable}
    If the delegation setting is separable with polynomial precision probabilities, we can find an optimal delegate $\mfnopt$ in time polynomial in $n$.
\end{restatable}
\begin{proof}[Proof sketch]
    The simplest case of this problem is when the human and machine partition the set of all features, and when each state occurs with equal probability. In this case, we can write $v_{ij} = \fnopt(\x_{ij}) = u_i + w_j$ for each $i, j$, for some $u_i, w_j \in \mathbb{R}$. By \Cref{cor:variancegame:simple}, we need to find some set $R$ solving \begin{align*}\min_R \sum_{i \notin R} \modvar(v_{ij} | i) + \sum_j \modvar(v_{ij}| i\in R, j).
    \end{align*} In the separable case, we may simplify this to \[\min_k \left[\left(1 - \frac{k}{|\hcats|}\right) \Var(w_j | i)+ \frac{k}{|\hcats|} \min_{R: |R| = k} \Var(u_i | i \in R, j)\right].\] The problem is now simply to find the minimum variance subset of $\{u_i\}$ of size $k$ for each $1 \leq k \leq |\hcats|$, where $|\hcats| = O(n)$. We can do this efficiently by sorting the values of $u_i$ and examining the variance of \textit{contiguous} subsets of size $k$ \citep{stackoverflowminvar}.
    We may reduce more complicated versions of the problem with arbitrary human- and machine-observable features and arbitrary polynomial-precision probabilities $p, q$ to this setting.
\end{proof}

We also find that if the human or machine has access to a limited number of features, the problem is again tractable.

\begin{restatable}{theorem}{thmconstant}\label{thm:constant}
    Suppose that $|\idxset{H} \setminus \idxset{M}| = O(1)$ or $|\idxset{M} \setminus \idxset{H}|= O(1)$. Then we may find an optimal delegate in time polynomial in $n$.
\end{restatable}

Finally, in any given delegation setting we may efficiently determine whether it is possible to design a machine so that the human-machine team has zero loss.
\begin{restatable}{proposition}{propperfect}\label{prop:perfectteams}
    For any delegation setting, we may determine whether there exists a function $\mfnopt$ such that $\loss(\hfnopt, \mfnopt) = 0$ in time linear in $n$.
\end{restatable}

\subsection{Intractability}

The problem of finding an optimal delegate is NP-hard in general.

\begin{restatable}{theorem}{thmnp}\label{thm:np}
    Unless $\textnormal{P} = \textnormal{NP}$, there is no algorithm to find an optimal delegate $\mfnopt$ in time polynomial in $n$ for all delegation settings.
\end{restatable}

\begin{proof}[Proof sketch]
    We first consider the problem \textsc{NegRegularDSD} of finding a maximum density subgraph of a graph with weighted edges, where the weights may be negative but satisfy some regularity conditions. We reduce the NP-hard problem of finding a maximal clique in a regular graph to \textsc{NegRegularDSD}. We then reduce \textsc{NegRegularDSD} to \textsc{VarianceAssignment}
    \varXiv{by constructing an instance $\{v_{ij}\}$ of \textsc{VarianceAssignment} where rows $i$ correspond to nodes, columns $j$ correspond to edges, and the values $v_{ij}$ and $v_{kj}$ are functions of the weight on edge $j = (i,k)$.}{.}
\end{proof}
This result motivates why the optimal delegate has not previously been characterized: the problem has a fundamentally combinatorial nature that makes it intractable to solve.

\section{Computational experiments with iterative design}\label{sec:simulations}

In this section, we explore the implications of an alternative design pipeline: instead of designing the optimal delegate, the designer might iteratively update the delegate depending on when it is adopted. Concretely, this process is as follows.
\begin{enumerate}[noitemsep]
    \item The designer designs and deploys an oblivious delegate. 
    \item The designer observes the categories where the human delegates to this machine.
    \item The designer deploys a new machine that is optimized for these categories.
    \item The designer repeats steps (2) and (3) until convergence.
\end{enumerate} 
If we think back to our earlier example of an online shopping agent, this would correspond to a designer initially deploying a general-purpose agent, which seeks to assist with purchases of all kinds, but then observing that the agent turns out to only be used by people for a particular kind of purchases --- say for travel-related purchases. The designer then re-optimizes the agent for travel purchasing, and continues this cycle until the machine is designed for the settings where it's actually used. 

Formally, suppose the designer initially designs the oblivious machine $\mfnobliv$. The human then delegates to $\mfnobliv$ in categories $\dcats(\hfnopt, \mfnobliv)$. The designer can improve the team performance with \[\mfn^{(1)} = \mfn^{\dcats(\hfnopt, \mfnobliv)}.\] The human adopts $\mfn^{(1)}$ in categories $\dcats(\hfnopt, \mfn^{(1)})$, achieving even better team performance. The designer iteratively takes \[\mfn^{(t+1)} = \mfn^{\dcats(\hfnopt, \mfn^{(t)})}.\] Since there are a finite number of subsets of human categories $\dcats$, and the performance improves after each iteration, this process will eventually converge to some fixed point $\mfn^{\textnormal{iter}}$.\footnote{To be precise: when the human chooses to delegate to $\mfn^{(t+1)}$ in categories $\dcats(\hfnopt, \mfn^{(t+1)})$ rather than $\dcats(\hfnopt, \mfn^{(t)})$, by definition of $\dcats$ this must either \textit{strictly} improve the loss, or satisfy $\dcats(\hfnopt, \mfn^{(t+1)}) \subseteq \dcats(\hfnopt, \mfn^{(t)})$. The re-optimization step also weakly improves the loss. Thus, it is impossible to visit the same $\dcats$ more than once in non-consecutive iterations.} 

In order to carry out this iteration, the designer only needs to know the loss-minimizing action within each machine category, restricting to states where $\mfn^{(t)}$ is adopted -- formally, $\E[\fnopt | X(\dcats(\hfnopt, \mfn^{(t)}))\cap K]$ for each machine category $K$ -- rather than needing to know the full delegation setting $\idxset{H}, \idxset{M}, \Prob, \fnopt$. Notably, this iterative design process can arise naturally through training: roughly, if the machine is iteratively retrained on data from past interactions at long enough intervals, the distribution of past interactions will shift to categories in which the machine is adopted.

We give an example in Figure \ref{fig:iteration} where the iterative solution $\mfniter$ is significantly different from the optimal delegate $\mfnopt$: $\mfniter = (0.3, 0.85)$ takes a small action in $K_1$ and a large action in $K_2$, while $\mfnopt =  (0.9, 0.1)$ has the opposite design: it takes a large action in $K_1$ and a small action in $K_2$.

\begin{figure*}[h!]
    \centering
    \includegraphics[width=\linewidth]{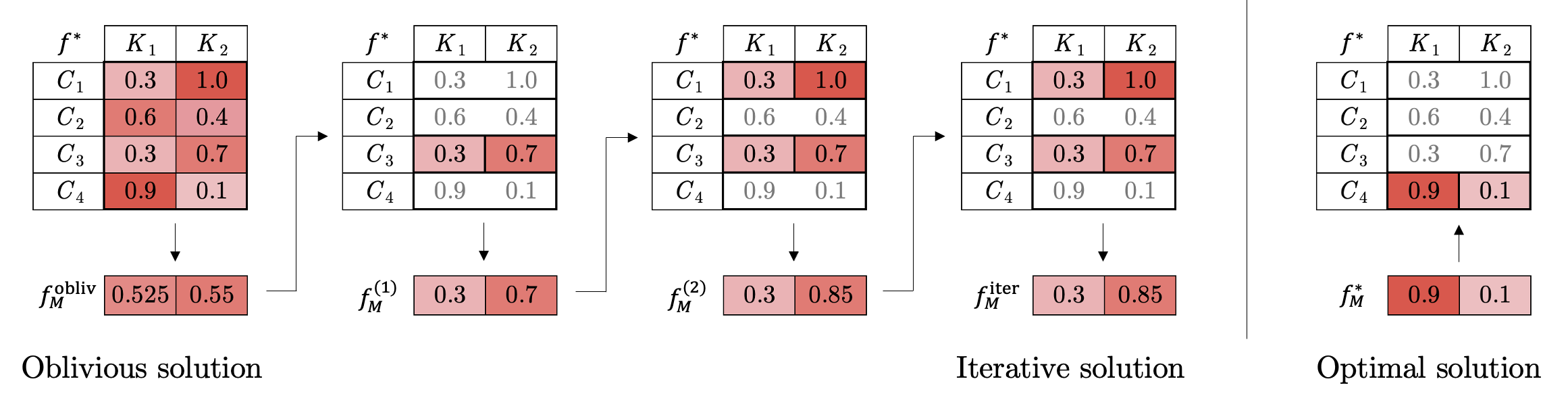}
    \caption{Illustration of iterative vs. optimal machine design. In the delegation setting shown, the human observes two binary features (inducing four human categories $C_1$ to $C_4$), the machine observes one distinct binary feature (inducing two machine categories $K_1, K_2$), and each feature takes value 0 or 1 with equal probability. There is then a single state $\x_{ij}$ in $C_i \cap K_j$, and the value in grid entry $(i,j)$ is $\fnopt(\x_{ij})$. The sequence of machines $\mfn$ on the left shows the process of iteration: a designer first builds a machine $\mfnobliv$ to be used in all human  categories. This machine is only used in category $C_3$, so the machine is updated to $\mfn^{(1)}$, which  performs optimally in $C_3$. Now, $\mfn^{(1)}$ is used in $\{C_1, C_3\}$, so the machine is updated to retain only those categories, yielding $\mfn^{(2)}$. Finally, $\mfn^{(2)}$ is actually used in $\{C_1, C_3\}$, so the process ends with $\mfniter = \mfn^{(2)}$. By contrast, the optimal machine $\mfnopt$ retains only $C_4$, and takes nearly opposite actions as $\mfniter$.}
    \Description{On the left, a sequence of $\fnopt$ and $\mfn^{(i)}$ that converges to $\mfniter = (0.3, 0.85)$ which is delegated to in $\{C_1, C_3\}$. On the right, the optimal machine $\mfnopt = (0.9, 0.1)$ which is delegated to in $C_4$ only.}
    \label{fig:iteration}
\end{figure*}
\varXiv{}{In Appendix \ref{app:formaliteration} we formalize this process and argue that it always converges to some machine $\mfniter$ in a finite number of steps. We also show an example where the iterative solution is significantly different from the optimal delegate.
}

We may now compare the performance of $\mfniter$ and $\mfnopt$. In Figure \ref{fig:lossratios} of Appendix \ref{app:additionalexperiments} we show the result of computational experiments that investigate how far from optimal the iterative solution's performance is in the case of linear functions with independent features. As the number of features increases, the iterative solutions achieve near-optimal performance. 

We run an additional experiment in the same setting as Figure \ref{fig:lossratios} to determine the proportion of times the iterative process finds the optimal solution for several different numbers of observed features $|I_H|$ and $|I_M|$. We show the results of this in Table \ref{tab:props} in Appendix \ref{app:additionalexperiments}. We find that the proportion \textit{decreases monotonically} in the number of features the human observes -- which makes sense, as this corresponds to strictly more complex options for the retained categories $\kept$ -- and \textit{increases monotonically} in the number of features the machine observes. As expected from Figure \ref{fig:lossratios}, the proportion of optimal delegates find goes up along the diagonal $|I_H| = |I_M|,$ but this increase is very mild, indicating that many of the iterative solutions that have near-optimal performance are simply very good alternatives to the optimal delegate.

\section{Discussion}\label{sec:discussion}

In this paper, we developed a formal model for settings where a human decides whether to delegate to a machine, and where categories arise from incomplete information. 
We studied the problem of designing an optimal machine for delegation in the presence of categories. We showed that this induces surprisingly clean algorithmic formulations and derived tractability results.

\subsection{General relationships between teams}\label{sec:discussion:losscomp}

Our motivation for studying the design of optimal delegates is that the optimal delegate $\mfnopt$  is always weakly better than the oblivious machine $\mfnobliv$, and may be arbitrarily better. Additional general relationships exist between different human-algorithm teams.

Define $\loss(\fn_A)$ as the loss of agent $A$ acting individually. Any team with the human optimally delegating between $\hfn$ and some machine $\mfn$ will always be better than $\hfn$ and $\mfn$ acting individually, because the human could choose to always delegate (thus emulating $\mfn$), or never delegate (thus emulating $\hfn$), but will actually choose the better agent in each category. Formally, $\loss(\hfn, \mfn) \leq \loss(\hfn)$ and $\loss(\hfn, \mfn) \leq \loss(\mfn)$. 
\varXiv{This produces the arrows from the standalone models to the model teams.}{}

Moreover, $\loss(\hfnopt, \mfnopt) \leq \loss(\hfnopt, \mfn)$ for any machine $\mfn$; in particular, delegating to the oblivious machine will have worse loss than delegating to the optimal delegate $\mfnopt$, i.e., $\loss(\hfnopt, \mfnopt) \leq \loss(\hfnopt, \mfnobliv)$. 
\varXiv{This results in an arrow from $\hfnopt \land \mfnobliv$ to $\hfnopt \land \mfnopt$.}{}
However, without delegation, the oblivious machine will perform better than the optimal delegate, $\loss(\mfnobliv) \leq \loss(\mfnopt)$: the oblivious machine is defined to be the best machine without delegation. 
\varXiv{This results in the final arrow from $\mfnopt$ to $\mfnobliv$.}{}

We illustrate these relationships in Figure \ref{fig:losscomparison}. Of particular note is that not only is $\loss(\hfnopt, \mfnopt) \leq \loss(\hfnopt, \mfnobliv)$, but it is also always true that $\loss(\mfnobliv) \leq \loss(\mfnopt)$, where $\loss(\mfn)$ is the loss of the machine acting individually. These two relationships echo the results of \citet{bansal2021} and \citet{hamade2024chesshandoff}, who show in other human-algorithm collaboration settings that the optimal individual algorithmic agent is a worse collaborator than an algorithmic agent designed for collaboration.\begin{figure}[ht]
        \centering
        \begin{tikzpicture}
            \node[shape=rectangle] (N) at (0,2) {$\mfnobliv$};
            \node[shape=rectangle] (H) at (2,2) {$\hfnopt$};
            \node[shape=rectangle] (O) at (4,2) {$\mfnopt$};
            \node[shape=rectangle] (HN) at (1,1) {$\hfnopt \land\mfnobliv$};
            \node[shape=rectangle] (HO) at (3,0) {$\hfnopt\land \mfnopt$};
            \path [->] (O) edge [bend right] node [left] {} (N); 
            \path [->] (O) edge (HO); 
            \path [->] (N) edge (HN); 
            \path [->] (H) edge (HN); 
            \path [->] (H) edge (HO); 
            \path [->] (HN) edge (HO); 
        \end{tikzpicture}            
        \caption{Relationships between the losses of different possible human and machine teams. 
         $\hfn\land\mfn$ denotes a human with function $\hfn$ optimally delegating to the machine $\mfn$; an arrow from team A to team B indicates that the loss of team B will always be (weakly) lower than team A.} 
         \Description{A set of arrows between different teams; justification for each arrow is described in Appendix \ref{app:losscomparisons}.}
        \label{fig:losscomparison}
    \end{figure}
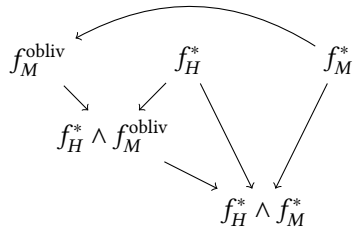

\subsection{Extensions}\label{sec:extensions}

We now show that our model and results extend to cases where categories are arbitrary partitions of states, including those that arise behaviorally. We also show how our results extend to cases where the machine's designer can choose the best \textit{fixed-size subset} of features. We then demonstrate how our model may be simply extended to describe other aspects of decision-making such as human behavior, but leave extending our results to future work. Finally, we discuss the limitations of our results for providing insights into machine-to-human delegation.

\paragraph{Extending our results} While our model showed categories arising from agents perceiving incomplete information, all of our theoretical results extend to the case where categories are arbitrary partitions of the world; we provide details in Appendix \ref{app:arbitrarycategories}.

Another natural extension of our model is where instead of a fixed set of features the machine can observe or interpret, the machine's designer can select the best subset of features from a set of features available to the machine, where the subset has some size constraint. In this case, our efficient algorithms may be extended, as there are at most $O(2^d) = O(n)$ such subsets of features, so adding an outer loop of brute-force feature selection only adds an multiplicative factor of $n$. The optimal design problem with feature selection is still hard: if maximum number of features the machine can use is the number of features available to it, we arrive at our original problem.

\paragraph{Modeling other aspects of decision-making} While we model the human's decision to delegate as a rational comparison between the machine's performance and her own in a given category in order to isolate the effect of categories, in reality it has been well-established that humans prefer to retain control of decisions \cite{owens2014}, may be biased toward or against algorithmic agents \cite{dietvorst2014,logg2019}, or exhibit other behavioral biases. These behavioral biases may be modeled by changing the human's decision to delegate from simply delegating if and only if the machine is better, to some behavioral policy that depends on the machine's average action or loss in each category. For example, to model a desire for control, one could model the human as only adopting the machine if the machine's average loss exceeds the human's by some factor. If the human is averse to unexpected outputs, one could add a penalty to the machine's loss that is a function of the average difference between the machine's action and the human's action within each human category, and so on.
It is also reasonable that the human and machine might have different levels of skill at different tasks, which we could model with additive or multiplicative factors on the machine and human losses.
While all of these aspects are easy to include in the \textit{model}, our \textit{results} do not necessarily extend in these more general settings.

\paragraph{Other delegation settings} In some cases, rather than the human delegating to the algorithmic agent, the algorithmic agent has default control over the choice of action but may sometimes delegate to the human. This has been studied in the literature on learning to defer \cite{madras2018,lykouris2024}. In our model there is no intrinsic requirement that the delegator be a human and the delegate be an algorithm. However, our results still speak to the optimal design of the \textit{delegate}, and therefore when the roles of human and algorithm are swapped, our results address how a human should be trained to work with a machine. Similarly, our model could describe humans delegating to other humans in cases where communication is impossible, such as an emergency room physician referring a patient to a specialist where the two doctors share the patient's medical history but cannot necessarily share intuition that might inform a diagnosis.

\subsection{Further directions}

An interesting direction for future work is determining other delegation settings in which finding an optimal delegate is tractable. Moreover, generalizing our analysis in Section \ref{sec:twofeature} to investigate which ground truth optimal action functions $\fnopt$ lead to qualitatively different optimal delegates could yield heuristics for intractable instances. It would also be interesting to extend our computational experiments on iterative design to non-linear functions and prove bounds on the performance of the iterative solution. It is also important to investigate designing optimal delegates when our model is extended to include behavioral factors beyond categories or different levels of skill as described above, or when the team performance is measured with a different loss.

\begin{acks}

The authors would like to thank Erica Chiang, Kate Donahue, Greg d'Eon, Nikhil Garg, Malte Jung, Rajiv Movva, Kenny Peng, and Emma Pierson for valuable discussions and feedback, as well as the Artificial Intelligence, Policy, and Practice (AIPP) group at Cornell. The authors also thank the reviewers of the NeurIPS 2024 Workshop on Behavioral Machine Learning and the reviewers of EC'25 for helpful suggestions. 

SG is supported by a fellowship from the Cornell University Department of Computer Science and an NSERC PGS-D fellowship [587665]. The work is supported in part by a grant from the John D. and Catherine T. MacArthur Foundation.
\end{acks}

\bibliographystyle{ACM-Reference-Format}
\bibliography{bibliography}


\begin{thebibliography}{53}


\ifx \showCODEN    \undefined \def \showCODEN     #1{\unskip}     \fi
\ifx \showDOI      \undefined \def \showDOI       #1{#1}\fi
\ifx \showISBNx    \undefined \def \showISBNx     #1{\unskip}     \fi
\ifx \showISBNxiii \undefined \def \showISBNxiii  #1{\unskip}     \fi
\ifx \showISSN     \undefined \def \showISSN      #1{\unskip}     \fi
\ifx \showLCCN     \undefined \def \showLCCN      #1{\unskip}     \fi
\ifx \shownote     \undefined \def \shownote      #1{#1}          \fi
\ifx \showarticletitle \undefined \def \showarticletitle #1{#1}   \fi
\ifx \showURL      \undefined \def \showURL       {\relax}        \fi
\providecommand\bibfield[2]{#2}
\providecommand\bibinfo[2]{#2}
\providecommand\natexlab[1]{#1}
\providecommand\showeprint[2][]{arXiv:#2}

\bibitem[Adler et~al\mbox{.}(2023)]%
        {oai2023}
\bibfield{author}{\bibinfo{person}{Steven Adler}, \bibinfo{person}{Cullen O’Keefe}, \bibinfo{person}{Rosie Campbell}, \bibinfo{person}{Teddy Lee}, \bibinfo{person}{Pamela Mishkin}, \bibinfo{person}{Tyna Eloundou}, \bibinfo{person}{Alan Hickey}, \bibinfo{person}{Katarina Slama}, \bibinfo{person}{Lama Ahmad}, \bibinfo{person}{Paul McMillan}, \bibinfo{person}{Andrea Vallone}, \bibinfo{person}{Alexandre Passos}, {and} \bibinfo{person}{David~G. Robinson}.} \bibinfo{year}{2023}\natexlab{}.
\newblock \showarticletitle{Practices for Governing Agentic AI Systems}.
\newblock  (\bibinfo{year}{2023}).
\newblock


\bibitem[Akinola et~al\mbox{.}(2017)]%
        {akinola2017}
\bibfield{author}{\bibinfo{person}{Modupe Akinola}, \bibinfo{person}{Ashley~E. Martin}, {and} \bibinfo{person}{Katherine~W. Phillips}.} \bibinfo{year}{2017}\natexlab{}.
\newblock \showarticletitle{To Delegate or Not to Delegate: Gender Differences in Affective Associations and Behavioral Responses to Delegation}.
\newblock \bibinfo{journal}{\emph{Academy of Management Journal}} (\bibinfo{year}{2017}).
\newblock
\urldef\tempurl%
\url{https://psycnet.apa.org/record/2018-44226-011}
\showURL{%
\tempurl}


\bibitem[Alur et~al\mbox{.}(2024)]%
        {alur2024}
\bibfield{author}{\bibinfo{person}{Rohan Alur}, \bibinfo{person}{Manish Raghavan}, {and} \bibinfo{person}{Devavrat Shah}.} \bibinfo{year}{2024}\natexlab{}.
\newblock \showarticletitle{Human Expertise in Algorithmic Prediction}. In \bibinfo{booktitle}{\emph{Advances in Neural Information Processing Systems}}, \bibfield{editor}{\bibinfo{person}{A.~Globerson}, \bibinfo{person}{L.~Mackey}, \bibinfo{person}{D.~Belgrave}, \bibinfo{person}{A.~Fan}, \bibinfo{person}{U.~Paquet}, \bibinfo{person}{J.~Tomczak}, {and} \bibinfo{person}{C.~Zhang}} (Eds.), Vol.~\bibinfo{volume}{37}. \bibinfo{publisher}{Curran Associates, Inc.}, \bibinfo{pages}{138088--138129}.
\newblock
\urldef\tempurl%
\url{https://proceedings.neurips.cc/paper_files/paper/2024/file/f9743a7bfee6592e3ff913ffadd8a857-Paper-Conference.pdf}
\showURL{%
\tempurl}


\bibitem[Anderson(1991)]%
        {anderson1991}
\bibfield{author}{\bibinfo{person}{John~R. Anderson}.} \bibinfo{year}{1991}\natexlab{}.
\newblock \showarticletitle{The Adaptive Nature of Human Categorization}.
\newblock \bibinfo{journal}{\emph{Psychological Review}} \bibinfo{volume}{98}, \bibinfo{number}{3} (\bibinfo{year}{1991}), \bibinfo{pages}{409--429}.
\newblock
\urldef\tempurl%
\url{https://doi.org/10.1037/0033-295x.98.3.409}
\showDOI{\tempurl}


\bibitem[Ashby and Maddox(2005)]%
        {ashby2005}
\bibfield{author}{\bibinfo{person}{F. Ashby} {and} \bibinfo{person}{W. Maddox}.} \bibinfo{year}{2005}\natexlab{}.
\newblock \showarticletitle{Human Category Learning}.
\newblock \bibinfo{journal}{\emph{Annual review of psychology}}  \bibinfo{volume}{56} (\bibinfo{date}{02} \bibinfo{year}{2005}), \bibinfo{pages}{149--78}.
\newblock
\urldef\tempurl%
\url{https://doi.org/10.1146/annurev.psych.56.091103.070217}
\showDOI{\tempurl}


\bibitem[Bansal et~al\mbox{.}(2021)]%
        {bansal2021}
\bibfield{author}{\bibinfo{person}{Gagan Bansal}, \bibinfo{person}{Besmira Nushi}, \bibinfo{person}{Ece Kamar}, \bibinfo{person}{Eric Horvitz}, {and} \bibinfo{person}{Daniel~S. Weld}.} \bibinfo{year}{2021}\natexlab{}.
\newblock \showarticletitle{Is the Most Accurate {AI} the Best Teammate? {O}ptimizing {AI} for Teamwork}.
\newblock \bibinfo{journal}{\emph{Proceedings of the AAAI Conference on Artificial Intelligence}} \bibinfo{volume}{35}, \bibinfo{number}{13} (\bibinfo{date}{May} \bibinfo{year}{2021}), \bibinfo{pages}{11405--11414}.
\newblock
\urldef\tempurl%
\url{https://doi.org/10.1609/aaai.v35i13.17359}
\showDOI{\tempurl}


\bibitem[Bansal et~al\mbox{.}(2019a)]%
        {bansal2019a}
\bibfield{author}{\bibinfo{person}{Gagan Bansal}, \bibinfo{person}{Besmira Nushi}, \bibinfo{person}{Ece Kamar}, \bibinfo{person}{Walter~S. Lasecki}, \bibinfo{person}{Daniel~S. Weld}, {and} \bibinfo{person}{Eric Horvitz}.} \bibinfo{year}{2019}\natexlab{a}.
\newblock \showarticletitle{Beyond Accuracy: The Role of Mental Models in Human-{AI} Team Performance}.
\newblock \bibinfo{journal}{\emph{Proceedings of the AAAI Conference on Human Computation and Crowdsourcing}} \bibinfo{volume}{7}, \bibinfo{number}{1} (\bibinfo{date}{Oct.} \bibinfo{year}{2019}), \bibinfo{pages}{2--11}.
\newblock
\urldef\tempurl%
\url{https://doi.org/10.1609/hcomp.v7i1.5285}
\showDOI{\tempurl}


\bibitem[Bansal et~al\mbox{.}(2019b)]%
        {bansal2019b}
\bibfield{author}{\bibinfo{person}{Gagan Bansal}, \bibinfo{person}{Besmira Nushi}, \bibinfo{person}{Ece Kamar}, \bibinfo{person}{Daniel~S. Weld}, \bibinfo{person}{Walter~S. Lasecki}, {and} \bibinfo{person}{Eric Horvitz}.} \bibinfo{year}{2019}\natexlab{b}.
\newblock \showarticletitle{Updates in Human-{AI} Teams: Understanding and Addressing the Performance/Compatibility Tradeoff}.
\newblock \bibinfo{journal}{\emph{Proceedings of the AAAI Conference on Artificial Intelligence}} \bibinfo{volume}{33}, \bibinfo{number}{01} (\bibinfo{date}{Jul.} \bibinfo{year}{2019}), \bibinfo{pages}{2429--2437}.
\newblock
\urldef\tempurl%
\url{https://doi.org/10.1609/aaai.v33i01.33012429}
\showDOI{\tempurl}


\bibitem[Bobadilla-Suarez et~al\mbox{.}(2017)]%
        {bobadillasuarez2017}
\bibfield{author}{\bibinfo{person}{Sebastian Bobadilla-Suarez}, \bibinfo{person}{Cass~Robert Sunstein}, {and} \bibinfo{person}{Tali Sharot}.} \bibinfo{year}{2017}\natexlab{}.
\newblock \showarticletitle{The intrinsic value of choice: The propensity to under-delegate in the face of potential gains and losses}.
\newblock \bibinfo{journal}{\emph{Journal of Risk and Uncertainty}}  \bibinfo{volume}{54} (\bibinfo{year}{2017}), \bibinfo{pages}{187 -- 202}.
\newblock
\urldef\tempurl%
\url{https://doi.org/10.1007/s11166-017-9259-x}
\showDOI{\tempurl}


\bibitem[{Booked AI}(2024)]%
        {bookedaitool}
\bibfield{author}{\bibinfo{person}{{Booked AI}}.} \bibinfo{year}{2024}\natexlab{}.
\newblock \bibinfo{booktitle}{\emph{How {AI} Travel Can Curate Your Perfect Wine Tour in {I}taly}}.
\newblock
\urldef\tempurl%
\url{https://www.booked.ai/blogs/how-ai-travel-can-curate-your-perfect-wine-tour-in-italy}
\showURL{%
\tempurl}


\bibitem[Brandes et~al\mbox{.}(2016)]%
        {brandes2016}
\bibfield{author}{\bibinfo{person}{Ulrik Brandes}, \bibinfo{person}{Eugenia Holm}, {and} \bibinfo{person}{Andreas Karrenbauer}.} \bibinfo{year}{2016}\natexlab{}.
\newblock \showarticletitle{Cliques in Regular Graphs and the Core-Periphery Problem in Social Networks}. In \bibinfo{booktitle}{\emph{Combinatorial Optimization and Applications}}, \bibfield{editor}{\bibinfo{person}{T-H.~Hubert Chan}, \bibinfo{person}{Minming Li}, {and} \bibinfo{person}{Lusheng Wang}} (Eds.). \bibinfo{publisher}{Springer International Publishing}, \bibinfo{address}{Cham}, \bibinfo{pages}{175--186}.
\newblock
\showISBNx{978-3-319-48749-6}


\bibitem[Candrian and Scherer(2022)]%
        {candrianscherer2022}
\bibfield{author}{\bibinfo{person}{Cindy Candrian} {and} \bibinfo{person}{Anne Scherer}.} \bibinfo{year}{2022}\natexlab{}.
\newblock \showarticletitle{Rise of the machines: Delegating decisions to autonomous AI}.
\newblock \bibinfo{journal}{\emph{Computers in Human Behavior}}  \bibinfo{volume}{134} (\bibinfo{year}{2022}), \bibinfo{pages}{107308}.
\newblock
\showISSN{0747-5632}
\urldef\tempurl%
\url{https://doi.org/10.1016/j.chb.2022.107308}
\showDOI{\tempurl}


\bibitem[Carvalho et~al\mbox{.}(2019)]%
        {carvalho2019}
\bibfield{author}{\bibinfo{person}{Diogo~V. Carvalho}, \bibinfo{person}{Eduardo~M. Pereira}, {and} \bibinfo{person}{Jaime~S. Cardoso}.} \bibinfo{year}{2019}\natexlab{}.
\newblock \showarticletitle{Machine Learning Interpretability: A Survey on Methods and Metrics}.
\newblock \bibinfo{journal}{\emph{Electronics}} \bibinfo{volume}{8}, \bibinfo{number}{8} (\bibinfo{year}{2019}).
\newblock
\showISSN{2079-9292}
\urldef\tempurl%
\url{https://doi.org/10.3390/electronics8080832}
\showDOI{\tempurl}


\bibitem[Chan et~al\mbox{.}(2023)]%
        {chan2023}
\bibfield{author}{\bibinfo{person}{Alan Chan}, \bibinfo{person}{Rebecca Salganik}, \bibinfo{person}{Alva Markelius}, \bibinfo{person}{Chris Pang}, \bibinfo{person}{Nitarshan Rajkumar}, \bibinfo{person}{Dmitrii Krasheninnikov}, \bibinfo{person}{Lauro Langosco}, \bibinfo{person}{Zhonghao He}, \bibinfo{person}{Yawen Duan}, \bibinfo{person}{Micah Carroll}, \bibinfo{person}{Michelle Lin}, \bibinfo{person}{Alex Mayhew}, \bibinfo{person}{Katherine Collins}, \bibinfo{person}{Maryam Molamohammadi}, \bibinfo{person}{John Burden}, \bibinfo{person}{Wanru Zhao}, \bibinfo{person}{Shalaleh Rismani}, \bibinfo{person}{Konstantinos Voudouris}, \bibinfo{person}{Umang Bhatt}, \bibinfo{person}{Adrian Weller}, \bibinfo{person}{David Krueger}, {and} \bibinfo{person}{Tegan Maharaj}.} \bibinfo{year}{2023}\natexlab{}.
\newblock \showarticletitle{Harms from Increasingly Agentic Algorithmic Systems}. In \bibinfo{booktitle}{\emph{Proceedings of the 2023 ACM Conference on Fairness, Accountability, and Transparency}} (Chicago, IL, USA) \emph{(\bibinfo{series}{FAccT '23})}. \bibinfo{publisher}{Association for Computing Machinery}, \bibinfo{address}{New York, NY, USA}, \bibinfo{pages}{651–666}.
\newblock
\showISBNx{9798400701924}
\urldef\tempurl%
\url{https://doi.org/10.1145/3593013.3594033}
\showDOI{\tempurl}


\bibitem[Chazelle et~al\mbox{.}(1991)]%
        {chazelle1991}
\bibfield{author}{\bibinfo{person}{Bernard Chazelle}, \bibinfo{person}{Herbert Edelsbrunner}, \bibinfo{person}{Leonidas~J. Guibas}, {and} \bibinfo{person}{Micha Sharir}.} \bibinfo{year}{1991}\natexlab{}.
\newblock \showarticletitle{A singly exponential stratification scheme for real semi-algebraic varieties and its applications}.
\newblock \bibinfo{journal}{\emph{Theoretical Computer Science}} \bibinfo{volume}{84}, \bibinfo{number}{1} (\bibinfo{year}{1991}), \bibinfo{pages}{77--105}.
\newblock
\showISSN{0304-3975}
\urldef\tempurl%
\url{https://doi.org/10.1016/0304-3975(91)90261-Y}
\showDOI{\tempurl}


\bibitem[De~Toni et~al\mbox{.}(2024)]%
        {detoni2024}
\bibfield{author}{\bibinfo{person}{Giovanni De~Toni}, \bibinfo{person}{Nastaran Okati}, \bibinfo{person}{Suhas Thejaswi}, \bibinfo{person}{Eleni Straitouri}, {and} \bibinfo{person}{Manuel Gomez-Rodriguez}.} \bibinfo{year}{2024}\natexlab{}.
\newblock \showarticletitle{Towards Human-AI Complementarity with Prediction Sets}. In \bibinfo{booktitle}{\emph{Advances in Neural Information Processing Systems}}, \bibfield{editor}{\bibinfo{person}{A.~Globerson}, \bibinfo{person}{L.~Mackey}, \bibinfo{person}{D.~Belgrave}, \bibinfo{person}{A.~Fan}, \bibinfo{person}{U.~Paquet}, \bibinfo{person}{J.~Tomczak}, {and} \bibinfo{person}{C.~Zhang}} (Eds.), Vol.~\bibinfo{volume}{37}. \bibinfo{publisher}{Curran Associates, Inc.}, \bibinfo{pages}{31380--31409}.
\newblock
\urldef\tempurl%
\url{https://proceedings.neurips.cc/paper_files/paper/2024/file/37d4d4413b7c7558cc27a6d3d42ea998-Paper-Conference.pdf}
\showURL{%
\tempurl}


\bibitem[Deng et~al\mbox{.}(2023)]%
        {deng2023}
\bibfield{author}{\bibinfo{person}{Xiang Deng}, \bibinfo{person}{Yu Gu}, \bibinfo{person}{Boyuan Zheng}, \bibinfo{person}{Shijie Chen}, \bibinfo{person}{Sam Stevens}, \bibinfo{person}{Boshi Wang}, \bibinfo{person}{Huan Sun}, {and} \bibinfo{person}{Yu Su}.} \bibinfo{year}{2023}\natexlab{}.
\newblock \showarticletitle{Mind2Web: Towards a Generalist Agent for the Web}. In \bibinfo{booktitle}{\emph{Advances in Neural Information Processing Systems}}, \bibfield{editor}{\bibinfo{person}{A.~Oh}, \bibinfo{person}{T.~Naumann}, \bibinfo{person}{A.~Globerson}, \bibinfo{person}{K.~Saenko}, \bibinfo{person}{M.~Hardt}, {and} \bibinfo{person}{S.~Levine}} (Eds.), Vol.~\bibinfo{volume}{36}. \bibinfo{publisher}{Curran Associates, Inc.}, \bibinfo{pages}{28091--28114}.
\newblock
\urldef\tempurl%
\url{https://proceedings.neurips.cc/paper_files/paper/2023/file/5950bf290a1570ea401bf98882128160-Paper-Datasets_and_Benchmarks.pdf}
\showURL{%
\tempurl}


\bibitem[Dietvorst et~al\mbox{.}(2014)]%
        {dietvorst2014}
\bibfield{author}{\bibinfo{person}{Berkeley~J. Dietvorst}, \bibinfo{person}{Joseph~P. Simmons}, {and} \bibinfo{person}{Cade Massey}.} \bibinfo{year}{2014}\natexlab{}.
\newblock \showarticletitle{Algorithm Aversion: People Erroneously Avoid Algorithms after Seeing Them Err}.
\newblock \bibinfo{journal}{\emph{CSN: Business (Topic)}} (\bibinfo{year}{2014}).
\newblock
\urldef\tempurl%
\url{https://doi.org/10.1037/xge0000033}
\showDOI{\tempurl}


\bibitem[Drouin et~al\mbox{.}(2024)]%
        {drouin2024workarena}
\bibfield{author}{\bibinfo{person}{Alexandre Drouin}, \bibinfo{person}{Maxime Gasse}, \bibinfo{person}{Massimo Caccia}, \bibinfo{person}{Issam~H. Laradji}, \bibinfo{person}{Manuel Del~Verme}, \bibinfo{person}{Tom Marty}, \bibinfo{person}{David Vazquez}, \bibinfo{person}{Nicolas Chapados}, {and} \bibinfo{person}{Alexandre Lacoste}.} \bibinfo{year}{2024}\natexlab{}.
\newblock \showarticletitle{{W}ork{A}rena: How Capable are Web Agents at Solving Common Knowledge Work Tasks?}. In \bibinfo{booktitle}{\emph{Proceedings of the 41st International Conference on Machine Learning}} \emph{(\bibinfo{series}{Proceedings of Machine Learning Research}, Vol.~\bibinfo{volume}{235})}, \bibfield{editor}{\bibinfo{person}{Ruslan Salakhutdinov}, \bibinfo{person}{Zico Kolter}, \bibinfo{person}{Katherine Heller}, \bibinfo{person}{Adrian Weller}, \bibinfo{person}{Nuria Oliver}, \bibinfo{person}{Jonathan Scarlett}, {and} \bibinfo{person}{Felix Berkenkamp}} (Eds.). \bibinfo{publisher}{PMLR}, \bibinfo{pages}{11642--11662}.
\newblock
\urldef\tempurl%
\url{https://proceedings.mlr.press/v235/drouin24a.html}
\showURL{%
\tempurl}


\bibitem[{General Motors}(2024)]%
        {gm2024}
\bibfield{author}{\bibinfo{person}{{General Motors}}.} \bibinfo{year}{2024}\natexlab{}.
\newblock \bibinfo{booktitle}{\emph{Hands-Free, Eyes On}}.
\newblock
\urldef\tempurl%
\url{https://www.gm.com/innovation/av-safe-deployment}
\showURL{%
\tempurl}


\bibitem[Grgić-Hlača et~al\mbox{.}(2022)]%
        {grgichlaca2022}
\bibfield{author}{\bibinfo{person}{Nina Grgić-Hlača}, \bibinfo{person}{Claude Castelluccia}, {and} \bibinfo{person}{Krishna~P. Gummadi}.} \bibinfo{year}{2022}\natexlab{}.
\newblock \showarticletitle{Taking Advice from (Dis)Similar Machines: The Impact of Human-Machine Similarity on Machine-Assisted Decision-Making}.
\newblock \bibinfo{journal}{\emph{Proceedings of the AAAI Conference on Human Computation and Crowdsourcing}} \bibinfo{volume}{10}, \bibinfo{number}{1} (\bibinfo{date}{Oct.} \bibinfo{year}{2022}), \bibinfo{pages}{74--88}.
\newblock
\urldef\tempurl%
\url{https://doi.org/10.1609/hcomp.v10i1.21989}
\showDOI{\tempurl}


\bibitem[Hamade et~al\mbox{.}(2024)]%
        {hamade2024chesshandoff}
\bibfield{author}{\bibinfo{person}{Karim Hamade}, \bibinfo{person}{Reid McIlroy-Young}, \bibinfo{person}{Siddhartha Sen}, \bibinfo{person}{Jon Kleinberg}, {and} \bibinfo{person}{Ashton Anderson}.} \bibinfo{year}{2024}\natexlab{}.
\newblock \showarticletitle{Designing Skill-Compatible {AI}: Methodologies and Frameworks in Chess}. In \bibinfo{booktitle}{\emph{The Twelfth International Conference on Learning Representations}}.
\newblock


\bibitem[Iakovlev and Liang(2024)]%
        {iakovlev2024}
\bibfield{author}{\bibinfo{person}{Andrei Iakovlev} {and} \bibinfo{person}{Annie Liang}.} \bibinfo{year}{2024}\natexlab{}.
\newblock \bibinfo{title}{The Value of Context: Human versus Black Box Evaluators}.
\newblock
\newblock
\showeprint[arxiv]{2402.11157}~[econ.TH]
\urldef\tempurl%
\url{https://arxiv.org/abs/2402.11157}
\showURL{%
\tempurl}


\bibitem[Lai et~al\mbox{.}(2022)]%
        {lai2022}
\bibfield{author}{\bibinfo{person}{Vivian Lai}, \bibinfo{person}{Samuel Carton}, \bibinfo{person}{Rajat Bhatnagar}, \bibinfo{person}{Q.~Vera Liao}, \bibinfo{person}{Yunfeng Zhang}, {and} \bibinfo{person}{Chenhao Tan}.} \bibinfo{year}{2022}\natexlab{}.
\newblock \showarticletitle{Human-{AI} Collaboration via Conditional Delegation: A Case Study of Content Moderation}. In \bibinfo{booktitle}{\emph{Proceedings of the 2022 CHI Conference on Human Factors in Computing Systems}} (New Orleans, LA, USA) \emph{(\bibinfo{series}{CHI '22})}. \bibinfo{publisher}{Association for Computing Machinery}, \bibinfo{address}{New York, NY, USA}, Article \bibinfo{articleno}{54}, \bibinfo{numpages}{18}~pages.
\newblock
\showISBNx{9781450391573}
\urldef\tempurl%
\url{https://doi.org/10.1145/3491102.3501999}
\showDOI{\tempurl}


\bibitem[Leana(1986)]%
        {leana1986}
\bibfield{author}{\bibinfo{person}{Carrie~R. Leana}.} \bibinfo{year}{1986}\natexlab{}.
\newblock \showarticletitle{Predictors and Consequences of Delegation}.
\newblock \bibinfo{journal}{\emph{The Academy of Management Journal}} \bibinfo{volume}{29}, \bibinfo{number}{4} (\bibinfo{year}{1986}), \bibinfo{pages}{754--774}.
\newblock
\showISSN{00014273}
\urldef\tempurl%
\url{http://www.jstor.org/stable/255943}
\showURL{%
\tempurl}


\bibitem[Logg et~al\mbox{.}(2019)]%
        {logg2019}
\bibfield{author}{\bibinfo{person}{Jennifer~M. Logg}, \bibinfo{person}{Julia~A. Minson}, {and} \bibinfo{person}{Don~A. Moore}.} \bibinfo{year}{2019}\natexlab{}.
\newblock \showarticletitle{Algorithm appreciation: People prefer algorithmic to human judgment}.
\newblock \bibinfo{journal}{\emph{Organizational Behavior and Human Decision Processes}}  \bibinfo{volume}{151} (\bibinfo{year}{2019}), \bibinfo{pages}{90--103}.
\newblock
\showISSN{0749-5978}
\urldef\tempurl%
\url{https://doi.org/10.1016/j.obhdp.2018.12.005}
\showDOI{\tempurl}


\bibitem[Lubars and Tan(2019)]%
        {lubarstan2019}
\bibfield{author}{\bibinfo{person}{Brian Lubars} {and} \bibinfo{person}{Chenhao Tan}.} \bibinfo{year}{2019}\natexlab{}.
\newblock \showarticletitle{Ask not what {AI} can do, but what {AI} should do: towards a framework of task delegability}. In \bibinfo{booktitle}{\emph{Proceedings of the 33rd International Conference on Neural Information Processing Systems}}. \bibinfo{publisher}{Curran Associates Inc.}, \bibinfo{address}{Red Hook, NY, USA}, Article \bibinfo{articleno}{6}, \bibinfo{numpages}{11}~pages.
\newblock


\bibitem[Lykouris and Weng(2024)]%
        {lykouris2024}
\bibfield{author}{\bibinfo{person}{Thodoris Lykouris} {and} \bibinfo{person}{Wentao Weng}.} \bibinfo{year}{2024}\natexlab{}.
\newblock \bibinfo{title}{Learning to Defer in Content Moderation: The Human-{AI} Interplay}.
\newblock
\newblock
\showeprint[arxiv]{2402.12237}~[cs.LG]
\urldef\tempurl%
\url{https://arxiv.org/abs/2402.12237}
\showURL{%
\tempurl}


\bibitem[Madras et~al\mbox{.}(2018)]%
        {madras2018}
\bibfield{author}{\bibinfo{person}{David Madras}, \bibinfo{person}{Toniann Pitassi}, {and} \bibinfo{person}{Richard Zemel}.} \bibinfo{year}{2018}\natexlab{}.
\newblock \showarticletitle{Predict responsibly: improving fairness and accuracy by learning to defer}. In \bibinfo{booktitle}{\emph{Proceedings of the 32nd International Conference on Neural Information Processing Systems}} (Montr\'{e}al, Canada) \emph{(\bibinfo{series}{NIPS'18})}. \bibinfo{publisher}{Curran Associates Inc.}, \bibinfo{address}{Red Hook, NY, USA}, \bibinfo{pages}{6150–6160}.
\newblock


\bibitem[Marsden and Kirby(2005)]%
        {marsden2005}
\bibfield{author}{\bibinfo{person}{Philip Marsden} {and} \bibinfo{person}{Mark Kirby}.} \bibinfo{year}{2005}\natexlab{}.
\newblock \showarticletitle{Allocation of functions}.
\newblock \bibinfo{journal}{\emph{Handbook of human factors and ergonomics methods}} (\bibinfo{year}{2005}), \bibinfo{pages}{34--1}.
\newblock


\bibitem[Mervis and Rosch(1981)]%
        {mervis1981}
\bibfield{author}{\bibinfo{person}{C~B Mervis} {and} \bibinfo{person}{E Rosch}.} \bibinfo{year}{1981}\natexlab{}.
\newblock \showarticletitle{Categorization of Natural Objects}.
\newblock \bibinfo{journal}{\emph{Annual Review of Psychology}} \bibinfo{volume}{32}, \bibinfo{number}{Volume 32, 1981} (\bibinfo{year}{1981}), \bibinfo{pages}{89--115}.
\newblock
\showISSN{1545-2085}
\urldef\tempurl%
\url{https://doi.org/10.1146/annurev.ps.32.020181.000513}
\showDOI{\tempurl}


\bibitem[Milewski and Lewis(1997)]%
        {milewskilewis1997delegatingtosa}
\bibfield{author}{\bibinfo{person}{Allen~E. Milewski} {and} \bibinfo{person}{Steven~H. Lewis}.} \bibinfo{year}{1997}\natexlab{}.
\newblock \showarticletitle{Delegating to software agents}.
\newblock \bibinfo{journal}{\emph{International Journal of Human-Computer Studies}} \bibinfo{volume}{46}, \bibinfo{number}{4} (\bibinfo{year}{1997}), \bibinfo{pages}{485--500}.
\newblock
\showISSN{1071-5819}
\urldef\tempurl%
\url{https://doi.org/10.1006/ijhc.1996.0100}
\showDOI{\tempurl}


\bibitem[Mullainathan(2002)]%
        {mullainathan2002}
\bibfield{author}{\bibinfo{person}{Sendhil Mullainathan}.} \bibinfo{year}{2002}\natexlab{}.
\newblock \bibinfo{title}{Thinking Through Categories}.
\newblock
\newblock


\bibitem[O'Neill(2014)]%
        {oneill14}
\bibfield{author}{\bibinfo{person}{B. O'Neill}.} \bibinfo{year}{2014}\natexlab{}.
\newblock \showarticletitle{Some Useful Moment Results in Sampling Problems}.
\newblock \bibinfo{journal}{\emph{The American Statistician}} \bibinfo{volume}{68}, \bibinfo{number}{4} (\bibinfo{year}{2014}), \bibinfo{pages}{282--296}.
\newblock
\showISSN{00031305}
\urldef\tempurl%
\url{http://www.jstor.org/stable/24591747}
\showURL{%
\tempurl}


\bibitem[Owens et~al\mbox{.}(2014)]%
        {owens2014}
\bibfield{author}{\bibinfo{person}{David Owens}, \bibinfo{person}{Zachary Grossman}, {and} \bibinfo{person}{Ryan Fackler}.} \bibinfo{year}{2014}\natexlab{}.
\newblock \showarticletitle{The Control Premium: A Preference for Payoff Autonomy}.
\newblock \bibinfo{journal}{\emph{American Economic Journal: Microeconomics}} \bibinfo{volume}{6}, \bibinfo{number}{4} (\bibinfo{date}{November} \bibinfo{year}{2014}), \bibinfo{pages}{138–61}.
\newblock
\urldef\tempurl%
\url{https://doi.org/10.1257/mic.6.4.138}
\showDOI{\tempurl}


\bibitem[Passi and Vorvoreanu(2022)]%
        {passi2022}
\bibfield{author}{\bibinfo{person}{Samir Passi} {and} \bibinfo{person}{Mihaela Vorvoreanu}.} \bibinfo{year}{2022}\natexlab{}.
\newblock \bibinfo{booktitle}{\emph{Overreliance on {AI}: Literature Review}}.
\newblock \bibinfo{type}{{T}echnical {R}eport} MSR-TR-2022-12. \bibinfo{institution}{Microsoft}.
\newblock
\urldef\tempurl%
\url{https://www.microsoft.com/en-us/research/publication/overreliance-on-ai-literature-review/}
\showURL{%
\tempurl}


\bibitem[Price(1985)]%
        {price1985}
\bibfield{author}{\bibinfo{person}{Harold~E. Price}.} \bibinfo{year}{1985}\natexlab{}.
\newblock \showarticletitle{The Allocation of Functions in Systems}.
\newblock \bibinfo{journal}{\emph{Human Factors}} \bibinfo{volume}{27}, \bibinfo{number}{1} (\bibinfo{year}{1985}), \bibinfo{pages}{33--45}.
\newblock
\urldef\tempurl%
\url{https://doi.org/10.1177/001872088502700104}
\showDOI{\tempurl}


\bibitem[Purdy(2024)]%
        {hbr2024}
\bibfield{author}{\bibinfo{person}{Mark Purdy}.} \bibinfo{year}{2024}\natexlab{}.
\newblock \showarticletitle{What Is Agentic {AI}, and How Will It Change Work?}
\newblock  (\bibinfo{year}{2024}).
\newblock


\bibitem[Raghu et~al\mbox{.}(2018)]%
        {raghu2018}
\bibfield{author}{\bibinfo{person}{Maithra Raghu}, \bibinfo{person}{Katy Blumer}, \bibinfo{person}{Greg Corrado}, \bibinfo{person}{Jon Kleinberg}, \bibinfo{person}{Ziad Obermeyer}, {and} \bibinfo{person}{Sendhil Mullainathan}.} \bibinfo{year}{2018}\natexlab{}.
\newblock \showarticletitle{The Algorithmic Automation Problem: Prediction, Triage, and Human Effort}.
\newblock  (\bibinfo{year}{2018}).
\newblock


\bibitem[Shademan et~al\mbox{.}(2016)]%
        {shademan2016surgery}
\bibfield{author}{\bibinfo{person}{Azad Shademan}, \bibinfo{person}{Ryan~S. Decker}, \bibinfo{person}{Justin~D. Opfermann}, \bibinfo{person}{Simon Leonard}, \bibinfo{person}{Axel Krieger}, {and} \bibinfo{person}{Peter C.~W. Kim}.} \bibinfo{year}{2016}\natexlab{}.
\newblock \showarticletitle{Supervised autonomous robotic soft tissue surgery}.
\newblock \bibinfo{journal}{\emph{Science Translational Medicine}} \bibinfo{volume}{8}, \bibinfo{number}{337} (\bibinfo{year}{2016}), \bibinfo{pages}{337ra64--337ra64}.
\newblock
\urldef\tempurl%
\url{https://doi.org/10.1126/scitranslmed.aad9398}
\showDOI{\tempurl}
\showeprint{https://www.science.org/doi/pdf/10.1126/scitranslmed.aad9398}


\bibitem[Smith(2014)]%
        {smith2014}
\bibfield{author}{\bibinfo{person}{J.~David Smith}.} \bibinfo{year}{2014}\natexlab{}.
\newblock \showarticletitle{Prototypes, exemplars, and the natural history of categorization}.
\newblock \bibinfo{journal}{\emph{Psychonomic Bulletin {\&} Review}} \bibinfo{volume}{21}, \bibinfo{number}{2} (\bibinfo{date}{01 Apr} \bibinfo{year}{2014}), \bibinfo{pages}{312--331}.
\newblock
\showISSN{1531-5320}
\urldef\tempurl%
\url{https://doi.org/10.3758/s13423-013-0506-0}
\showDOI{\tempurl}


\bibitem[Spanos(1999)]%
        {spanos2019}
\bibfield{author}{\bibinfo{person}{Aris Spanos}.} \bibinfo{year}{1999}\natexlab{}.
\newblock \bibinfo{booktitle}{\emph{Probability theory and statistical inference: empirical modelling with observational data}}.
\newblock \bibinfo{publisher}{Cambridge University Press}. 339--356 pages.
\newblock
\showISBNx{0 521 41354 0}
\urldef\tempurl%
\url{https://doi.org/10.1017/9781316882825}
\showDOI{\tempurl}


\bibitem[Stout et~al\mbox{.}(2014)]%
        {stout2014}
\bibfield{author}{\bibinfo{person}{Nathan Stout}, \bibinfo{person}{Alan Dennis}, {and} \bibinfo{person}{Taylor Wells}.} \bibinfo{year}{2014}\natexlab{}.
\newblock \showarticletitle{The Buck Stops There: The Impact of Perceived Accountability and Control on the Intention to Delegate to Software Agents}.
\newblock \bibinfo{journal}{\emph{AIS Transactions on Human-Computer Interaction}}  \bibinfo{volume}{6} (\bibinfo{date}{03} \bibinfo{year}{2014}).
\newblock
\urldef\tempurl%
\url{https://doi.org/10.17705/1thci.00058}
\showDOI{\tempurl}


\bibitem[Tesla(2024)]%
        {tesla2024}
\bibfield{author}{\bibinfo{person}{Tesla}.} \bibinfo{year}{2024}\natexlab{}.
\newblock \bibinfo{booktitle}{\emph{Autopilot and Full Self-Driving (Supervised)}}.
\newblock
\urldef\tempurl%
\url{https://www.tesla.com/support/autopilot}
\showURL{%
\tempurl}


\bibitem[Townsend et~al\mbox{.}(2000)]%
        {townsend2000}
\bibfield{author}{\bibinfo{person}{James~T. Townsend}, \bibinfo{person}{Kam~M. Silva}, \bibinfo{person}{Jesse Spencer-Smith}, {and} \bibinfo{person}{Michael~J. Wenger}.} \bibinfo{year}{2000}\natexlab{}.
\newblock \showarticletitle{Exploring the relations between categorization and decision making with regard to realistic face stimuli}.
\newblock \bibinfo{journal}{\emph{Pragmatics \& Cognition}} \bibinfo{volume}{8}, \bibinfo{number}{1} (\bibinfo{year}{2000}), \bibinfo{pages}{83--105}.
\newblock
\showISSN{0929-0907}
\urldef\tempurl%
\url{https://doi.org/10.1075/pc.8.1.05tow}
\showDOI{\tempurl}


\bibitem[Tsourakakis et~al\mbox{.}(2019)]%
        {tsourkakis2019}
\bibfield{author}{\bibinfo{person}{Charalampos~E. Tsourakakis}, \bibinfo{person}{Tianyi Chen}, \bibinfo{person}{Naonori Kakimura}, {and} \bibinfo{person}{Jakub Pachocki}.} \bibinfo{year}{2019}\natexlab{}.
\newblock \showarticletitle{Novel Dense Subgraph Discovery Primitives: Risk Aversion and Exclusion Queries}. In \bibinfo{booktitle}{\emph{Machine Learning and Knowledge Discovery in Databases: European Conference, ECML PKDD 2019, W\"{u}rzburg, Germany, September 16–20, 2019, Proceedings, Part I}} (W\"{u}rzburg, Germany). \bibinfo{publisher}{Springer-Verlag}, \bibinfo{address}{Berlin, Heidelberg}, \bibinfo{pages}{378–394}.
\newblock
\showISBNx{978-3-030-46149-2}
\urldef\tempurl%
\url{https://doi.org/10.1007/978-3-030-46150-8_23}
\showDOI{\tempurl}


\bibitem[Vafa et~al\mbox{.}(2024)]%
        {vafa2024}
\bibfield{author}{\bibinfo{person}{Keyon Vafa}, \bibinfo{person}{Ashesh Rambachan}, {and} \bibinfo{person}{Sendhil Mullainathan}.} \bibinfo{year}{2024}\natexlab{}.
\newblock \showarticletitle{Do Large Language Models Perform the Way People Expect? {M}easuring the Human Generalization Function}. In \bibinfo{booktitle}{\emph{Proceedings of the 41st International Conference on Machine Learning}} \emph{(\bibinfo{series}{Proceedings of Machine Learning Research}, Vol.~\bibinfo{volume}{235})}, \bibfield{editor}{\bibinfo{person}{Ruslan Salakhutdinov}, \bibinfo{person}{Zico Kolter}, \bibinfo{person}{Katherine Heller}, \bibinfo{person}{Adrian Weller}, \bibinfo{person}{Nuria Oliver}, \bibinfo{person}{Jonathan Scarlett}, {and} \bibinfo{person}{Felix Berkenkamp}} (Eds.). \bibinfo{publisher}{PMLR}, \bibinfo{pages}{48919--48937}.
\newblock
\urldef\tempurl%
\url{https://proceedings.mlr.press/v235/vafa24a.html}
\showURL{%
\tempurl}


\bibitem[{Wing Assistant}(2024)]%
        {wingaitool}
\bibfield{author}{\bibinfo{person}{{Wing Assistant}}.} \bibinfo{year}{2024}\natexlab{}.
\newblock \bibinfo{booktitle}{\emph{What can {Wing General VAs} do?}}
\newblock
\urldef\tempurl%
\url{https://wingassistant.com}
\showURL{%
\tempurl}


\bibitem[Xu(2024)]%
        {xu2024}
\bibfield{author}{\bibinfo{person}{Ruqing Xu}.} \bibinfo{year}{2024}\natexlab{}.
\newblock \bibinfo{title}{Persuasion, Delegation, and Private Information in Algorithm-Assisted Decisions}.
\newblock
\newblock
\showeprint[arxiv]{2402.09384}~[econ.TH]
\urldef\tempurl%
\url{https://arxiv.org/abs/2402.09384}
\showURL{%
\tempurl}


\bibitem[Yao et~al\mbox{.}(2022)]%
        {yao2022webshop}
\bibfield{author}{\bibinfo{person}{Shunyu Yao}, \bibinfo{person}{Howard Chen}, \bibinfo{person}{John Yang}, {and} \bibinfo{person}{Karthik Narasimhan}.} \bibinfo{year}{2022}\natexlab{}.
\newblock \showarticletitle{WebShop: Towards Scalable Real-World Web Interaction with Grounded Language Agents}. In \bibinfo{booktitle}{\emph{Advances in Neural Information Processing Systems}}, \bibfield{editor}{\bibinfo{person}{S.~Koyejo}, \bibinfo{person}{S.~Mohamed}, \bibinfo{person}{A.~Agarwal}, \bibinfo{person}{D.~Belgrave}, \bibinfo{person}{K.~Cho}, {and} \bibinfo{person}{A.~Oh}} (Eds.), Vol.~\bibinfo{volume}{35}. \bibinfo{publisher}{Curran Associates, Inc.}, \bibinfo{pages}{20744--20757}.
\newblock
\urldef\tempurl%
\url{https://proceedings.neurips.cc/paper_files/paper/2022/file/82ad13ec01f9fe44c01cb91814fd7b8c-Paper-Conference.pdf}
\showURL{%
\tempurl}


\bibitem[YXD(2013)]%
        {stackoverflowminvar}
\bibfield{author}{\bibinfo{person}{YXD}.} \bibinfo{year}{2013}\natexlab{}.
\newblock \showarticletitle{Computing the subset giving the minimum standard deviation in an array}.
\newblock \bibinfo{journal}{\emph{Stack Overflow}} (\bibinfo{year}{2013}).
\newblock
\urldef\tempurl%
\url{https://stackoverflow.com/a/20143904}
\showURL{%
\tempurl}


\bibitem[Zheng et~al\mbox{.}(2023)]%
        {zheng2023}
\bibfield{author}{\bibinfo{person}{Rong Zheng}, \bibinfo{person}{Jerome~R. Busemeyer}, {and} \bibinfo{person}{Robert~M. Nosofsky}.} \bibinfo{year}{2023}\natexlab{}.
\newblock \showarticletitle{Integrating Categorization and Decision-Making}.
\newblock \bibinfo{journal}{\emph{Cognitive Science}} \bibinfo{volume}{47}, \bibinfo{number}{1} (\bibinfo{year}{2023}), \bibinfo{pages}{e13235}.
\newblock
\urldef\tempurl%
\url{https://doi.org/10.1111/cogs.13235}
\showDOI{\tempurl}
\showeprint{https://onlinelibrary.wiley.com/doi/pdf/10.1111/cogs.13235}


\bibitem[Zhou et~al\mbox{.}(2023)]%
        {zhou2023webarena}
\bibfield{author}{\bibinfo{person}{Shuyan Zhou}, \bibinfo{person}{Frank~F Xu}, \bibinfo{person}{Hao Zhu}, \bibinfo{person}{Xuhui Zhou}, \bibinfo{person}{Robert Lo}, \bibinfo{person}{Abishek Sridhar}, \bibinfo{person}{Xianyi Cheng}, \bibinfo{person}{Yonatan Bisk}, \bibinfo{person}{Daniel Fried}, \bibinfo{person}{Uri Alon}, {et~al\mbox{.}}} \bibinfo{year}{2023}\natexlab{}.
\newblock \bibinfo{title}{WebArena: A Realistic Web Environment for Building Autonomous Agents}.
\newblock
\newblock
\showeprint[arxiv]{2307.13854}
\urldef\tempurl%
\url{https://webarena.dev}
\showURL{%
\tempurl}


\end{thebibliography}

\appendix

\section{Computing the loss in the single feature setting}\label{app:twofeatjustification}

Recall the setting of Section \ref{sec:twofeature}, where each state has two binary features, the human observes the first feature, and the machine observes the second feature. The human thus has two categories $C_1$ and $C_2$ based on the value of the first feature $x_1$, and the machine has two categories $K_1$ and $K_2$ based on the second feature $x_2$. Let $\x_{ij}$ be the state in $C_i \cap K_j$. We restricted our consideration to settings where each state occurs with equal probability, and the ground truth optimal action function $\fnopt$ is of the form \[\fnopt(\x) = \begin{cases}
    0, & \x = \x_{11}\\
    1, & \x = \x_{12}\\
    a, & \x = \x_{21}\\
    b, & \x = \x_{22}
\end{cases}\] for $a, b \in \mathbb{R}$.

In this section, we compute the loss of the human delegating to $\mfn^\kept$ in the categories in $\kept$ for each $\kept \subseteq \{C_1, C_2\}$, which we denote $\loss(\kept)$.

First, notice that \[\hfnopt(C_1) = \frac{1}{2}, \quad \hfnopt(C_2) = \frac{a+b}{2},\] so \begin{align*}
    \loss_H(\hfnopt, C_1) &= \frac{1}{2}(1/2)^2 + \frac{1}{2}(1/2)^2 = 1/4
\end{align*} and \begin{align*}
    \loss_H(\hfnopt, C_2) &= \frac{1}{2}((a-b)/2)^2 + \frac{1}{2}((a-b)/2)^2 = (a-b)^2 / 4.
\end{align*}

We consider each value of $\kept$ in turn. If $\kept = \{C_1, C_2\}$, \begin{align*}\mfn^\kept(K_1) &= \frac{\fnopt(\x_{11}) + \fnopt(\x_{21})}{2} = \frac{a}{2}, \\
\mfn^\kept(K_1) &= \frac{\fnopt(\x_{12}) +\fnopt(\x_{22})}{2} = \frac{b+1}{2}.\end{align*} Then \begin{align*}
        \mloss(\mfn^\kept, C_1) &= \frac{1}{2}(\fnopt(\x_{11}) - \mfn^\kept(K_1))^2 + \frac{1}{2}(\fnopt(\x_{12}) - \mfn^\kept(K_2))^2\\
        &= \frac{1}{2}(a/2)^2 + \frac{1}{2}((b-1)/2)^2\\
        &= \frac{1}{8}(a^2 + (b-1)^2)
    \end{align*} Similarly, \begin{align*}
        \mloss(\mfn^\kept, C_2) &= \frac{1}{2}(\fnopt(\x_{21}) - \mfn^\kept(K_1))^2 + \frac{1}{2}(\fnopt(\x_{22}) - \mfn^\kept(K_2))^2\\
        &= \frac{1}{2}(a/2)^2 + \frac{1}{2}((b-1)/2)^2\\
        &= \frac{1}{8}(a^2 + (b-1)^2).
    \end{align*}
    Thus \[\loss(\{C_1, C_2\}) = \frac{1}{2}\loss_M(\mfn^\kept, C_1)+ \frac{1}{2}\loss_M(\mfn^\kept, C_2) = \frac{1}{8}(a^2 + (b-1)^2).\]

Next, if $\kept = \{C_1\}$, \[\mfn^\kept(K_1) = 0, \quad \mfn^\kept(K_1) = 1,\] and \[\mloss(\mfn^\kept, C_1) = 0.\] So \[\loss(\{C_1\}) = \frac{1}{2}\mloss(\mfn^\kept, C_1) + \frac{1}{2}\hloss(\hfnopt, C_2) = \frac{1}{2}\cdot\frac{(a-b)^2}{4} = (a-b)^2/8.\]

If $\kept = \{C_2\}$, \[\mfn^\kept(K_1) = a, \quad \mfn^\kept(K_1) = b,\] and \[\mloss(\mfn^\kept, C_2) = 0.\] So \[\loss(\{C_2\}) = \frac{1}{2}\hloss(\hfnopt, C_1) + \frac{1}{2}\mloss(\mfn^\kept, C_2) = \frac{1}{2}\cdot\frac{1}{4} = 1/8.\]

Finally, if $\kept = \{\}$, then \[\loss(\kept) = \frac{1}{2}\hloss(\hfnopt, C_1) + \frac{1}{2}\hloss(\hfnopt, C_2) = 1/8 + (a-b)^2/8.\]

Clearly, $\loss(\emptyset) \geq \loss(\{C_1\}), \loss(\{C_2\})$ as we claimed in Section \ref{sec:twofeature}.

Thus by \Cref{prop:dropisoptimal}, the loss of the optimal machine is \begin{align*}
    \loss(\hfnopt, \mfnopt) &= \min_\kept \loss(\kept) =\frac{1}{8}\cdot \min\{1, (a-b)^2, a^2 + (b-1)^2\}.
\end{align*}

 \subsection{Maximally adopted delegates in two-feature settings}\label{app:twofeatjustification:maxadopt}

 We now prove that in the two-feature setting, the settings with optimal action function $\fnopt_{\fx, \fy}$ in which it is possible to design a machine that the human will always delegate to corresponds exactly to the set of $(a,b)$ shown in Figure \ref{fig:twofeat:maxadopt}.

Still in the two-feature setting, consider some machine $\mfn$, and denote $y_1 = \mfn(K_1), y_2 = \mfn(K_2)$. This machine is adopted in both $C_1$ and $C_2$ if both \begin{align*}
    \mloss(\mfn, C_1) < \hloss(\hfnopt, C_1)
    \quad \text{and} \quad\mloss(\mfn, C_2) < \hloss(\hfnopt, C_2)
\end{align*} which corresponds to \begin{align*}
    (y_1 - 0)^2 + (y_2 - 1)^2 < (1/\sqrt{2})^2 \quad \text{and} \quad (y_1 - a)^2 + (y_2 - b)^2 < ((a-b)/\sqrt{2})^2
\end{align*} That is, such a $\mfn$ exists there is some intersection between the circle centered at $(0,1)$ with radius $1/\sqrt{2}$, and the circle centered at $(a,b)$ with radius $|a-b|/\sqrt{2}$. This is the case if and only if the distance between the two circles is less than the sum of radii, that is, \[\sqrt{(a-0)^2 + (b-1)^2} < \frac{1}{\sqrt{2}} + \frac{|a-b|}{\sqrt{2}}.\] 

Figure \ref{fig:twofeat:maxadopt} shows the set of $(\fx,\fy)$ that satisfy this condition.

 This set roughly covers the region $a < 0$, $b > 1$. In this region, the human's loss in $C_2$ suffers from her inability to distinguish between $K_1$ and $K_2$ where the optimal actions $a$ and $b$ are very different. The machine can't distinguish between $C_1$ and $C_2$, but the difference between the optimal actions (between $a$ and $0$ in $C_1$ and between $b$ and $1$ in $C_2$) is not as large as the difference between $a$ and $b$. The machine can stay close to ($0,1$) to capture $C_1$ while automatically doing better than the human on $C_2$. Moreover, $\{(a,b): a < 0$ and $b = 1-a\}$ is contained in this set. For such $(a,b)$, the delegate $(0,1)$ will be adopted in both human categories but will have a loss of $a^2$ in the second human category, which becomes unbounded as the magnitude of $a$ grows. By contrast, the optimal delegate in the two-feature setting always has bounded loss. Therefore, a firm maximizing adoption rates rather than team performance could design an arbitrarily worse machine. 

Notably, the machine $(0,1)$ is not the oblivious machine, despite the oblivious machine being the machine achieving the best expected performance across both $C_1$ and $C_2$! In order to guarantee adoption in $C_1$, the firm sacrifices performance in $C_2$ rather than treating the two categories equally as in the oblivious delegate (even though the two categories have equal probability). In general, if $a^2 + (b-1)^2 > 2$, the oblivious machine is not adopted in $C_1$, so the maximally-adopted machine cannot be the oblivious machine. This means that in this region, the maximally-adopted machine has worse performance than even the oblivious machine acting alone, which in turn has worse team performance than the human using the oblivious machine as a delegate.  
 
\section{Additional iterative design experiments}\label{app:additionalexperiments}

In Figure \ref{fig:lossratios}, we show the results computational simulations to compare the performance of the iterative solution to both the optimal solution and the oblivious solution; we provide details in the caption.

\begin{figure}[ht]
    \centering
    \includegraphics[width=0.5\linewidth]{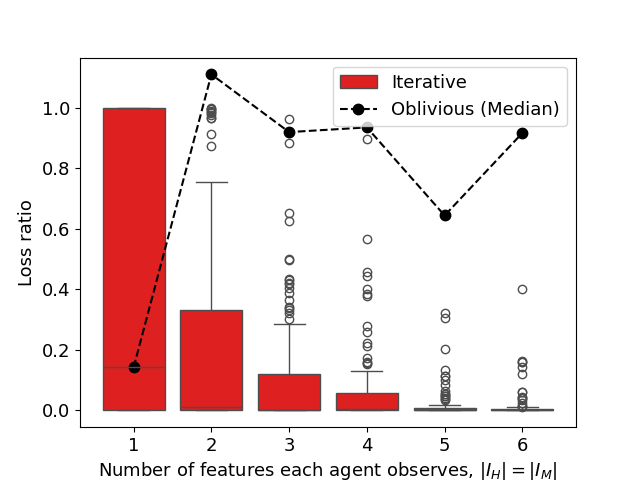}
    \caption{Relative performance between the iterative solution and the optimal solution. We consider delegation settings where the ground truth optimal action is a linear function $\fnopt(\x) = \mathbf{v}^T\x$ where $\mathbf{v} \sim \textnormal{Normal}(0, I)$, and each state has uniform probability. We vary the number of features the agents observe, so that $|\idxset{H}| = |\idxset{M}|$ ranges from $1$ to $6$, and $d = |\idxset{H}| + |\idxset{M}|$. For each $|\idxset{H}|$ we sample 100 ground truth functions $\fnopt$, and plot the relative difference in loss between the iterative solution $\mfn^{\textnormal{iter}}$ and the optimal solution $\mfnopt$, 
    $[\loss(\hfnopt, \mfn^{\textnormal{iter}}) - \loss(\hfnopt, \mfnopt)]/\loss(\hfnopt, \mfnopt).$
    We also plot the median relative difference in loss for the oblivious machine in the same 100 samples as a baseline.}
    \Description{A boxplot showing the loss ratio for different numbers of observed features. All loss ratios are positive, and the distribution of loss ratios shifts strongly toward zero as the number of observed features increases. By contrast, the median loss ratio for the oblivious machine decreases at first, but plateaus near 0.9.}
    \label{fig:lossratios}
\end{figure}

In Table \ref{tab:props}, we extend the simulations of Figure \ref{fig:lossratios} to cases where the human and machine have a different number of features, and determine the empirical proportion of sampled settings where the iterative approach arrives at an optimal solution.

\begin{table}[ht]
    \centering
    \caption{
    We again let $\idxset{H}$ and $\idxset{M}$ partition $[d]$, but vary each of $|I_H|$ and $|I_M|$ between 1 and 6. We examine $\fnopt(\x) = \textbf{v}^T \x$ for $1000$ samples of $\textbf{v} \sim \textnormal{Normal}(0, I)$, and compute the sample proportion of settings $\fnopt$ in which the iterative solution is optimal (specifically, if the loss of the iterative solution is within a tolerance of $10^{-8}$ of the optimal loss.)}
    \label{tab:props}
    \begin{tabular}{rc|ccccccc}
    \multicolumn{2}{c}{} & \multicolumn{6}{c}{\# of machine-observable features, $|I_M|$}\\
     \multicolumn{2}{c|}{} & 1 & 2& 3 & 4 & 5 & 6\\
    \cline{2-8}
    \multirow{6}{2.2cm}{\# of human-observable features, $|I_H|$} & 1 & 0.39 & 0.57 & 0.69 & 0.76 & 0.82 & 0.86 \\
    & 2 & 0.34 & 0.48 & 0.56 & 0.66 & 0.68 & 0.74\\
    & 3 & 0.31 & 0.43 & 0.49 & 0.56 & 0.60 & 0.62 \\
    & 4 & 0.31 & 0.42 & 0.45 & 0.52 & 0.52 & 0.59 \\
    & 5 & 0.28 & 0.37 & 0.40 & 0.44 & 0.50 & 0.54 \\
    & 6 & 0.26 & 0.32 & 0.39 & 0.42 & 0.44 & 0.52\\
    \end{tabular}
\end{table}

\section{Extension to arbitrary categories}\label{app:arbitrarycategories}

To see this, first note that the proofs of our characterization results \Cref{prop:dropisoptimal} and \Cref{prop:variancegame:full} (found in Appendix \ref{app:charproofs}) hold when $C, K$ are arbitrary partitions of $\hcats, \mcats$.

An analogue of Lemma \ref{lem:extra} defines $\bar\fn(C\cap K) = \sum_{\x \in C\cap K} \frac{\Prob(\x)}{\Prob(C \cap K)} \fnopt(\x)$ for each pair $C, K$, so we may treat each intersection $C \cap K$ as containing a single instance with probability $\Prob(C \cap K)$.

Our tractability result when the human or machine observes a constant number of features (\Cref{thm:constant}) likewise holds for arbitrary partitions,\footnote{If any intersection $C \cap K = \emptyset$, create a dummy state $\x$ and set $\Prob(\x) = 0$. Combining this with the analogue of Lemma \ref{lem:extra}, we may assume as before that there is exactly one state $\x$ in $C \cap K$ for each $C,K$.} as does the result on efficiently determining the existence of a perfect team (\Cref{prop:perfectteams}). Since these arbitrary categories are a generalization of categories arising from limited information, our hardness result in \Cref{thm:np} still holds. The only result that truly relies on the existence of binary features is \Cref{thm:separable}. This can be shown to extend when we extend the definition of consistency and consequently of separable distributions. We say that two categories $C$ and $K$ are inconsistent if $C \cap K = \emptyset$. We say that  $S_1, ..., S_s$ is a set of mutually inconsistent subproblems if $S_1, ..., S_s$ is a partition of the human and machine categories, and if $C \in S_t$ and $K \in S_t'$, then $C$ and $K$ are inconsistent. We say that $\Prob$ is separable if for all $t$, and for all $C_i, K_j \in S_t$, $\Prob(C_i \cap K_j) = p_i \cdot q_j$ for some $p_i, q_j$. With this assumption, the proofs of Lemma \ref{lem:shared} and \Cref{thm:separable} easily extend. 

\section{Characterization proofs (Propositions \ref{prop:dropisoptimal} and \ref{prop:variancegame:full})}\label{app:charproofs}

First, recall Proposition \ref{prop:dropisoptimal}.

\dropisoptimalinformal*

\begin{proof} We first define an intermediate problem, \begin{equation}\label{problem:intermediate}
        \min_{\mfn, \kept \subseteq \hcats} \sum_{C \in \hcats \setminus \kept} \Prob(C) \hloss(\hfnopt, C) + \sum_{C \in \kept} \Prob(C) \mloss(\mfn, C)
    \end{equation}
If $P(C) = 0$ for some category $C$, then $C$ does not affect the expected team loss regardless of which agent takes action. In the remainder of this proof, we ignore these categories without loss of generality, and assume that $P(C) > 0$ for all $C$. 
\paragraph{Part 1} Problem \ref{problem:discrete} can be written as 
\begin{align*} 
\min_{\mfn, \kept \subseteq \hcats} \quad & \sum_{C \in \hcats \setminus \kept} \Prob(C)\hloss(\hfnopt, C) + \sum_{C \in \kept} \Prob(C) \mloss(\mfn, C) \tag{\ref{problem:discrete}$'$}\label{problem:discrete:2}\\
    \text{s.t.} \quad & \mfn = \mfn^\kept.
\end{align*}

Clearly Problem \ref{problem:discrete:2} is a constrained version of Problem \ref{problem:intermediate}, and thus if every solution $\kept, \mfn$ to Problem \ref{problem:intermediate} satisfies $\mfn^\kept = \mfn$, then Problem \ref{problem:discrete:2} and Problem \ref{problem:intermediate} have the same set of solutions.

This is indeed the case. We know that \begin{align*}
        \text{Problem }\ref{problem:intermediate} &\equiv \min_\kept \min_{\mfn} \sum_{C \in \hcats \setminus \kept}\Prob(C) \hloss(\hfnopt, C) + \sum_{C \in \kept} \Prob(C) \mloss(\mfn, C)\\
        &\equiv \min_\kept \sum_{C \in \hcats \setminus \kept} \Prob(C) \hloss(\hfnopt, C) + \min_{\mfn} \sum_{C \in \kept} \Prob(C) \mloss(\mfn, C).    \end{align*} 
        Moreover, \begin{align*}
    \argmin_{\mfn} \sum_{C \in \kept} \Prob(C) \mloss(\mfn, C)
    = \argmin_{\mfn} \sum_K \sum_{\x \in X(\kept) \cap K} \Prob(\x) (\mfn(\x) - \fnopt(\x))^2
    := \mfn^\kept,\end{align*} since squared loss is uniquely minimized by the expected value, so any optimal solution $\kept, \mfn$ to Problem \ref{problem:intermediate} must satisfy $\mfn^\kept = \mfn$. 
    Thus Problem \ref{problem:intermediate} is equivalent to Problem \ref{problem:discrete:2}.

    \paragraph{Part 2} Problem \ref{problem:continuous} can be expressed as \begin{align*}
    \min_{\mfn, \kept \subseteq \hcats} \quad & \sum_{C \in \hcats \setminus \kept}\Prob(C)\hloss(\hfnopt, C) + \sum_{C \in \kept} \Prob(C)\mloss(\mfn, C)\tag{\ref{problem:continuous}$'$}\label{problem:continuous:2} \\
        \text{s.t.} \quad & C \in \kept \iff C \in  \dcats(\hfnopt, \mfn) \textnormal{ for all } C \textnormal{ where } \hloss(\hfnopt, C) \neq \mloss(\mfn, C).
\end{align*}
    This is again a constrained version of Problem \ref{problem:intermediate}, and thus if every solution $\mfn, \kept$ to Problem \ref{problem:intermediate} satisfies the constraint, then Problems \ref{problem:intermediate} and \ref{problem:continuous:2} have the same set of solutions. 

Suppose for the sake of contradiction that $\mfn, \kept$ is an optimal solution to Problem \ref{problem:intermediate} but is not a feasible solution to Problem \ref{problem:continuous:2}. This means that for some $C$ where $\hloss(\hfnopt, C) \neq \mloss(\mfn, C)$, either $C \in \hcats\setminus \kept$ and $\mloss(\mfn, C) < \hloss(\hfnopt, C)$ or $C \in \kept$ and $\hloss(\hfnopt, C) < \mloss(\mfn, C)$.
    
    If for some $C \in  \hcats\setminus \kept$, $\mloss(\mfn, C) < \hloss(\hfnopt, C)$, let $\kept' = \kept \cup \{C\}$. Then \begin{align*}
        &\left(\sum_{C \in \hcats \setminus \kept'}\Prob(C)\hloss(\hfnopt, C) + \sum_{C \in \kept'} \Prob(C) \mloss(\mfn, C)\right) -\left(\sum_{C \in \hcats \setminus \kept}\Prob(C) \hloss(\hfnopt, C) + \sum_{C \in \kept} \Prob(C) \mloss(\mfn, C)\right)\\
        &=  \Prob(C)\left(\mloss(\mfn, C) - \hloss(\hfnopt, C)\right)\\
        &< 0
    \end{align*}
    and $\mfn, \kept$ is not an optimal solution to Problem \ref{problem:intermediate}. If instead there is some $C \in \kept$, and $\hloss(\hfnopt, C) < \mloss(\mfn, C)$ we can similarly let $\kept' = \kept \setminus \{C\}$, and \begin{align*}
        &\left(\sum_{C \in \hcats \setminus \kept'} \Prob(C) \hloss(\hfnopt, C) + \sum_{C \in \kept'} \Prob(C) \mloss(\mfn, C)\right) - \left(\sum_{C \in \hcats \setminus \kept} \Prob(C)\hloss(\hfnopt, C) + \sum_{C \in \kept} \Prob(C) \mloss(\mfn, C)\right)\\
        &=  \Prob(C)\left(\hloss(\hfnopt, C) - \mloss(\mfn, C)\right)\\
        &< 0
    \end{align*} so $\mfn, \kept$ is not an optimal solution to Problem \ref{problem:intermediate} and we have a contradiction. Thus $\mfn, \kept$ must be a feasible solution to Problem \ref{problem:continuous:2}, and Problems \ref{problem:continuous:2} and \ref{problem:intermediate} must have the same set of solutions.

    We have established that Problems \ref{problem:continuous:2}, \ref{problem:discrete:2}, and \ref{problem:intermediate} all have the same set of solutions $\mfn, \kept$. Then if $\kept^*$ is a solution to Problem \ref{problem:discrete}, then $\kept^*, \mfn^{\kept^*}$ is a solution to Problem \ref{problem:discrete:2}, and thus to Problem \ref{problem:continuous:2}. This in turn means that $\mfn^{\kept^*}$ is a solution to Problem \ref{problem:continuous}. 
    
    On the other hand, if $\mfnopt$ is a solution to Problem \ref{problem:continuous}, then $\mfnopt, \dcats(\hfnopt, \mfnopt)$ is a solution to Problem \ref{problem:continuous:2} and thus to Problem \ref{problem:discrete:2}. This in turn means that $\dcats(\hfnopt, \mfnopt)$ is a solution to Problem \ref{problem:discrete}.
\end{proof}

Now, recall \Cref{prop:variancegame:full}, where $X(\kept) = \bigcup_{C \in \kept} C$ and $\mfn^\kept(\mcat) = \E[\fnopt | X(\kept) \cap \mcat]$.

\propvariancegameinformal*

\begin{proof}
Recall from \Cref{prop:dropisoptimal} that to find an optimal delegate $\mfnopt$, it is sufficient to find $\kept$ that solves
\begin{equation*}
    \argmin_{\kept \subseteq \hcats} \sum_{C \in \hcats \setminus \kept}\Prob(C)\hloss(\hfnopt, C) + \sum_{C \in \kept} \Prob(C)\mloss(\mfn^\kept, C).
\end{equation*} and take $\mfnopt = \mfn^\kept$. Thus for the remainder of this proof, we will focus on solving the problem above.

Substituting in the expressions for $\hloss$ and $\mloss$, we have the objective \begin{align*}
    &\left(\sum_{C \in \hcats \setminus \kept} \Prob(C)\sum_{\x \in C} \frac{\Prob(\x)}{\Prob(C)}(\hfnopt(C) - \fnopt(\x))^2\right) +\left(\sum_{C \in \kept}\Prob(C) \sum_{\x \in C} \frac{\Prob(\x)}{\Prob(C)}(\mfn^\kept(K(\x)) - \fnopt(\x))^2\right)\\
    &=\sum_{C \in \hcats \setminus \kept} \Prob(C)\Var(\fnopt| C) + \left(\sum_{C \in \kept} \sum_{\x \in C} \Prob(\x)(\mfn^\kept(K(\x)) - \fnopt(\x))^2\right) \tag{by defn of $\hfnopt$}\\
    &= \sum_{C \in \hcats \setminus \kept} \modvar(\fnopt | C) + \left(\sum_K \sum_{\x \in K \cap X(\kept)} \Prob(\x)(\mfn^\kept(K) - \fnopt(\x))^2\right)\\
    &= \sum_{C \in \hcats \setminus \kept} \modvar(\fnopt | C)+ \sum_K \Prob(K\cap X(\kept)) \sum_{\x \in K \cap X(\kept)} \frac{\Prob(\x)}{\Prob(K \cap X(\kept))} (\mfn^\kept(K) - \fnopt(\x))^2\\
    &= \sum_{C \in \hcats \setminus \kept} \modvar(\fnopt | C) + \sum_K \Prob(K \cap X(\kept)) \Var(\fnopt | K \cap X(\kept))\tag{by defn of $\mfn^\kept$}\\
    &= \sum_{C \in \hcats \setminus \kept} \modvar(\fnopt | C) + \sum_K \modvar(\fnopt | K \cap X(\kept))
\end{align*} and we are done.
\end{proof}

\section{Tractability proofs}\label{app:tractproofs}

In this section, we provide proofs for the theoretical results in Section \ref{sec:results}.

First, we show that any efficient algorithm for finding an optimal delegate when all features are observed by either the human or the machine can be used to find an optimal delegate in general settings. Thus in all subsequent results, we will assume that $\idxset{H} \cup \idxset{M} = [d]$.

\begin{lem}\label{lem:extra}
    Let $\x_{H\cup M}$ denote $\x$ restricted to the features in $\idxset{H}\cup\idxset{M}$. Given a ground truth optimal action function $\fnopt$, for $\x \in \hcat \cap \mcat$ define \[\bar \fn(\x_{H\cup M}) = \sum_{\z\in \hcat \cap \mcat} \frac{\Prob(\z)}{\Prob(C \cap K)}\fnopt(\z).\] Then a machine function $\mfnopt$ is an optimal delegate for ground truth function $\fnopt$ if and only if $\mfnopt$ is an optimal delegate for ground truth function $\bar f$.
\end{lem}

Intuitively, if neither agent can observe some features, this results in an additional loss penalty which is the same 

\begin{proof} A human or machine agent $A$ with action function $f_A$ and categories $\{S\}$ ($\{C\}$ if $A$ is the human or $\{K\}$ if $A$ is the machine) will have expected loss in human category $C$ in a delegation setting with ground truth optimal action function $\fnopt$ of 
    \begin{align*}&\sum_K \sum_{\x \in C\cap K} \frac{\Prob(\x)}{\Prob(C)} (f_A(S(\x)) - \fnopt(\x))^2\\
    &= \sum_K \sum_{\x \in C\cap K}\frac{\Prob(\x)}{\Prob(C)} (f_A(S(\x)) - \bar f(\x_{H \cup M}) + \bar f(\x_{H \cup M}) - \fnopt(\x))^2\\
    &= \sum_K \sum_{\x \in C\cap K} \frac{\Prob(\x)}{\Prob(C)} (f_A(S(\x)) - \bar f(\x_{H \cup M}))^2\\
    &\quad - \sum_K \sum_{\x \in C \cap K} \frac{\Prob(\x)}{\Prob(C)} 2(f_A(S(\x)) - \bar f(\x_{H \cup M}))(\bar f(\x_{H \cup M}) - \fnopt(\x))\\
    &\quad + \sum_K \sum_{\x \in C\cap K}\frac{\Prob(\x)}{\Prob(C)}(\bar f(\x_{H \cup M}) - \fnopt(\x))^2\\
    &= (\textnormal{T1}) - (\textnormal{T2}) + (\textnormal{T3})
    \end{align*}
    Term (T2) can be simplified, since $f_A$ and $\bar f$ are constant in $C \cap K$: \begin{align*}
        &\sum_K \sum_{\x \in C \cap K} \frac{\Prob(\x)}{\Prob(C)} 2(f_A(S(\x)) - \bar f(\x_{H \cup M}))(\bar f(\x_{H \cup M}) - \fnopt(\x))\\
        &=2\sum_K (f_A(S(\x)) - \bar f(\x_{H \cup M})) \sum_{\x \in C \cap K} \frac{\Prob(\x)}{\Prob(C)} (\bar f(\x_{H \cup M}) - \fnopt(\x))\\
        &= 2\sum_K (f_A(S(\x)) - \bar f(\x_{H \cup M})) \cdot 0\tag{definition of $\bar f$}\\
        &= 0
    \end{align*} Term (T1) is the expected loss in category $C$ of using function $f_A$ in category $C$ when the ground truth action function is $\bar f$, and term (T3) is a constant independent of $f_A$. Thus the team loss with ground truth function $\fnopt$ will be different from the team loss with ground truth function $\bar f$ by a constant factor of \[\sum_C \Prob(C) \sum_K \sum_{\x\in C \cap K}\frac{\Prob(\x)}{\Prob(C)}(\bar f(\x_{H \cup M}) - \fnopt(\x))^2,\] and the minimization problems have the same set of solutions.
\end{proof}

To prove Theorems \ref{thm:separable} and \ref{thm:constant}, we will need the following useful Lemma. Given a delegation setting $\idxset{H}, \idxset{M}, \Prob, \fnopt$, if $\idxset{S} = \idxset{H} \cap \idxset{M}$, let $\x_S$ be value of $\x$ restricted to the features in $\idxset{S}$. There are $s \leq h$ values that $\x_S$ may take: label these $\x_S^{(t)}$ for $1 \leq t \leq s.$ Let $\idxset{N} = [d]\setminus \idxset{S}$ be the set of features that are not shared by the human and machine, and let $\x_N$ be the value of $\x$ restricted to these features.

For $1 \leq t \leq s$, define \[g_t^*(\x_N) = \fnopt((\x_N, \x_S^{(t)})) \quad\text{and}\quad Q_t(\x) = P(\x |\x_S = \x_S^{(t)}) = \frac{\Prob(\x)}{\Prob(\mathcal{C}_t)}.\]

\begin{lem}\label{lem:shared}
    Suppose we have a delegation setting $\idxset{H}, \idxset{M}, \Prob, \fnopt$ where $\idxset{H} \cap \idxset{M} = \idxset{S} \neq \emptyset$, and an efficient algorithm to find an optimal delegate for delegation settings $\idxset{H}', \idxset{M}', Q, g^*$ where the optimal action $g^* \in \mathcal{F}$ and the distribution $Q \in \Pi$ for some family of functions $\mathcal{F}$ and distributions $\Pi$, and $\idxset{H}' \subseteq \idxset{H}$, $\idxset{M}' \subseteq \idxset{M}$ satisfy $\idxset{H}' \cap \idxset{M}' = \emptyset$. If $g_t^* \in \mathcal{F},Q_t \in \Pi$ for all $t$, then there is an efficient algorithm for finding the optimal delegate. 
\end{lem}

In other words, if we have a tractable set of instances when the human and machine share no features, then when the human and machine do share features we may break up $\fnopt$ into $s$ independent sub-problems based on the value of the shared features $\x_S$. If the subproblems are also in the tractable set, we can apply the efficient algorithm on each subproblem, and glue together the resulting optimal delegates to obtain the optimal delegate. For Theorem \ref{thm:separable}, $\mathcal{F}$ will be the set of separable functions and $\Pi$ will be distributions where $\Prob(\x_{ij}) = p_i q_j$ for some $p_i, q_j$. For Theorem \ref{thm:constant}, $\mathcal{F}$ will be the set of all functions and $\Pi$ will be the set of all distributions over $\x$.

\begin{proof}
    First, note that all $\x$ in a human category have the same value for the human features $\x_H$, and thus the same value for the shared features $\x_S$. We may partition the set of human categories based on the value of $\x_S$ for $\x \in C$: if $\x \in C$ has $\x_S = \x_S^{(t)}$, then $C \in \mathcal{C}_t$.

    Since the machine can see the features $\x_S$, the machine can choose a different action function $\mfn[t]$ for each $\x_S^{(t)}$. 
    The optimal delegate design problem is
    \begin{align*}
        &\argmin_{\mfn} \sum_{C \subseteq \hcats} \Prob(C)\min\{\hloss(\hfnopt, C), \mloss(\mfn, C)\}\\
        &= \argmin_{\mfn[1], ... \mfn[s]} \sum_{t = 1}^s \sum_{C \subseteq \hcats \cap \mathcal{C}_t} \Prob(C)\min\{\hloss(\hfnopt, C), \mloss(\mfn[t], C)\}\\
        &= \argmin_{\mfn[1], ... \mfn[s]} \sum_{t = 1}^s \sum_{C \subseteq \hcats \cap \mathcal{C}_t} \min\left\{\sum_{\x \in C}\Prob(\x)(\hfnopt(C) - \fnopt(\x))^2,\sum_{\x\in C}\Prob(\x) (\mfn[t](K(\x)) - \fnopt(\x))^2\right\}\\
        &= \argmin_{\mfn[1], ... \mfn[s]} \sum_{t = 1}^s \sum_{C \subseteq \hcats \cap \mathcal{C}_t}\min\left\{\sum_{\x \in C}Q_t(\x)\Prob(\mathcal{C}_t)(\hfnopt(C) - \fnopt(\x))^2,\sum_{\x\in C}Q_t(\x)\Prob(\mathcal{C}_t) (\mfn[t](K(\x)) - \fnopt(\x))^2\right\}\\
        &= \argmin_{\mfn[1], ... \mfn[s]} \sum_{t = 1}^s \Prob(\mathcal{C}_t)\sum_{C \subseteq \hcats \cap \mathcal{C}_t} \min\left\{\sum_{\x \in C} Q_t(\x)(\hfnopt(C) - \fnopt(\x))^2,\sum_{\x\in C}Q_t(\x) (\mfn[t](K(\x)) - \fnopt(\x))^2\right\}
    \end{align*}
    Then \begin{align*}
        \mfnopt[t] &= \argmin_{\mfn[t]} \sum_{C \subseteq \hcats \cap \mathcal{C}_t} Q_t(\x) \min\{\hloss(\hfnopt, C), \mloss(\mfn[t], C)\}
    \end{align*} 
    Furthermore, for $C \in \mathcal{C}_t$, $\hfnopt(C)$ is the same for $\fnopt$ and $g^*_t$. This means that computing $\mfnopt[t]$ is  
    exactly the problem for finding the optimal delegate on sub-problem \[(\idxset{H} \setminus \idxset{M}), (\idxset{M} \setminus \idxset{H}), Q_t, g_t^*\] for which we have an efficient algorithm by assumption. Furthermore, there are only $s \leq h \leq n$ sub-problems.
\end{proof}

\subsection{Proof of \Cref{thm:separable}}

Recall \Cref{thm:separable}, which showed that for functions that are separable into functions of the human features and machine features respectively, we may find an optimal delegate in polynomial time in the size of $\fnopt$, which is $n$.

\thmseparable*

\begin{proof} We divide this proof into two parts: first, we prove this result under the assumption that each state occurs with uniform probability. Then, we reduce the case of general independent probabilities to the first case.

Define $u_i := u(C_i)$, $w_j := w(K_j)$.

\paragraph{Part 1.} We first assume that $\Prob(\x) = 1/n$ for each $\x$. Since by \Cref{lem:extra} and \Cref{lem:shared} we may assume without loss of generality that the human and machine partition the set of all features, so, $n = h m$. In this case, $\Prob(C) = 1/h$ for each category $C$, and $\Prob(K \cap X(\kept)) = |\kept|/n$.

Moreover, in this case, we can write $v_{ij} := \fnopt(\x_{ij}) = u_i + w_j$. We previously did not specify how we indexed the human categories $\hcats$; we may now index $\hcats$ so that $u_1 \leq u_2 \leq ... \leq u_{h}$; performing this indexing has polynomial time complexity $O(h \log h)= O(n \log n)$ since $h = O(n)$ in the worst case.

By \Cref{prop:variancegame:full}, we need to find some set $R$ solving \begin{align*}
&\min_R \sum_{i \notin R} \modvar(v_{ij} | i) + \sum_j \modvar (v_{ij}| i\in R, j)\\
&= \min_R \frac{1}{h}\sum_{i \notin R} \Var(u_i + w_j | j \in [m], i) + \frac{|R|}{n}\sum_j \Var (u_i + w_j | i\in R, j) \tag{by definition of $\Delta^2$}\\
&= \min_R \frac{1}{h}\sum_{i \notin R} \Var(w_j |  j \in [m], i) + \frac{|R|}{n}\sum_j \Var (u_i | i\in R, j)
\end{align*}

since variance is does not depend on additive constants. In turn, this is equivalent to
\begin{align*}
&\min_R \frac{h - |R|}{h}\Var(w_j | j \in [m]) + \frac{|R|}{n}m\Var (u_i | i\in R)\\
&= \min_R \left(1 - \frac{|R|}{h}\right)\Var(w) + \frac{|R|}{n}\frac{n}{h}\Var (u_i | i\in R) \tag{$\Var(w) := \Var(w_j |  j \in [m])$}\\
&= \min_k \min_{R: |R| = k}\left(1 - \frac{k}{h}\right)\Var(w) + \frac{k}{h} \Var (u_i | i\in R)\\
&= \min_k \left(1 - \frac{k}{h}\right)\Var(w) + \frac{k}{h}\min_{R: |R| = k} \Var (u_i | i\in R)
\end{align*}
where the last step was suggested by \citet{stackoverflowminvar}.
Let \[R_k \in \argmin_{R\subseteq [h]: |R| = k} \Var(u_i | i \in R).\] Then our objective is \[\min_k \left[\left(1 - \frac{k}{h}\right)\cdot \Var(w) + \left(\frac{k}{h}\right)\cdot \Var(u_i | i \in R_k)\right].\]

We may now outline an algorithm to compute the optimal retained rows $R^*$. First, compute $R_k$ for each $k$. Iterate over each $1 \leq k \leq h$, and find the $k^*$ that minimizes $\left(1 - \frac{k}{h}\right)\Var(w) + \left(\frac{k}{h}\right) \Var(u_i | i \in R_k)$. Then $R_{k^*}$ minimizes the  objective. Take $\kept = \{C_i : i \in R_{k^*}\}$. By \Cref{prop:variancegame:full}, $\mfn^{\kept}$ is an optimal machine design. It remains for us to show that each of these steps may be performed efficiently.

Given $\{R_k\}_{k=1}^h$, finding $k^*$ simply requires computing the objective for each $1 \leq k \leq h$, which has total time complexity $O(h)$. It therefore only remains to find a polynomial-time algorithm to compute $R_k$. This problem is equivalent to finding the minimum variance subset of $u$ of size $k$, and setting $R_k$ to be the indices corresponding to that subset. This can be done by observing that the minimum variance subset must be \textit{contiguous}. This was previously conjectured  \citep{stackoverflowminvar}; we prove this formally in \Cref{lem:minvar} using a proof technique similar that suggested by \citet{stackoverflowminvar}. We can then compute the variance of each of the $h-k + 1$ contiguous subsets of $\{u_i\}$ in time $O(h^2)$, for  a total time of $O(h^3)= O(n^3)$.

\paragraph{Part 2.} Now, we consider the case of general separable probabilities.

By definition of consistency, if the human and machine share any features, the independent sub-problems from Lemma \ref{lem:shared} satisfy $\Prob(\x_{ij})= p_i\cdot q_j$ for each $i,j$, so for the remainder of the proof we will assume that the human and machine partition the set of features and that $\Prob(\x_{ij}) = p_i \cdot q_j$. By \Cref{cor:variancegame:simple}, we must solve \[\min_R \sum_{i \notin R} \modvar(v_{ij} | i) + \sum_j \modvar(v_{ij} | i \in R, j).\] Expanding the definition of $\modvar$, we see that \begin{align}
    &\sum_{i \notin R} \modvar_{p \times q}(v_{ij} | i) + \sum_j \modvar_{p \times q}(v_{ij} | i \in R, j)\nonumber\\
    &= \sum_{i \notin R} p_i \Var_{p \times q}(u_i + w_j | i) + \sum_j \left(\sum_{i' \in R} p_{i'} q_j \right)\Var_{p \times q}(u_i + w_j | i \in R, j)\nonumber\\
    &= \sum_{i \notin R} p_i \Var_{p \times q}(w_j | i) + \sum_j \left(\sum_{i' \in R} p_{i'} q_j \right)\Var_{p \times q}(u_i | i \in R, j)\nonumber\\
    &= \sum_{i \notin R} p_i \Var_q(w) + \left(\sum_{i \in R} p_{i}\right) \sum_j  q_j \Var_p(u_i | i \in R)\nonumber\\
    &=\left(\sum_{i \notin R} p_i\right) \Var_q(w) + \left(\sum_{i \in R} p_i\right) \Var_p(u_i | i \in R)\nonumber\\
    &=\left(\sum_{i \notin R} p_i\right) \Var_q(w) + \left(\sum_{i \in R} p_i\right) \Var_p(u_i | i \in R).\label{prob:minvariancerationals}
\end{align}
Recall that since the probabilities $p_i$ have polynomial precision, there is some $T\in \mathbb{N}$ such that for all $i$, $p_i = \frac{t_i}{T}$ for some $t_i \in \mathbb{N}$, and that $T \in O(\textnormal{poly}(n))$.

Let $S$ be a multiset of values $\{s_r\}_{r=1}^T$, which is the set induced when each $u_i$ is repeated $t_i$ times. Suppose each $s_r$ occurs with equal probability. Consider the problem \begin{equation}\label{prob:minvariancenaturals}
    \min_{R\subseteq [T]} \left(1 - \frac{|R|}{T}\right)\Var(w) + \frac{|R|}{T} \Var(s_r | r \in R).
\end{equation}
This problem produces a set $R^*$ in polynomial time, by the proof as in Part 1. 

In \Cref{lem:separable:fullsetsincluded}, we find that for any solution $R^*$ of Problem \ref{prob:minvariancenaturals}, if some $s_r \in R^*$ has value $u_i$, then all $s_{r'}$ with value $u_i$  are also in $R^*$. This means that in Problem \ref{prob:minvariancenaturals} we may actually minimize over $R\subseteq [T]$ in which values $u_i$ that are included in the solution appear $t_i$ times. In this case, it is easy to see that the objective in Problem \ref{prob:minvariancenaturals} is precisely the objective from the original problem expressed in Equation \ref{prob:minvariancerationals}.
\end{proof}

To prove Lemmas \ref{lem:minvar} and \ref{lem:separable:fullsetsincluded}, we will rely heavily on a result relating the variance of a set to that of its subsets, shown by \citet{oneill14}.\footnote{\citet{oneill14} shows this result in the case of sample variance, but it is easy to extend to true (``population'') variance as we show here.}
\begin{result}[Corollary of Result 1 of \citet{oneill14}]\label{oneillresult} For $S \subset \mathbb{R}$, define $\mu(S) = \frac{1}{|S|} \sum_{x \in S} x,$ and $\sigma^2(S) = \frac{1}{|S|}\sum_{x \in S}(x - \mu(S))^2.$ Let $S_1, S_2$ partition $S$: $S_1 \cap S_2 = \emptyset, S_1 \cup S_2 = S$. Define $D^2 = \frac{|S_1||S_2|}{|S|}(\mu(S_1) - \mu(S_2))^2$. Then \[|S|\mu(S) = |S_1|\mu(S_1) + |S_2| \mu(S_2)\] and \[|S| \sigma^2(S) = |S_1|\sigma^2(S_1) + |S_2|\sigma^2(S_2) + D^2.\]
\end{result}

\begin{lem}\label{lem:minvar} For each $k$, there is some $1 \leq t \leq h-k+1$ such that \[\{t,t+1,..., t+k-1\}\in \argmin_{R \subseteq [h]: |R| = k} \Var(u_i | i\in R).\]
\end{lem}

\begin{proof}[Proof of \Cref{lem:minvar}] 
    Fix $k$, and suppose that there is no contiguous minimum variance subset of size $k$. Let \[R_k \in \argmin_{R\subseteq [h]: |R| = k} \Var(u_i | i \in R).\] Let $i = \min R_k, i' = \max R_k$. We may assume without loss of generality that there is no $u_j = u_i$ for $j > i$ and $j \notin R_k$; otherwise replace $R_k$ with $R_k\cup \{j\} \setminus \{i\} $, which will also be a minimum variance subset. Similarly we may assume that there is no $u_j = u_{i'}$ for $j < i'$ and $j \notin R_k$.

    We use the shorthand $\mu(R) = \mu(u_r | r \in R)$, $\Var(R) = \Var(u_r | r \in R)$.

    By assumption, $R_k$ is not contiguous, that is, there is some $i < j < i'$ such that $j \notin R_k$ and $u_i < u_j < u_{i'}$. Suppose that $u_j \leq \mu(R_k)$. Let $R_k^{0} := R_k \setminus \{i\}$. Since $u_i \leq u_s$ for all $s \in I_k$, $\mu(R_k^0) \geq \mu(R_k)$.

Now, let $R_k' := R_k^{0} \cup j$; essentially $R_k'$ is the result of replacing $\hcati{i}$ with $C_j$ in $R_k$. Now by Result \ref{oneillresult} \cite{oneill14} for the true variance rather than the sample variance, we see that \[k\Var(R_k) = 1 + (k-1)\Var(R_k^0) +\frac{k-1}{k}\left(u_i - \mu(R_k^0)\right)^2\] and \[k\Var(R_k') = 1 + (k-1)\Var(R_k^0) + \frac{k-1}{k}\left(u_j - \mu(R_k^0)\right)^2\] Thus \begin{align*}
    \Var(R_k') - \Var(R_k)&= \frac{k-1}{k^2}\left(\left(u_j - \mu(R_k^0)\right)^2 - \left(u_i - \mu(R_k^0)\right)^2\right)\\
    &= \frac{k-1}{k^2}\left(\left(u_j - \mu(R_k^0)\right) + \left(u_i - \mu(R_k^0)\right)\right)\left(\left(u_j - \mu(R_k^0)\right) - \left(u_i - \mu(R_k^0)\right)\right)\\
    &= \frac{k-1}{k^2}\left(u_j - \mu(R_k^0) + u_i - \mu(R_k^0)\right)\left(u_j - u_i\right)
    \end{align*}
    Since $u_j < u_i$ and $u_i, u_j \leq  \mu(R_k) \leq \mu(R_k^0)$, $\Var(R_k') - \Var(R_k) < 0$. Thus $R_k$ is not the minimum variance subset, and we have a contradiction.

    If $u_j > \mu(R_k)$, we can let $R_k^0 := R_k \setminus \{C_{i'}\}$ and $R_k' := R_k^0 \cup \{C_j\}$, which by a symmetric argument again yields $\Var(R_k') - \Var(R_k) < 0$.
\end{proof}

\begin{lem}\label{lem:separable:fullsetsincluded}
    Let $R^*$ be a solution to Problem \ref{prob:minvariancenaturals}. If for some $i$, there are two values $s_r, s_{r'}$ such that $s_r = s_{r'} = u_i$, then $s_r \in R^* \implies s_{r'} \in R^*$.
\end{lem}
\begin{proof}
    Suppose not, that is, there are two values $s_r, s_{r'}$ such that $s_r = s_{r'} = u_i$, but $s_r \in R^*$, $s_{r'} \notin R^*$.  We again use the shorthand $\mu(R) = \mu(s_r  |r \in R), \sigma^2(R) = \sigma^2(s_r| r\in R)$.

    Let $R^+ = R^* \cup \{s_{r'}\}$, $R^- = R^* \setminus \{s_{r'}\}$. Since $R^*$ is an optimal solution, it must be the case that \begin{align*}
        \left(1 - \frac{|R^*|}{T}\right)\Var(w) + \frac{|R^*|}{T} \Var(s_r | r \in R^*)\leq \left(1 - \frac{|R^+|}{T}\right)\Var(w) + \frac{|R^+|}{T} \Var(R^+)
    \end{align*} and \begin{align*}
        \left(1 - \frac{|R^*|}{T}\right)\Var(w) + \frac{|R^*|}{T} \Var(R^*)\leq \left(1 - \frac{|R^-|}{T}\right)\Var(w) + \frac{|R^-|}{T} \Var(R^-)
    \end{align*} Since $|R^*| = |R^-| + 1 = |R^+| - 1$, rearranging we see that \begin{align*}
        \Var(w) \leq |R^+| \Var(R^+) - |R^*| \Var(R^*)
    \end{align*} and \begin{align*}
        |R^*| \Var(R^*) - |R^-| \Var(R^-) \leq \Var(w) 
    \end{align*} so that \[ |R^*| \Var( R^*) - |R^-| \Var(R^-) \leq |R^+| \Var(s | R^+) - |R^*| \Var(R^*)\]

    Then \begin{align*}
        \mu(R^+) - \mu(R^*) &\triangleq \frac{1}{|R^*|+1}\left(\sum_{r \in R^*} s_r + s_{r'}\right) - \frac{1}{|R^*|}\left(\sum_{r \in R^*} s_r\right)\\
        &=\frac{1}{|R^*|(|R^*|+1)}\left(|R^*|\sum_{r \in R^*} s_r + |R^*| s_{r'} -(|R^*|+1)\sum_{r \in R^*} s_r\right)\\
        &=\frac{1}{|R^*|(|R^*|+1)}\left(|R^*| s_{r'} - \sum_{r \in R^*} s_r\right)
    \end{align*} and by an analogous argument, \begin{align*}
        \mu(R^*) - \mu(R^-) =\frac{1}{|R^*|(|R^*|-1)}\left(|R^*| s_{r'} - \sum_{r \in R^*} s_r\right)
    \end{align*}
    Thus the ratio \begin{align*}
        \frac{(\mu(R^+) - \mu(R^*))^2}{(\mu(R^*) - \mu(R^-))^2} &= \left(\frac{|R^*| - 1}{|R^*| + 1}\right)^2.
    \end{align*}

    Again applying Result 1 \cite{oneill14}, we see that \begin{align*}
        |R^+| \Var(R^+)- |R^*| \Var(R^*) =  (D^2)^+ \quad \textnormal{and} \quad
        |R^*| \Var(R^*) -  |R^-| \Var(R^-) = (D^2)^-,
    \end{align*} where \[(D^2)^+ := \frac{|R^*|}{|R^*| + 1}\left(\mu(R^+) - \mu(R^*)\right)^2\quad \textnormal{and} \quad (D^2)^- := \frac{|R^*|-1}{|R^*|}\left(\mu(R^*) - \mu(R^-)\right)^2.\] 

    Then the inequality above can be written as $(D^2)^- \leq (D^2)^+$. This means that \begin{align*}
        1 \leq \frac{(D^2)^+}{(D^2)^-}
        = \left(\frac{|R^*| - 1}{|R^*| + 1}\right)\frac{(\mu(R^+) - \mu(R^*))^2}{(\mu(R^*) - \mu(R^-))^2}= \left(\frac{|R^*| - 1}{|R^*| + 1}\right)^3< 1, 
    \end{align*} since $|R^*| \geq 1$, and we have a contradiction.
\end{proof}

\subsection{Proof of \Cref{thm:constant}}

We now show that when the human or the machine can observe only a constant number of features, it is possible to efficiently find an optimal delegate; recall \Cref{thm:constant}.

\thmconstant*

\begin{proof} 
We again assume that the human and machine partition the set of all features, so that $|\idxset{H} \setminus \idxset{M}| = |\idxset{H}|$ and $|\idxset{M} \setminus \idxset{H}| = |\idxset{M}|$. If not, we can apply Lemmas \ref{lem:extra} and \ref{lem:shared} to find an efficient algorithm, since the independent subproblems in \Cref{lem:shared} will satisfy $|\idxset{H}| = O(1)$ or $|\idxset{M}| = O(1)$ respectively.

We first prove that that when $|\idxset{H}| = O(1)$ we may efficiently find an optimal delegate through a brute-force search over sets of retained human categories, and then prove that when $|\idxset{M}| = O(1)$ we may find an optimal delegate by restricting the possible sets of retained categories through finding arrangements of hyperspheres.
    \paragraph{Part 1} Suppose that $|\idxset{H}| = O(1)$.     By \Cref{prop:dropisoptimal}, it is sufficient to find $\kept \subseteq \hcats$ that minimizes $\loss(\hfnopt, \mfn^\kept)$. If $|\idxset{H}| = O(1)$, then $|\hcats| = 2^{|\idxset{H}|} = O(1)$, and there are $2^{O(1)} = O(1)$ subsets of $\hcats$. Moreover, we may compute $\mfn^\kept$ and then $\loss(\hfnopt, \mfn^\kept)$ in $O(n)$ for each $\kept$. Thus, we may compare the loss of $\mfn^\kept$ for each $\kept \subseteq \hcats$ in time $O(n^2)$ and select $\kept$ with the minimum loss.

    \paragraph{Part 2} Suppose that $|\idxset{M}| = O(1)$. As before, we assume $\Prob(C) > 0$ for all $C$; these categories do not affect the loss and we can arbitrarily assign them to the human.

    Since we assume that $\idxset{H}$ and $\idxset{M}$ partition the set of all features, there is a single state $\x_{ij}$ in $C_i \cap K_j$ for each human category $C_i$ and machine category $K_j$. For any machine function $\mfn$, we may define the vector $\mathbf{y} \in \mathbb{R}^m$ by setting $\mathbf{y}_j = \mfn(K_j)$. Define \[c_{ij} := \fnopt(\x_{ij}), \quad r_i^2 := \sum_{j} (\hfnopt(C_i) - \fnopt(\x_{ij}))^2, \quad 
 a_{ij}^2 := \frac{\Prob(\x_{ij})}{\Prob(C_i)}.\]

    We know that a human with action function $\hfnopt$ will delegate to a machine with action function $\mfn$ in category $C_i$ if and only if \[\sum_{j = 1}^m \frac{\Prob(\x_{ij})}{\Prob(C_i)}(\mfn(K_j) - \fnopt(\x_{ij})^2 < \sum_{j=1}^m (\hfnopt(C_i) - \fnopt(\x_{ij}))^2,\] assuming that the human breaks ties by not delegating. We can rewrite this condition as \[\sum_{j=1}^m a_{ij}^2 (\mathbf{y}_j - c_{ij})^2 < r_i^2,\] which is the equation for the interior of an ellipsoid where $\mathbf{y} \in \mathbb{R}^m$. This means that for any machine $\mfn$, $\mfn$ is in the region formed by the intersection of the ellipsoids corresponding to the categories where $\mfn$ is adopted, $\dcats(\hfnopt, \mfn)$. 
    
    If $\fnopt(\x), \Prob(\x)$ are rational for all $\x$ -- which is a reasonable assumption if we hope to optimize on a computer with floating point precision -- then the ellipsoid $g_i(\mathbf{y}) := \sum_{j=1}^m a_{ij}^2 (\mathbf{y}_j - c_{ij}^2) - r_i^2$ is a polynomial with maximum degree $2m$ in $\mathbb{R}^m$, and $g_i(\y) = 0$ is the ellipsoid corresponding to category $C_i$. The regions of intersection of these $h$ ellipsoids are known as ``arrangements'' in computational geometry. There are only $O(h^m)$ such regions, and these regions may be found in time $O(h^m)$ \citep{chazelle1991}. The first step of this algorithm will be to find these regions, which takes $O(h^m)$.

    In \Cref{prop:dropisoptimal}, we showed that to find an optimal delegate, it is sufficient to find a set of categories $\kept^*$ such that $\mfn^\kept$ minimizes the team loss, and $\kept^*$ satisfies $\kept^* = \dcats(\hfnopt, \mfnopt)$ for some optimal delegate $\mfnopt$. This means that $C_i\in \kept^*$ if and only if $\mfnopt$ is in ellipsoid $i$. This in turn implies that in searching over different sets of retained categories $\kept$, we can consider only subsets of categories whose corresponding set of ellipsoids has a non-empty intersection. Now, instead of $2^h$ possible options for $\kept$, we are only searching over $O(h^m)$ different subsets $\kept$. 
    
    We can also compute the loss $\mfn^\kept$ of a given $\kept$ in $O(n)$, so we may find an optimal delegate in $O(nh^m)$. Since $h = O(n)$ and $m = O(1)$, this means we may find an optimal delegate in time polynomial in $n$.
\end{proof}

\subsection{Proof of \Cref{prop:perfectteams}}

Recall \Cref{prop:perfectteams}, which states that it is always possible to efficiently determine whether a delegate with perfect team performance exists; we prove this below.

\propperfect*

\begin{proof} By \Cref{prop:dropisoptimal}, we can achieve zero expected hybrid loss if and only if there is some set of human categories $\kept \subseteq \hcats$ such that \[\loss(\hfnopt, \mfn^\kept)= \sum_{C \in \hcats\setminus \kept} \Prob(C) \hloss(\hfnopt, C) + \sum_{C \in \kept}\Prob(C)\mloss(\mfn^\kept, C) = 0.\]
Since $\hloss, \mloss$ are always non-negative, this means that we can achieve zero hybrid loss if and only if there is some set of categories $\kept$ for the machine to retain such that for all $C \in \kept$ with $\Prob(C) > 0$, it is the case that $\mloss(\mfn^\kept, C) = 0$ and for all yielded categories $C \in \hcats\setminus \kept=: \mathcal{Y}$ with $\Prob(C) > 0$, it is the case that $\hloss(\hfnopt, C) = 0$.

For a human category $C \in \hcats$, \begin{align*}
    \hloss(\hfnopt, C) = 0 &\iff \sum_{\x \in C}\frac{\Prob(\x)}{\Prob(C)}(\hfnopt(\hcatvec{x}) - \fnopt(\x))^2 = 0\\
    &\iff \sum_{\x \in C}\frac{\Prob(\x)}{\Prob(C)}\left(\sum_{\z \in C} \frac{\Prob(\z)}{\Prob(C)}\fnopt(\z) - \fnopt(\x)\right)^2 = 0\\
    &\iff \fnopt(\x) = \sum_{\z \in C} \frac{\Prob(\z)}{\Prob(C)} \fnopt(\z) \text{ for all } \x \in C, \Prob(\x) > 0\\
    &\iff \fnopt(\x) = \fnopt(\z) \text{ for all } \x,\z \in C, \Prob(\x), \Prob(\z) > 0.
\end{align*}
For each category, we can check this condition simply by iterating linearly over each state in the category. We can then collect the categories for which this condition holds into \[\kept^0 := \{C \in \hcats: \fnopt(\x) \neq \fnopt(\z) \text{ for some } \x, \z \in C\},\] which is the minimal set of human categories that the machine must retain to achieve non-zero loss. Then if $\kept$ achieving zero hybrid loss exists, the optimal design must retain $\kept\supseteq \kept^0$. We can compute $\kept^0$ by making a linear pass through all the states. This has time complexity $O(n)$. 

For a human category $C \in \hcats$, \begin{align*}
\mloss(\mfn^\kept, C) = 0
    \iff& \sum_{\x\in C}\frac{\Prob(\x)}{\Prob(C)}\left(\mfn^\kept(K(\x)) - \fnopt(\x)\right)^2 = 0\\
    \iff& \mfn^\kept(K(\x))= \fnopt(\x) \text{ for all } \x \in C, \Prob(\x) > 0\\
    \iff& \fnopt(\x) = \fnopt(\z) \text{ for all } \x \in C, \z \in C'\\
    & \quad \text{ where } C' \in \kept, K(\z) = K(\x), \text{ and } \Prob(\x), \Prob(\z) > 0
\end{align*}
where the last equivalence follows by definition of $\mfn^\kept$.

Suppose that retaining $\kept^0$ does not lead to zero hybrid loss. Then it must be the case that for some $C \in \kept^0$, $\mloss(\mfn^\kept, C) > 0$ and thus $\fnopt(\x) \neq \fnopt(\z)$ for some $\x \in C$, $\z \in C' \in \kept^0$ with $K(\x) = K(\z)$ and $\Prob(\x), \Prob(\z) > 0$.

If instead we retain $\kept' \supset \kept^0$, we may use the same $\x, \z$ show that $\mloss(\mfn^{\kept'}, C) > 0$. Thus there is a machine with zero hybrid error if and only if retaining $\kept^0$ results in zero hybrid error. We may then simply check whether $\loss(\hfnopt, \mfn^{\kept^0}) = 0$, which takes time $O(n)$.

We spend $O(n)$ time to compute $\kept^0$ and $O(n)$ time to compute $ \loss(\hfnopt, \mfn^{\kept^0})$, so the total time needed to determine whether $\loss(\hfnopt, \mfnopt) = 0$ is $O(n)$.
\end{proof}

\subsection{Proof of \Cref{thm:np}}\label{app:np}

Finally, we prove that finding an optimal delegate is NP-hard in general. Recall \Cref{thm:np}.

\thmnp*

\begin{proof}
In \Cref{prop:variancegame:full}, we showed that when the human and machine features partition the set of all features it is necessary to solve the problem \textsc{VarianceAssignment} in order to find an optimal delegate. We will show that \textsc{VarianceAssignment} is NP-hard. In particular, we restrict our attention to the case where each state occurs with uniform probability, that is $\Prob(\x) = 1/n$ for all $\x$. Since the problem is hard in this special case, it is hard in general. First, consider the problem \textsc{MaxRegularClique}. 

\begin{quote}
    \textsc{MaxRegularClique}. Let $G = (V,E)$ be a regular graph. A clique is a set of nodes $S\subseteq V$ such that for each pair of nodes $u, v \in S$, $(u, v) \in E$. Find a clique with maximum size $|S|$.
\end{quote}
\cite{brandes2016} define the problem \textsc{RegularClique}, which determines whether a regular graph $G$ contains a clique of size $k$, show that it is NP-hard. We can solve \textsc{RegularClique} by solving \textsc{MaxRegularClique} and checking whether the solution has size $\geq k$; thus \textsc{MaxRegularClique} is also NP-hard.

We now define the intermediate problem of densest subgraph discovery in the presence of possibly negative weights and a regularity condition on each node.

\begin{quote}
    \textsc{NegRegularDSD.} Let $G = (V,E,w)$ be an undirected graph with weighted edges $w: E \to \mathbb{R}$. Suppose that for each node $v$, \[\sum_{(u, v) \in E} |w(u,v)| = 1.\] For a subset of nodes $S \subseteq V$, let $E(S)$ be the edges in the induced subgraph, and define \[w(S) = \sum_{(u,v) \in E(S)} w(u,v).\] Find the subset of nodes $S$ that maximizes the density of the induced subgraph \[d(S)=\frac{w(S)}{|S|},\] where $d(S) = 0$ when $S = \emptyset$.
\end{quote}

Without the regularity condition that the absolute sum of a node's edge weights is equal to 1, this is the NP-hard problem \textsc{NegDSD} introduced in \cite{tsourkakis2019}. It is also folk knowledge that \textsc{NegDSD} may be proved via a reduction from \textsc{MaxClique}. 

We now show that \textsc{NegRegularDSD} is NP-hard via a reduction from \textsc{MaxRegularClique}. Let the $d$-regular graph $G = (V,E)$ be an instance of \textsc{MaxRegularClique}, we construct a complete graph $G' = (V', E',w)$ where $V' = V$ and $E'$ is the set of all pairs of nodes. Let \[w(u,v) = \frac{1}{1 + (n-d)}\cdot\begin{cases}
        \frac{1}{d}, & (u,v) \in E,\\
        -1, & (u, v) \notin E
    \end{cases}.\] 

    Constructing $G'$ can be completed in polynomial time $O(|V|^2)$. Now \[\sum_{(u,v) \in E'} |w(u,v)| = d\cdot\frac{1}{d}\cdot\frac{1}{1 + (n-d)} + (n-d)\cdot 1 \cdot\frac{1}{1 + n-d} = 1.\]

    First note that the solution $S$ to \textsc{NegRegularDSD} will always have non-negative density, since we could always pick the empty set. Now, suppose there is a solution $S$ to \textsc{NegRegularDSD} that is not a clique in $G$. Then there is some pair $u, v \in S$ such that $(u,v) \notin E$. Let $E(N)$ be the edges in $G'$ induced by a set of nodes $N$. Then, \begin{align*}
        d(S) &\triangleq \frac{\sum_{(s,t) \in E(S)}}{|S|}\\
        &= \frac{\sum_{(s, t) \in E(S), t \neq v} w(s,t)+ \sum_{(s, v) \in E(S), s\neq u} w(s,v) + w(u,v)}{|S|}\\
        &\leq \frac{\sum_{(s, t) \in E(S), t \neq v} w(s,t)+ \sum_{(s, v) \in E(S) s\neq u} w(s,v) -\frac{1}{1+n-d}}{|S|-1}\\
        &\leq \frac{\sum_{(s, t) \in E(S) t \neq v} w(s,t)+ d\cdot \frac{1}{d} \cdot \frac{1}{1 + (n-d)} -\frac{1}{1+n-d}}{|S|-1}\\
        &= \frac{\sum_{(s, t) \in E(S\setminus \{v\})}w(s,t)}{|S|-1}\\
        &\triangleq d(S\setminus \{v\}),
    \end{align*} so $d(S)$ cannot have been a solution of \textsc{NegRegularDSD}. Thus \textsc{NegRegularDSD} will produce a solution which is a clique in $G$. For a subset $S\subseteq V$ that is a clique in $G$, the density is \begin{align*}
        d(S) &= \frac{\sum_{(u,v) \in E(S)} w(u,v)}{|S|}
        = \frac{{|S| \choose 2} \frac{1}{d}\cdot \frac{1}{1 + (n-d)}}{|S|}
        = \frac{1}{d}\cdot \frac{1}{1+(n-d)}\cdot \frac{|S| - 1}{2}\propto |S|-1
    \end{align*} Thus \textsc{NegRegularDSD} will select the clique of maximum size.
    
Finally, we reduce \textsc{NegRegularDSD} to \textsc{VarianceAssignment}. Let $G = (V,E,w)$ be an instance of \textsc{NegRegularDSD}. Construct an instance $A$ of \textsc{VarianceAssignment} as follows.

\begin{itemize}[noitemsep]
    \item First, create a matrix $A^0$: for each node $v \in V$ create a row and for each edge $e \in E$ create a column.
    
    \item For each edge $e_j = (v_i,v_k)$:
    \begin{itemize}
        \item If $w(v_i,v_k) > 0$, let $a^0_{ij} = a^0_{kj} = \sqrt{|w(v_i,v_k)|/2}$.
        \item If $w(v_i,v_k) \leq 0$, let $a^0_{ij} = -a^0_{kj} = \sqrt{|w(v_i,v_k)|/2}$. 
    \end{itemize}
    \item Set all other entries of $A^0$ to zero.

    \item Now, create a matrix $A$ as follows: for each column $a_j$ of $A^0$, add both $a_j$ and $-a_j$ to $A$.
    
\end{itemize}

Constructing $A$ takes time $O(2|E||V|) = O(|V|^3)$. 

Now for each row $i$ \[\E(a_{ij} | i) = \frac{1}{m}\sum_j a_{ij} = 0 \quad \textnormal{and} \quad \|a_i\|_2^2 = \sum_j a_{ij}^2 = 2\sum_{k: (v_i,v_k)\in E} |w(v_i,v_k)|/2 = 1.\]

    Let $R$ be a subset of rows. We have that $A \in \mathbb{R}^h \times \mathbb{R}^m$ for $m := 2|E|$, $h := |V|$. Define $n := hm$ as before. Then the objective of \textsc{VarianceAssignment} is to minimize \[\frac{1}{h}\sum_{i \notin R}\sigma^2(a_{ij} | i) + \frac{|R|}{n}\sum_j \sigma^2(a_{ij} | i \in R, j).\]

    We may expand this to \begin{align*}
        &\frac{1}{h}\sum_{i \notin R}\E(a_{ij}^2 | i) - \E(a_{ij} | i)^2 + \frac{|R|}{n}\sum_j \E(a_{ij}^2 | i \in R, j) - \E(a_{ij} | i \in R, j)^2\\
        &= \frac{1}{h}\sum_{i \notin R}\E(a_{ij}^2 | i) +\frac{|R|}{n}\sum_j \E(a_{ij}^2 | i \in R, j) - \left(\frac{1}{h}\sum_{i \notin R}\E(a_{ij} | i)^2 + \frac{|R|}{n}\sum_j \E(a_{ij} | i \in R, j)^2\right)
    \end{align*}

    The first two terms can be simplified as \begin{align*}
        \frac{1}{h}\sum_{i \notin R}\E(a_{ij}^2 | i) +\frac{|R|}{n}\sum_j \E(a_{ij}^2 | i \in R, j)
        &= \frac{1}{h} \sum_{i \notin R} \frac{1}{m} \sum_j a_{ij}^2 + \frac{|R|}{n} \sum_j \frac{1}{|R|} \sum_{i \in R} a_{ij}^2\\
        &= \frac{1}{n} \sum_{i \notin R}\sum_j a_{ij}^2 + \frac{1}{n}\sum_{i \in R}\sum_j a_{ij}^2\\
        &= \frac{1}{n}\sum_{i,j} a_{ij}^2
    \end{align*} This term is constant in $R$, so the problem of \textsc{VarianceAssignment} is merely the problem of minimizing the second two terms, or \textit{maximizing} \begin{align*}
        \frac{1}{h}\sum_{i \notin R}\E(a_{ij} | i)^2 + \frac{|R|}{n}\sum_j \E(a_{ij} | i \in R, j)^2 &= \frac{|R|}{n}\sum_j \E(a_{ij} | i \in R)^2\tag{$\E(a_{ij} | i) = 0$}\\
        &= |R|\sum_j\left(\frac{1}{|R|}\sum_{i \in R} a_{ij}\right)^2\\
&= \frac{1}{n|R|}\sum_j \sum_{i \in R} a_{ij}^2 +  \frac{2}{n|R|}\sum_j\sum_{i < k: i, k \in R} a_{ij} a_{kj}\\
&= \frac{1}{n|R|}\sum_{i \in R}\|a_i\|^2 + \frac{2}{n|R|}\sum_{i < k: i, k \in R} \sum_j a_{ij} a_{kj}\\
&= \frac{1}{n}\left(1 + \frac{2}{|R|}\sum_{i < k: i, k \in R} \sum_j a_{ij} a_{kj}\right),\tag{$\|a_i\|^2 = 1$}
\end{align*} so to solve \textsc{VarianceAssignment} it is necessary to find $R$ maximizing $\frac{1}{|R|}\sum_{i < k: i, k \in R} \sum_j a_{ij} a_{kj}.$
Let $S(R) = \{v_i: i\in R\}$. Recalling the construction of $A$ and $A^0$, \begin{align*}\frac{1}{|R|}\sum_{i < k: i, k \in R} \sum_j a_{ij} a_{kj} &= \frac{1}{|R|}\sum_{i < k: i, k \in R} \sum_{j=1}^{|E|} 2 a_{ij}^0 a_{kj}^0\\
&= \frac{1}{|R|}\sum_{(v_i, v_k) \in E(S(R))}  2 \cdot \sqrt{|w(v_i, v_k)|/2} \cdot \textnormal{sgn}(w(v_i, v_k)) \sqrt{|w(v_i, v_k)|/2}\\
&= \frac{\sum_{(v_i, v_k) \in E(S(R))} w(v_i, v_k)}{|R|}\\
&\triangleq d(S(R)).
\end{align*} Therefore, to maximize $d(S)$ it is sufficient to minimize the \textsc{VarianceAssignment} objective, so \textsc{VarianceAssignment} is NP-hard.
\end{proof}

\end{document}